\documentclass[11pt,a4paper]{article}
\usepackage{float}
\usepackage{jheppub}

\usepackage[T1]{fontenc}
\usepackage[utf8]{inputenc}
\usepackage{lmodern}
\usepackage{microtype}
\usepackage{amsfonts,amsthm,mathtools,bm}
\usepackage{booktabs,longtable,array,tabularx,multirow}
\usepackage{enumitem}
\usepackage{xcolor}
\usepackage{tikz}
\usetikzlibrary{arrows.meta,backgrounds,calc,decorations.pathreplacing,fit,matrix,positioning,shapes.geometric,shapes.misc}
\usepackage[most]{tcolorbox}
\usepackage{listings}
\usepackage{xurl}
\usepackage[nameinlink,noabbrev]{cleveref}

\definecolor{bcjblue}{HTML}{245B78}
\definecolor{bcjcyan}{HTML}{2A9D8F}
\definecolor{bcjgold}{HTML}{E9C46A}
\definecolor{bcjorange}{HTML}{F4A261}
\definecolor{bcjred}{HTML}{D95D39}
\definecolor{bcjgray}{HTML}{5E6472}
\definecolor{bcjlight}{HTML}{F3F6F8}
\definecolor{codegreen}{HTML}{287A4B}

\hypersetup{
  pdftitle={An Adaptive Inverse-Problem Framework for One-Loop Five-Gluon BCJ Numerators},
  pdfauthor={Lin Mai, Yaobo Zhang}
}

\setlist[itemize]{topsep=3pt,itemsep=2pt,parsep=1pt}
\setlist[enumerate]{topsep=3pt,itemsep=2pt,parsep=1pt}
\allowdisplaybreaks

\newcommand{\A}{\mathcal A}
\newcommand{\M}{\mathcal M}
\newcommand{\cQ}{\mathcal Q}
\newcommand{\I}{\mathcal I}

\newcommand{\F}{\mathbb F}
\newcommand{\Q}{\mathbb Q}
\newcommand{\Z}{\mathbb Z}

\newcommand{\cC}{\mathcal C}

\newcommand{\cS}{\mathcal S}
\newcommand{\cT}{\mathcal T}

\newcommand{\trfive}{\operatorname{tr}_5}
\newcommand{\rank}{\operatorname{rank}}
\newcommand{\nullity}{\operatorname{nullity}}
\newcommand{\Span}{\operatorname{span}}
\newcommand{\im}{\operatorname{im}}
\newcommand{\Res}{\operatorname{Res}}
\newcommand{\id}{\mathrm{id}}

\newcommand{\hashid}[1]{\nolinkurl{#1}}

\newtheorem{definition}{Definition}[section]
\newtheorem{proposition}[definition]{Proposition}
\newtheorem{theorem}[definition]{Theorem}

\newtheorem{corollary}[definition]{Corollary}
\newtheorem{remark}[definition]{Remark}

\newtcolorbox{resultbox}[1][]{
  enhanced,breakable,colback=bcjlight,colframe=bcjblue,
  boxrule=0.8pt,arc=1.5mm,left=2mm,right=2mm,top=1.5mm,bottom=1.5mm,
  title={#1},fonttitle=\bfseries
}
\newtcolorbox{boundarybox}[1][]{
  enhanced,breakable,colback=bcjorange!8,colframe=bcjorange!80!black,
  boxrule=0.8pt,arc=1.5mm,left=2mm,right=2mm,top=1.5mm,bottom=1.5mm,
  title={#1},fonttitle=\bfseries
}
\newtcolorbox{evidencebox}[1][]{
  enhanced,breakable,colback=bcjcyan!5,colframe=bcjcyan!80!black,
  boxrule=0.8pt,arc=1.5mm,left=2mm,right=2mm,top=1.5mm,bottom=1.5mm,
  title={#1},fonttitle=\bfseries
}

\lstdefinestyle{pseudocode}{
  basicstyle=\ttfamily\small,
  keywordstyle=\bfseries\color{bcjblue},
  commentstyle=\itshape\color{codegreen},
  stringstyle=\color{bcjred},
  showstringspaces=false,
  frame=single,
  rulecolor=\color{bcjgray!40},
  backgroundcolor=\color{bcjlight},
  breaklines=true,
  columns=fullflexible,
  keepspaces=true,
  xleftmargin=1em,
  framexleftmargin=0.5em
}
\tikzset{
  flow/.style={draw=bcjblue,thick,rounded corners=2pt,fill=bcjlight,align=center,minimum height=9mm},
  evidence/.style={draw=bcjcyan!90!black,thick,rounded corners=2pt,fill=bcjcyan!8,align=center,minimum height=9mm},
  pending/.style={draw=bcjorange!90!black,thick,dashed,rounded corners=2pt,fill=bcjorange!8,align=center,minimum height=9mm},
  danger/.style={draw=bcjred,thick,rounded corners=2pt,fill=bcjred!8,align=center,minimum height=9mm},
  arrow/.style={-{Stealth[length=2.5mm]},thick,draw=bcjgray},
  advice/.style={-{Stealth[length=2.5mm]},thick,dashed,draw=bcjred}
}

\crefname{equation}{eq.}{eqs.}
\Crefname{equation}{Equation}{Equations}

\title{An adaptive inverse-problem framework for one-loop five-gluon BCJ numerators}

\author[a]{Lin Mai,}
\author[a]{and Yaobo Zhang}

\affiliation[a]{School of Physics, Ningxia University,\\
Yinchuan 750021, China}

\emailAdd{12025130767@stu.nxu.edu.cn}
\emailAdd{yaobozhang@nxu.edu.cn}

\abstract{
An inverse problem comprises a matrix equation together with its unknown space,
physical data, equivalence relation, and validation tests.  We
formulate Bern--Carrasco--Johansson (BCJ) numerator construction as an exact adaptive
inverse problem.  A fixed scientific specification determines the theory, graph
conventions, coefficient field, locality and power counting, cut data, observable
equivalence, and independent checks.  Each finite working specification compiles to
\(\mathcal P_\sigma=(A,b;\cS;\cT)\), where \(Ax=b\) reconstructs numerator coefficients,
\(\cS\) classifies the solution fiber, and \(\cT\) tests it on held-out information.
Left-null obstructions identify candidate numerator-basis directions needed for
consistency, while the
action of candidate measurements on the right kernel identifies informative new cut
equations.  We illustrate these steps by hand at four points and apply them to one-loop
five-gluon pure Yang--Mills theory.

After kinematic and graph-symmetry reduction, the candidate numerator basis contains
\(1127\) independent
coordinates.  The combined maximal, box, triple, and double cuts have rank \(920\),
giving a \(207\)-dimensional affine solution fiber.  Exact reconstruction determines a
particular solution and the complete ordered kernel.  The specified \(R_{12345}\)
color-ring readout \(\cS\)
annihilates every kernel direction, so the full fiber represents one observable class.
A published forward-limit numerator lies in this fiber, and fresh cuts, independent
integral reductions, and helicity-amplitude benchmarks validate the result.  Explicit
search rules and agent interfaces can propose revisions.  Deterministic compilation
and exact evaluation assess each proposal.
}


\compress

\begin{document}

\hypersetup{pageanchor=false}
\maketitle
\hypersetup{pageanchor=true}
\flushbottom

\section{Introduction}
\label{sec:introduction}

Inverse problems are usually written as \(Ax=b\).  Gaussian elimination starts after
the unknown space, the equations, and the meaning of equality have been fixed.
In amplitude construction these ingredients are part of the problem.  The construction
specifies the local tensors, graph identities, on-shell data, representation
equivalence, and independent checks before solving the resulting system exactly.

This paper develops that point of view for Bern--Carrasco--Johansson (BCJ) numerators.
A gauge-theory amplitude can be organized as a sum over oriented cubic graphs.  Each
term contains a color factor, a product of propagators, and a kinematic graph numerator.
A BCJ representation assigns the numerators jointly so that their antisymmetry and
Jacobi relations mirror those of the color factors \cite{Bern:2008qj}.  Replacing the
color factors by a second dual numerator collection then gives the double-copy
representation of a gravity integrand \cite{Bern:2010ue}.  The amplitudes and their
generalized cuts are physical input, while the graph-by-graph numerator assignment is
representation data to be reconstructed.  This is an inverse problem with a generally
non-unique solution.

\subsection{A fixed scientific specification and an adaptive working problem}

We distinguish the fixed scientific specification from the finite working problem used
at a given stage.  The specification, denoted by \(\Omega\), fixes the theory and
external states, graph orientations and momentum routings, exact coefficient field,
locality and power-counting rules, cut oracle, observable-equivalence prescription, and
independent validation tests.  It also defines the allowed ways to enlarge the ansatz
or the cut set.  These choices remain fixed throughout the construction.

At stage \(t\), the working specification \(\sigma_t\) records the ordered basis
columns, graph-descendant compiler, cut rows, exact samples, and algebraic reduction
data.  Compiling \(\sigma_t\) according to \(\Omega\) produces the complete finite
problem
\begin{equation}
  \boxed{\mathcal P_\sigma=(A,b;\cS;\cT).}
  \label{eq:chartered-inverse-problem}
\end{equation}
Here \(A\) evaluates the candidate numerator coefficients on the selected constraints,
and \(b\) contains the corresponding physical targets.  The map \(\cS\) determines
whether two solutions represent the same amputated observable under the routing,
Lehmann--Symanzik--Zimmermann (LSZ) projection, which removes external-leg
contributions, and dimensional-regularization rules fixed by \(\Omega\).  The suite
\(\cT\) contains
held-out checks, including fresh cut samples, independent reductions, and amplitude
benchmarks.  Thus solving \(Ax=b\) is one part of solving \(\mathcal P_\sigma\).

When the system is consistent, its full coefficient fiber is
\begin{equation}
  \mathcal F_b
  =\{x:Ax=b\}
  =x_{\rm p}+\ker A.
  \label{eq:intro-affine-fiber}
\end{equation}
The construction reports this entire affine set, its image in the quotient defined by
\(\cS\), and the outcome of \(\cT\).  Once the fiber has been determined, a convenient
representative may be selected for a particular calculation.

The two null spaces provide the information needed to revise an incomplete
specification.  When the target lies outside the span of the current columns, exact
elimination returns a left-null witness
\begin{equation}
  \lambda^T A=0,
  \qquad \lambda^T b\ne0.
  \label{eq:intro-left-obstruction}
\end{equation}
A necessary condition for a candidate column \(c\) to repair this obstruction is
\(\lambda^Tc\ne0\).  If the system is consistent and \(K\) is an ordered basis of its
right kernel, a candidate measurement row \(r\) supplies new information precisely when
\(rK\ne0\).  The first test guides the enlargement of the representable numerator
space; the second designs measurements that resolve directions left free by the
current cuts.  Directions that remain free are subsequently classified by \(\cS\).
The resulting sequence is
\begin{equation}
  \text{compile }\mathcal P_{\sigma_t}
  \longrightarrow \text{exact diagnostic}
  \longrightarrow \text{candidate revision}
  \longrightarrow \mathcal P_{\sigma_{t+1}}.
  \label{eq:constraint-first-pattern}
\end{equation}

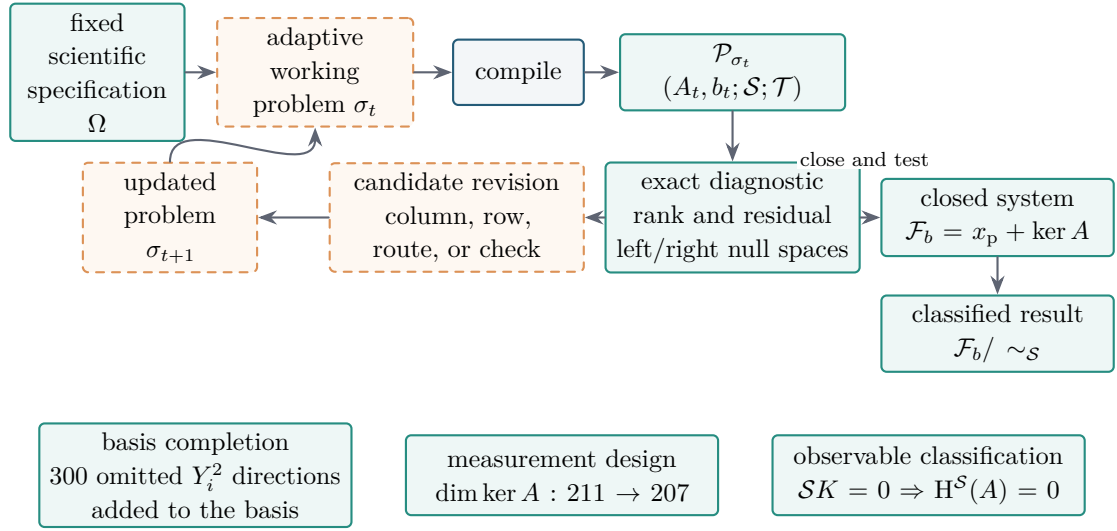
\begin{figure}[t]
\centering
\begin{tikzpicture}[x=1cm,y=1cm,scale=.96,transform shape,
  line cap=round,line join=round,
  every node/.append style={font=\small}]
  \node[evidence,text width=21mm,inner xsep=1.5mm] (charter) at (1.25,2.55)
    {fixed scientific\\specification \(\Omega\)};
  \node[pending,text width=24mm,inner xsep=1.5mm] (spec) at (4.25,2.55)
    {adaptive working\\problem \(\sigma_t\)};
  \node[flow,text width=15mm,inner xsep=1.5mm] (compile) at (7.05,2.55)
    {compile};
  \node[evidence,text width=28mm,inner xsep=1.5mm] (problem) at (10.00,2.55)
    {\(\mathcal P_{\sigma_t}\)\\\((A_t,b_t;\cS;\cT)\)};
  \node[evidence,text width=32mm,inner xsep=1.5mm] (diagnostic) at (10.00,.55)
    {exact diagnostic\\rank and residual\\left/right null spaces};

  \node[pending,text width=32mm,inner xsep=1.5mm] (proposal) at (6.20,.55)
    {candidate revision\\column, row, route, or check};
  \node[pending,text width=21mm,inner xsep=1.5mm] (next) at (2.25,.55)
    {updated problem\\\(\sigma_{t+1}\)};
  \node[evidence,text width=29mm,inner xsep=1.5mm] (fiber) at (13.65,.55)
    {closed system\\\(\mathcal F_b=x_{\rm p}+\ker A\)};
  \node[evidence,text width=29mm,inner xsep=1.5mm] (quotient) at (13.65,-1.05)
    {classified result\\\(\mathcal F_b/\!\sim_{\cS}\)};

  \draw[arrow] (charter.east)--(spec.west);
  \draw[arrow] (spec.east)--(compile.west);
  \draw[arrow] (compile.east)--(problem.west);
  \draw[arrow] (problem.south)--(diagnostic.north);
  \draw[arrow] (diagnostic.west)--(proposal.east);
  \draw[arrow] (proposal.west)--(next.east);
  \draw[arrow] (next.north) to[out=90,in=-90] (spec.south);
  \draw[arrow] (diagnostic.east)--(fiber.west);
  \node[font=\scriptsize,fill=white,inner sep=1pt] at (11.82,1.36)
    {close and test};
  \draw[arrow] (fiber.south)--(quotient.north);
  \node[evidence,text width=40mm,minimum height=11mm,inner xsep=1.5mm]
    at (2.60,-3.00)
    {basis completion\\\(300\) omitted \(Y_i^2\) directions\\added to the basis};
  \node[evidence,text width=40mm,minimum height=11mm,inner xsep=1.5mm]
    at (7.65,-3.00)
    {measurement design\\\(\dim\ker A:211\to207\)};
  \node[evidence,text width=40mm,minimum height=11mm,inner xsep=1.5mm]
    at (12.70,-3.00)
    {observable classification\\\(\cS K=0\Rightarrow \mathrm H^{\cS}(A)=0\)};
\end{tikzpicture}
\caption{Adaptive construction of the finite inverse problem.  The scientific
specification remains fixed while the working basis and measurement set evolve.
Exact diagnostics identify useful column and row revisions and determine when the
affine fiber is ready for observable classification and independent validation.  The
lower strip records the three diagnostic transitions used in the five-point case
study.}
\label{fig:claim-pipeline}
\end{figure}

The recurring terms can be read as follows.  An \emph{ansatz} is the finite ordered list
of allowed tensor structures with unknown coefficients.  A \emph{master numerator} is
a generating graph numerator from which graph descendants follow by symmetries, Jacobi
relations, and momentum shifts.  A \emph{cut row} evaluates these descendants on an
on-shell channel whose target is a product of tree amplitudes.  A left-null vector
certifies a target combination outside the image of the current ansatz.  A right-kernel
vector is a coefficient change invisible to the current rows.  The map \(\cS\) decides
whether such a change affects the stated observable, and \(\cT\) checks the closed
construction on held-out information.

\Cref{sec:four-point-walkthrough} works through these objects for three four-point tree
numerators and a \(3\times3\) matrix.  It derives the compatibility condition, the
one-dimensional generalized-gauge freedom, and the affine solution by hand.  The
five-point calculation applies the same logical structure to a larger basis and
measurement set.

\subsection{From inverse CHY construction to BCJ numerators}

The Cachazo--He--Yuan (CHY) representation provides a useful precursor for this form of
exact construction.  Companion matrices turn sums over solutions of the scattering
equations into finite linear-algebra operations \cite{Huang:2015yka}.  Cross-ratio
identities reduce higher-order poles to graph-level integration rules
\cite{Zhou:2017mfj}.  These methods handle worldsheet expressions, pole data, and
physical factorization channels through finite linear-algebra operations in place of an
explicit solution of every intermediate analytic equation.

Earlier work by Zhang and collaborators developed two parts of this
picture.  Labelled-tree integrands were reorganized into effective Feynman diagrams with
manifest compatible pole sets, and pole-picking cross ratios were extended from the
bi-adjoint scalar integrand to general simple-pole CHY integrands
\cite{Feng:2020opo}.  The regular and singular solution sectors of the scattering
equations were then followed through factorization, soft, and forward kinematic limits
\cite{Feng:2021iik}.  These studies proceed from an integrand or scattering equation to
its channel behavior.

The inverse CHY construction reverses that direction.  Its input is a prescribed set of
factorization channels and pole degrees.  Its unknowns are signed integer edge
multiplicities of candidate CHY graphs.  Exact constraint propagation on the subset
lattice and a final mixed-integer solve determine the compatible integrands
\cite{Li:2026inverseCHY}.  Residual channel violations select discrete edge assignments
for revision.  This gives the order of operations used here: state the representation
grammar and channel data, compile a finite exact
problem, use its residual structure to revise the candidate space, and report all
solutions admitted by the fixed rules.

BCJ numerator construction supplies a different realization of this inverse logic.
Constructive tree-level algorithms produce relabeling-symmetric BCJ numerators
\cite{Fu:2014pya}, and the scattering-equation formalism can make
color--kinematics duality manifest term by term \cite{Bjerrum-Bohr:2016axv}.  At one
loop, Q-cut representations assemble on-shell tree data and reduce channel by channel
to conventional unitarity-cut integrands \cite{Huang:2015cwh}; one-loop CHY
transmutation operators likewise factorize into tree-level operators on a unitarity
cut \cite{Zhou:2021kzv}.  In the present construction the unknowns are rational
coefficients of local kinematic tensors, and the target equations are generalized
unitarity cuts.  The graph grammar also imposes antisymmetry and Jacobi relations.  The
linear system is therefore different from the integer CHY system, while the
constraint-first organization is the same.

\subsection{The five-point case study and the four contributions}

The case study reconstructs a local parity-even one-loop five-gluon master numerator in
pure Yang--Mills theory.  Twenty-four master words generate pentagon, box, triangle, and
lower-topology descendants under fixed routing.  Locality and power counting define the
candidate columns.  Maximal, box, triple, and double generalized cuts define the rows.
The internal polarization-state dimension \(D_s\) is retained symbolically, and all
matrix ranks, residuals, and coordinate relations are reconstructed exactly.

The paper makes four contributions.
\begin{enumerate}[label=\arabic*.]
  \item It formulates amplitude-representation construction as a finite exact inverse
  problem.  The fixed scientific specification \(\Omega\) determines the physical and
  algebraic rules, while the adaptive working specification \(\sigma_t\) compiles to
  \(\mathcal P_\sigma=(A,b;\cS;\cT)\), combining reconstruction, observable
  classification, and independent validation.

  \item It gives two exact diagnostics for adapting the finite problem.  Left-null data
  identify ansatz extensions that can restore consistency, while the restriction of
  candidate measurements to the right kernel measures the information added by new cut
  rows.  In the five-point implementation, an exact quotient calculation finds \(300\)
  omitted \(Y_i^2\) directions in the initial square-free basis.  Adding the
  repeated-contact sector \(Y_i^2N_{ii}\) completes this basis.  The kernel-visibility
  calculation then selects four independent double-cut constraints, raising the cut
  rank from \(916\) to \(920\) and reducing the nullity from \(211\) to \(207\).

  \item It reconstructs the complete \(207\)-dimensional numerator fiber and shows that
  every kernel direction is annihilated by the specified observable map \(\cS\).  It
  locates a published forward-limit numerator inside this fiber, with \(15\) active
  kernel directions, and tests the result on fresh cuts, independent integral
  reductions, and helicity-amplitude benchmarks.

  \item It provides a common proposal interface for physicists, explicit search rules,
  and stateful agents.  Scientific skills encode typed revisions, while the compiler,
  harness, and evaluator construct and assess each stage.  The five-point calculation
  uses the deterministic basis and cut-selection rules described in the preceding item.
\end{enumerate}

The local ansatz contains \(17{,}824\) tensor monomials before algebraic reduction.
On-shell conditions, transversality, momentum conservation, and the stated tensor
relations give
\begin{equation}
  17{,}824
  =10{,}724_{\rm quotient}+7{,}100_{\rm algebraic\ kernel},
  \qquad
  10{,}724=5124+4700+900.
  \label{eq:raw-quotient-balance}
\end{equation}
The three summands are the contact-degree-zero, contact-degree-one, and
contact-degree-two layers.  Graph rotations, reflections, and orientation signs then
leave \(1127\) independent coefficient coordinates.  The cumulative cut ranks are
\begin{equation}
\begin{array}{c|cccc}
\text{cut layer} & \text{maximal} & \text{box} & \text{triple} & \text{double}\\ \hline
\text{cumulative rank} & 547&812&916&920\\
\text{rank increment}  & 547&265&104&4
\end{array}
\label{eq:rank-staircase-intro}
\end{equation}
and hence
\begin{equation}
  \dim\ker A=1127-920=207.
  \label{eq:nullity-207-intro}
\end{equation}

Over \(\Q(D_s)\), the exact coefficient fiber is
\begin{equation}
  \boxed{
  x(D_s,\boldsymbol\beta)
  =x^{(0)}+D_s x^{(1)}+K\boldsymbol\beta(D_s),
  \quad K=(k_1\ \cdots\ k_{207}),
  \quad\boldsymbol\beta(D_s)\in\Q(D_s)^{207}.}
  \label{eq:main-affine-family}
\end{equation}
A normalized reduced-row-echelon convention selects the particular solution
\(x^{(0)}+D_sx^{(1)}\) and fixes the order and normalization of the kernel vectors.
Exact rational reconstruction determines all \(209\) rational vectors: the two
particular-solution components and the \(207\) kernel directions.

The observable-equivalence map for the \(R_{12345}\) color-ring projection uses a fixed
common routing, exact coefficient cancellations, an LSZ projector, and stated
scaleless-integral identities.  Its action
on the complete ordered kernel is
\begin{equation}
  \cS K=0,
  \qquad
  \cS\!\left[x_{\rm p}+K\boldsymbol\beta\right]=\cS[x_{\rm p}].
  \label{eq:stage-s-kernel-zero-intro}
\end{equation}
Thus the coefficient fiber is \(207\)-dimensional, while its image in the specified
observable quotient is one class.  The kernel remains useful for selecting sparse,
local, double-copy-adapted, or integration-adapted representatives.

A published forward-limit numerator \cite{Cao:2025forward} is located by the exact
coordinate relation
\begin{equation}
  x_{\rm FL}(D_s)
  =x^{(0)}+D_s x^{(1)}
   +K\bigl(\boldsymbol\alpha_0+D_s\boldsymbol\alpha_1\bigr).
  \label{eq:known-membership-intro}
\end{equation}
Fifteen kernel directions are active, and nine coefficients depend on \(D_s\).  This
membership statement is a comparison in the common numerator coefficient space.  At
the later integration stage, the raw fixed-routing and forward-limit prescription maps
have a nonzero difference in one stated master-integral coefficient.  Together with
\cref{eq:stage-s-kernel-zero-intro}, this shows that a separately specified surface,
external-field, or scheme completion is required when the two prescriptions are compared.

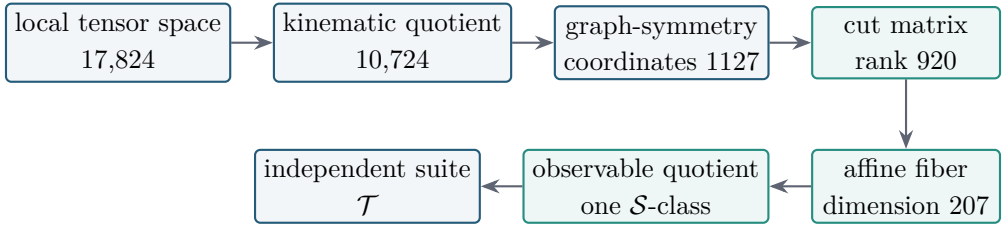
\begin{figure}[t]
\centering
\begin{tikzpicture}[node distance=7mm and 6mm,scale=.92,transform shape]
  \node[flow,minimum width=27mm] (raw) {local tensor space\\\(17{,}824\)};
  \node[flow,right=of raw,minimum width=27mm] (quot) {kinematic quotient\\\(10{,}724\)};
  \node[flow,right=of quot,minimum width=27mm] (orbit) {graph-symmetry\\coordinates \(1127\)};
  \node[evidence,right=of orbit,minimum width=27mm] (cuts) {cut matrix\\rank \(920\)};
  \node[evidence,below=10mm of cuts,minimum width=27mm] (aff) {affine fiber\\dimension \(207\)};
  \node[evidence,left=of aff,minimum width=29mm] (stage) {observable quotient\\one \(\cS\)-class};
  \node[flow,left=of stage,minimum width=29mm] (tests) {independent suite \\ \(\cT\)};
  \draw[arrow] (raw)--(quot);
  \draw[arrow] (quot)--(orbit);
  \draw[arrow] (orbit)--(cuts);
  \draw[arrow] (cuts)--(aff);
  \draw[arrow] (aff)--(stage);
  \draw[arrow] (stage)--(tests);
\end{tikzpicture}
\caption{The numerical funnel for the one-loop five-gluon case study.  Algebraic and
graph reductions define the columns, generalized cuts determine the affine fiber, the
map \(\cS\) classifies its observable content, and \(\cT\) supplies independent
checks.}
\label{fig:fivepoint-number-funnel}
\end{figure}

The validation suite uses fresh primes and random seeds to test the particular solution
and every kernel direction on independently generated cut equations.  Separate
scalar-source routes give the same eleven master-integral coefficients.  The all-plus,
single-minus, adjacent-MHV, and nonadjacent-MHV amplitudes reproduce the expected
rational terms, poles, scheme conversion, analytic branches, and finite-kinematics values.
Together these calculations form the held-out validation suite \(\cT\).

The adaptive loop has a common proposal interface for physicists, explicit search
rules, and language-model agents.  At each stage, the proposer receives the exact
ranks, residuals, null spaces, and visibility data and returns a candidate column, row,
routing, or validation update.  Scientific skills encode these typed transformations.
The compiler and evaluator construct the next finite problem and record its diagnostic.
In the five-point calculation, fixed rules complete the repeated-contact basis and
select the informative double-cut block.  Agent-generated proposals pass through the
same compiler and exact tests.  The conceptual organization is
described in \cref{sec:agent-assisted-framework}; implementation and record formats are
given in \cref{sec:agent-harness,app:pseudocode,app:evidence-ledger}.

The rest of the paper follows the construction in this order.
\Cref{sec:four-point-walkthrough} derives the small example by hand.
\Cref{sec:bcj-fivepoint-framework} gives the loop-level BCJ framework and defines the
five-point problem.  \Cref{sec:agent-assisted-framework} describes proposal interfaces
for the adaptive loop, including language-model agents, scientific skills, harnesses,
and exact evaluators.
\Cref{sec:ansatz-compiler} constructs the tensor space and graph descendants, while
\cref{sec:cuts-and-visibility} builds the generalized-cut rows and explains their
topology-dependent visibility.  \Cref{sec:solution-equivalence} reconstructs the affine
fiber and applies \(\cS\), and \cref{sec:validation} performs the independent integral
and amplitude checks.  \Cref{sec:discussion} summarizes the scope and subsequent uses
of the reconstructed family.  Exact arithmetic, pseudocode, harness details,
reproducibility information, and the evidence ledger are placed in the appendices.

\section{A complete four-point inverse problem}
\label{sec:four-point-walkthrough}

We begin with a four-gluon tree amplitude, the smallest example that contains the main
parts of the reconstruction.  It has three cubic graphs.  Its constraint matrix, left
null space, right null space, and complete affine solution can therefore be calculated
by hand.  With the unknown space, measurements, signs, and equality conventions fixed,
this example gives a complete inverse solve.  The one-loop five-point problem uses the
same linear algebra with \(1127\) coefficients; later sections construct the finite
basis and measurements that define that larger system.

The overall coupling and common phase are omitted.  The displayed signs use one fixed
orientation of the three cubic graphs.  Other consistent orientation choices give the
same relations after the corresponding signs are changed together.

\subsection{Cubic graphs and the inverse problem}

Let all four massless momenta be outgoing and define
\begin{equation}
  s=(k_1+k_2)^2,\qquad
  t=(k_2+k_3)^2,\qquad
  u=(k_1+k_3)^2,
  \qquad s+t+u=0.
  \label{eq:fourpoint-mandelstam}
\end{equation}
With totally antisymmetric structure constants, choose the oriented color factors
\begin{equation}
  c_s=f^{a_1a_2b}f^{ba_3a_4},\qquad
  c_t=f^{a_2a_3b}f^{ba_1a_4},\qquad
  c_u=f^{a_3a_1b}f^{ba_2a_4}.
  \label{eq:fourpoint-oriented-colors}
\end{equation}
There are three cubic tree graphs, distinguished by the momentum flowing through their
single internal edge.  Each graph contribution has three conceptually different pieces:
\begin{center}
\small
\begin{tabularx}{0.88\textwidth}{@{}lXXX@{}}
\toprule
channel & propagator & color factor & kinematic numerator \\
\midrule
\(s\) & \(1/s\) & \(c_s\) & \(n_s(k,\varepsilon)\) \\
\(t\) & \(1/t\) & \(c_t\) & \(n_t(k,\varepsilon)\) \\
\(u\) & \(1/u\) & \(c_u\) & \(n_u(k,\varepsilon)\) \\
\bottomrule
\end{tabularx}
\end{center}
The propagator is fixed by graph topology.  The color factor is a product of structure
constants.  After these two pieces are separated, the remaining factor contains
momenta, polarizations, and the kinematic content of the vertices; this factor is the
graph numerator.
Thus
\begin{equation}
  \A_4^{\rm tree}
  =\frac{c_s n_s}{s}+\frac{c_t n_t}{t}+\frac{c_u n_u}{u}.
  \label{eq:fourpoint-cubic-amplitude}
\end{equation}

Yang--Mills theory also has a four-gluon contact vertex.  In the cubic-graph
representation, its contribution is distributed among the three numerators.  Within a
fixed color component, schematically, for any
\(\eta_s+\eta_t+\eta_u=1\), a contact expression \(V_4\) may be written as
\begin{equation}
  V_4
  =\frac{s\,\eta_s V_4}{s}
   +\frac{t\,\eta_t V_4}{t}
   +\frac{u\,\eta_u V_4}{u}.
  \label{eq:fourpoint-contact-distribution}
\end{equation}
The inverse propagator in a numerator cancels the displayed propagator.  Different
choices of the \(\eta\)'s move contact information between graphs, changing the
individual graph numerators while leaving their assembled amplitude invariant.

With the chosen orientations, the color Jacobi identity is
\begin{equation}
  c_s+c_t+c_u=0.
  \label{eq:fourpoint-color-jacobi}
\end{equation}
Reversing two legs at one cubic vertex gives an oriented graph \(\bar g\) with
\(c_{\bar g}=-c_g\).  A BCJ collection assigns
\(n_{\bar g}=-n_g\), so the product \(c_gn_g\) is independent of which of the two
equivalent vertex drawings is used.  The two drawings therefore represent a single
graph unknown.

A BCJ representation requires the three kinematic weights to obey the parallel identity
\begin{equation}
  n_s+n_t+n_u=0.
  \label{eq:fourpoint-kinematic-jacobi}
\end{equation}
Thus a BCJ numerator is one member of a collection indexed by oriented cubic graphs.
Antisymmetry and Jacobi relations apply to the collection as a whole.  This graph
numerator is distinct from the numerator of a single rational function obtained after
combining the full amplitude.

Suppose the two coefficients in a two-element color basis are known from Feynman rules
or an on-shell formula.  We denote them by the two color-ordered amplitudes below.
Before color--kinematics duality is imposed, they serve as two coordinates of the color
decomposition.  The consistency calculation will derive the relation between them.  In
the orientation used here, eliminating
\(c_t=-c_s-c_u\) from \cref{eq:fourpoint-cubic-amplitude} gives
\begin{equation}
  \A_4^{\rm tree}
  =c_s A(1,2,3,4)-c_u A(1,3,2,4),
  \label{eq:fourpoint-color-basis}
\end{equation}
where matching the two color coefficients yields
\begin{align}
  A(1,2,3,4)&=\frac{n_s}{s}-\frac{n_t}{t},
  \label{eq:fourpoint-partial-first}\\
  A(1,3,2,4)&=\frac{n_t}{t}-\frac{n_u}{u}.
  \label{eq:fourpoint-partial-second}
\end{align}
The amplitudes on the left are the physical input data.  The three graph numerators on
the right are the unknown representation data.  The inverse problem is to find three
kinematic functions that reproduce the known amplitudes and obey the kinematic Jacobi
identity.  Their organization follows the color algebra and allows the resulting
collection to be used in a double copy.

\subsection{The constraint matrix and its complete solution}

Define
\begin{equation}
  n=\begin{pmatrix}n_s\\ n_t\\ n_u\end{pmatrix},
  \qquad
  b=\begin{pmatrix}
       A(1,2,3,4)\\ A(1,3,2,4)\\ 0
     \end{pmatrix}.
  \label{eq:fourpoint-vectors}
\end{equation}
Equations \eqref{eq:fourpoint-partial-first},
\eqref{eq:fourpoint-partial-second}, and
\eqref{eq:fourpoint-kinematic-jacobi} are the single linear system
\begin{equation}
  \underbrace{
  \begin{pmatrix}
    1/s & -1/t & 0\\
    0   & 1/t  & -1/u\\
    1   & 1    & 1
  \end{pmatrix}}_{M}
  \begin{pmatrix}n_s\\n_t\\n_u\end{pmatrix}
  =
  \begin{pmatrix}
    A(1,2,3,4)\\A(1,3,2,4)\\0
  \end{pmatrix}.
  \label{eq:fourpoint-constraint-matrix}
\end{equation}
The first two rows are \emph{measurement rows}: they demand agreement with physical
amplitude data.  The third is a \emph{structural row}: it demands color--kinematics
duality.  Together they form the linear system \(Ax=b\).

The matrix is singular.  Its consistency condition produces a relation between the
input amplitudes, and its right null vector produces the freedom in the numerator
assignment.  Both follow by direct elimination.

For compactness write
\begin{equation}
  A_s=A(1,2,3,4),\qquad A_u=A(1,3,2,4).
  \label{eq:fourpoint-amplitude-shorthand}
\end{equation}
The first measurement row gives
\begin{equation}
  n_s=sA_s+\frac{s}{t}n_t,
  \label{eq:fourpoint-solve-ns}
\end{equation}
and the second gives
\begin{equation}
  n_u=\frac{u}{t}n_t-uA_u.
  \label{eq:fourpoint-solve-nu}
\end{equation}
Substitution into the Jacobi equation produces
\begin{align}
  0=n_s+n_t+n_u
   &=sA_s-uA_u
     +n_t\left(\frac{s}{t}+1+\frac{u}{t}\right)\notag\\
   &=sA_s-uA_u,
  \label{eq:fourpoint-compatibility-by-hand}
\end{align}
because \(s+t+u=0\).  The \(n_t\) term cancels, leaving the compatibility condition
\begin{equation}
  \boxed{s\,A(1,2,3,4)=u\,A(1,3,2,4).}
  \label{eq:fourpoint-bcj-amplitude-relation}
\end{equation}
This is the four-point BCJ amplitude relation in the stated orientation.  The same
duality that constrains graph numerators therefore predicts a relation among the
color-ordered amplitudes used as target data.

When \cref{eq:fourpoint-bcj-amplitude-relation} holds, set
\(n_t=t\alpha\).  Equations \eqref{eq:fourpoint-solve-ns} and
\eqref{eq:fourpoint-solve-nu} give the complete solution
\begin{equation}
  \boxed{
  \begin{pmatrix}n_s\\n_t\\n_u\end{pmatrix}
  =
  \begin{pmatrix}sA_s\\0\\-sA_s\end{pmatrix}
  +\alpha
  \begin{pmatrix}s\\t\\u\end{pmatrix}.}
  \label{eq:fourpoint-affine-family}
\end{equation}
The first vector is one convenient representative, obtained by choosing \(n_t=0\).
The second vector spans a one-dimensional family of equally valid representations.
The complete answer is therefore an affine line: one particular solution plus the
one-dimensional right null space.

\begin{remark}[A numerical check]
Take a generic algebraic point
\begin{equation}
  s=2,\qquad t=-3,\qquad u=1,
  \qquad A_s=3,\qquad A_u=6.
  \label{eq:fourpoint-numeric-input}
\end{equation}
The compatibility condition is \(2\times3=1\times6\).  At \(\alpha=0\), the
numerators are \((6,0,-6)\).  At \(\alpha=1\), they are \((8,-3,-5)\).  Direct
substitution gives, for both choices,
\begin{equation}
  \frac{n_s}{2}-\frac{n_t}{-3}=3,
  \qquad
  \frac{n_t}{-3}-n_u=6,
  \qquad
  n_s+n_t+n_u=0.
  \label{eq:fourpoint-numeric-check}
\end{equation}
Both numerator choices therefore give the same two amplitude coefficients.
\end{remark}

The compatibility calculation also has a null-space form that extends directly to
larger systems.  The row vector
\begin{equation}
  \lambda^T=\begin{pmatrix}s&-u&-1\end{pmatrix}
  \label{eq:fourpoint-left-kernel}
\end{equation}
obeys \(\lambda^T M=0\).  Multiplying the equation \(Mn=b\) from the left therefore
requires
\begin{equation}
  0=\lambda^TMn=\lambda^Tb=sA_s-uA_u.
  \label{eq:fourpoint-left-compatibility}
\end{equation}
If, for example, the last datum in \cref{eq:fourpoint-numeric-input} were mistakenly set
to \(A_u=5\), then \(\lambda^Tb=1\neq0\).  No choice of the three numerators could solve
all the equations.  The vector \(\lambda\) identifies the incompatible combination of
targets explicitly.

We call such a vector a \emph{left-null inconsistency vector}, or simply an
\emph{obstruction}.  In a larger problem it can
identify a missing ansatz term, a sign or routing mismatch, an incomplete state sum, or
incompatible physical data.  For a candidate new column \(c\), the value
\(\lambda^Tc\) states directly whether the new term changes the obstructed equation.

The right null space describes the freedom left after all three equations have been
imposed.  The column vector
\begin{equation}
  v=\begin{pmatrix}s\\t\\u\end{pmatrix}
  \label{eq:fourpoint-right-kernel}
\end{equation}
obeys \(Mv=0\).  Consequently \(n\) and \(n+\alpha v\) give identical values for every
row in \cref{eq:fourpoint-constraint-matrix}.  The right kernel consists of changes of
the unknowns that leave the stated equations unchanged.

Here that statement can be checked directly at the color-dressed level.  Under
\begin{equation}
  \Delta n_s=\alpha s,\qquad
  \Delta n_t=\alpha t,\qquad
  \Delta n_u=\alpha u,
  \label{eq:fourpoint-gauge-shift}
\end{equation}
the Jacobi combination and the amplitude change by
\begin{align}
  \Delta(n_s+n_t+n_u)&=\alpha(s+t+u)=0,
  \label{eq:fourpoint-gauge-jacobi}\\
  \Delta\A_4^{\rm tree}
  &=\alpha(c_s+c_t+c_u)=0.
  \label{eq:fourpoint-gauge-amplitude}
\end{align}
Each inverse propagator in \cref{eq:fourpoint-gauge-shift} cancels its graph propagator,
and the resulting contact terms cancel by the color Jacobi identity.  This right-kernel
direction is the simplest generalized gauge transformation.

In the loop problem, the right kernel is first defined relative to the chosen cut
equations.  An observable-equivalence or integrated-amplitude map then determines the
effect of these directions on later quantities.  A fixed representative supplies
coordinates within this family.  The normalized reduced-row-echelon prescription used
later is the large-system analogue of setting one free coordinate, such as \(n_t\), to
zero.

\begin{remark}[An explicit Yang--Mills helicity target]
The preceding derivation applies to any helicity choice.  For the maximally
helicity-violating (MHV) configuration
\((1^-,2^-,3^+,4^+)\), the two Parke--Taylor amplitudes in spinor-helicity variables
\cite{Parke:1986gb} are
\begin{align}
  A(1^-,2^-,3^+,4^+)
  &=i\frac{\langle12\rangle^4}
  {\langle12\rangle\langle23\rangle\langle34\rangle\langle41\rangle},
  \label{eq:fourpoint-parke-taylor-first}\\
  A(1^-,3^+,2^-,4^+)
  &=i\frac{\langle12\rangle^4}
  {\langle13\rangle\langle32\rangle\langle24\rangle\langle41\rangle}.
  \label{eq:fourpoint-parke-taylor-second}
\end{align}
Momentum conservation sandwiched between
\(\langle4|\) and \(|1]\) gives
\begin{equation}
  0=\left\langle4\left|\sum_{i=1}^{4}k_i\right|1\right]
   =\langle42\rangle[21]+\langle43\rangle[31],
  \qquad
  \frac{[21]}{[31]}=-\frac{\langle34\rangle}{\langle24\rangle}.
  \label{eq:fourpoint-spinor-conservation}
\end{equation}
Consequently,
\begin{equation}
  \frac{A(1,3,2,4)}{A(1,2,3,4)}
  =-\frac{\langle12\rangle\langle34\rangle}
          {\langle13\rangle\langle24\rangle}
  =\frac{s}{u},
  \label{eq:fourpoint-parke-taylor-ratio}
\end{equation}
which verifies \cref{eq:fourpoint-bcj-amplitude-relation}.  Substituting the first of
these functions into \cref{eq:fourpoint-affine-family} gives, for example, the simple
representative
\begin{equation}
  n_s=sA(1,2,3,4),\qquad n_t=0,\qquad n_u=-sA(1,2,3,4).
  \label{eq:fourpoint-mhv-representative}
\end{equation}
This representative displays the solution directly and inherits the poles of the
color-ordered amplitude.  A local construction uses the freedom in \(\alpha\) together
with a local ansatz to impose the desired pole structure, power counting, and graph
symmetries.
\end{remark}

\subsection{From kinematic functions to finitely many coefficients}

At four points we solved directly for the functions \(n_s,n_t,n_u\).  At one loop the
unknown numerators depend on the loop momentum and contain many allowed contractions of
\(\ell\), external momenta, and polarizations.  The step that turns this functional
problem into finite linear algebra is an ansatz.  Choose a finite set of allowed basis
functions and write
\begin{equation}
  n_g(z;x)=\sum_{a=1}^{N}x_a B_{ga}(z),
  \qquad z=(\ell,k_i,\varepsilon_i).
  \label{eq:function-to-coefficients}
\end{equation}
The basis functions \(B_{ga}\) contain all chosen kinematic dependence; the unknowns
\(x_a\) are ordinary coefficients.  Locality, mass dimension, loop-momentum power
counting, antisymmetry, and Jacobi relations determine which \(B_{ga}\)'s are admitted and
how they are related.

As a two-column coefficient example, consider the following local four-point tensor
structures of the correct mass dimension:
\begin{equation}
  m_1=s(\varepsilon_1\mathbin{\cdot}\varepsilon_2)
        (\varepsilon_3\mathbin{\cdot}\varepsilon_4),
  \qquad
  m_2=u(\varepsilon_1\mathbin{\cdot}\varepsilon_3)
        (\varepsilon_2\mathbin{\cdot}\varepsilon_4).
  \label{eq:fourpoint-local-monomials}
\end{equation}
Suppose that a two-column local ansatz and its signed graph
relabelings give \(n_g=x_1B_{g1}+x_2B_{g2}\).  The two amplitude measurements then act
directly on the coefficient vector:
\begin{equation}
  \begin{pmatrix}
    B_{s1}/s-B_{t1}/t & B_{s2}/s-B_{t2}/t\\
    B_{t1}/t-B_{u1}/u & B_{t2}/t-B_{u2}/u
  \end{pmatrix}_{z=z_q}
  \begin{pmatrix}x_1\\x_2\end{pmatrix}
  =\begin{pmatrix}A_s\\A_u\end{pmatrix}_{z=z_q}.
  \label{eq:fourpoint-two-column-row-example}
\end{equation}
These two monomials illustrate the conversion.  The five-point calculation uses a
complete symmetry-compatible basis constructed in \cref{sec:ansatz-compiler}.  After
evaluation at \(z_q\), every tensor contraction is a known exact number, leaving
\(x_1,x_2\) as the unknowns in this two-column example.

Let \(C\) denote a tree-amplitude relation or a generalized cut, evaluated at an exact
on-shell kinematic point \(z_q\).  If \(W_{Cg}(z_q)\) denotes the known propagator,
routing, and residue weight of graph
\(g\), then
\begin{align}
  \sum_g W_{Cg}(z_q)n_g(z_q;x)
  &=b_C(z_q),\notag\\
  \sum_{a=1}^{N}
  \underbrace{\left[\sum_gW_{Cg}(z_q)B_{ga}(z_q)\right]}_{A_{(C,q),a}}
  x_a
  &=b_C(z_q).
  \label{eq:one-measurement-one-row}
\end{align}
The expression in brackets is one matrix entry.  One cut at one sample point gives one
or more rows; stacking topologies, internal-state choices, and kinematic points builds
\(Ax=b\).
At tree level the targets are color-ordered amplitudes.  On a loop generalized cut they
are products of on-shell tree amplitudes summed over internal states.  Both sources lead
to the same linear inverse-problem structure.

For the five-point calculation, four linear maps organize the bracketed operation in
\cref{eq:one-measurement-one-row}:
\begin{equation}
  \begin{aligned}
    \underbrace{x}_{\text{coefficients}}
    &\xrightarrow{\mathsf{Build}}
    \underbrace{N}_{\text{master numerator}}
    \xrightarrow{\mathsf{Jac}}
    \underbrace{\{n_g\}}_{\text{related graph numerators}}\\[-1mm]
    &\xrightarrow{\mathsf{Prop}}
    \underbrace{\text{graph integrand}}_{\text{propagators and routes}}
    \xrightarrow{\mathsf{Cut}_C}
    \underbrace{b_C}_{\text{cut data}}.
  \end{aligned}
  \label{eq:fourpoint-to-compiled-map}
\end{equation}
Here \(\mathsf{Build}\) constructs the chosen master numerator from the coefficient
vector, \(\mathsf{Jac}\) generates the graph numerators related by Jacobi identities,
\(\mathsf{Prop}\) inserts the propagator and routing factors, and
\(\mathsf{Cut}_C\) takes the residue on cut \(C\).  Their composition is
\(A_C=\mathsf{Cut}_C\mathsf{Prop}\mathsf{Jac}\mathsf{Build}\).

The correspondence between this hand calculation and the five-point reconstruction is
summarized below.
\begin{center}
\small
\begin{tabularx}{\textwidth}{@{}
  >{\raggedright\arraybackslash}X
  >{\raggedright\arraybackslash}X
  >{\raggedright\arraybackslash}X@{}}
\toprule
four-point hand calculation & one-loop five-point calculation & common meaning \\
\midrule
\(n_s,n_t,n_u\) & \(1127\) coefficients after graph-symmetry reduction & unknown representation data \\
one Jacobi row & graph numerators generated with antisymmetry and Jacobi identities &
structural constraints \\
two color-ordered amplitudes & maximal, box, triple, and double cuts & physical
measurements \\
\(3\times3\) matrix \(M\) & sparse cut matrices
\(A_C=\mathsf{Cut}_C\mathsf{Prop}\mathsf{Jac}\mathsf{Build}\) & measurement map \\
\(\lambda^Tb\neq0\) & left-null obstruction & target lies outside the image of the chosen basis \\
\(\alpha(s,t,u)\) & \(K\boldsymbol\beta\), with \(207\) directions & unresolved
right-kernel freedom \\
\(n_t=0\) as a convenient choice & fixed reduced-row-echelon representative & deterministic
coordinate choice \\
direct substitution & independent cuts, implementations, and exact reruns & verification \\
\bottomrule
\end{tabularx}
\end{center}

The same construction applies when the unknown numerators are functions: a finite
basis converts them into coefficients, and physical evaluations give linear rows.  At
four points, graph signs, routes, basis choices, and all three equations fit on one page.
The five-point calculation organizes thousands of tensor structures, related graph
numerators, cut topologies, exact kinematic samples, and independent checks.

Explicit graph-sign, routing, basis, and cut records support the same two null-space
calculations and interpretations illustrated by
\cref{eq:fourpoint-constraint-matrix}.
The routines and records used to perform the larger calculation are collected in
\cref{app:pseudocode,app:evidence-ledger}.  The next section supplies the cubic-graph,
generalized-cut, and dimensional-regularization background needed to formulate that
larger system.

\section{One-loop BCJ reconstruction from cubic graphs and generalized cuts}
\label{sec:bcj-fivepoint-framework}
\label{sec:physics-background}

\subsection{Cubic graphs, Jacobi identities, and representation freedom}

The hand calculation in \cref{sec:four-point-walkthrough} separated each cubic-graph
term into a propagator, a color factor, and a kinematic numerator.  The same separation
works at loop level.  Interactions with four or more legs are assigned to cubic graphs by
putting suitable inverse propagators into their numerators; when divided by the graph
propagators, those factors reproduce the original contact terms.  A one-loop
\(n\)-gluon amplitude can then be represented as
\begin{equation}
  \A_n^{(1)}
  =i g^n\!\int\!\frac{d^D\ell}{(2\pi)^D}
  \sum_{g\in\Gamma_{n,3}^{(1)}}
  \frac{1}{S_g}\frac{c_g\,n_g(\ell)}
  {\prod_{e\in g}D_e(\ell)}.
  \label{eq:cubic-one-loop-amplitude}
\end{equation}
Here \(S_g\) is the symmetry factor that corrects for graph automorphisms, and \(c_g\)
is built from structure constants.  The quantity
\(D_e(\ell)=(\ell+\Delta_e)^2+i0\) is the inverse propagator of an internal line.  Here
\(\Delta_e\) is the fixed external-momentum offset assigned to that line.  Once a
loop-momentum route and vertex orientations have been fixed, \(n_g(\ell)\) is an ordinary
kinematic function of momenta and polarizations.  It is graph dependent and generally
gauge dependent, while the properly assembled amplitude is the observable quantity.

\begin{definition}[BCJ numerator collection]
For a fixed set of oriented cubic graphs and momentum routes, a numerator collection is
an assignment \(g\mapsto n_g\) that reconstructs the amplitude through
\cref{eq:cubic-one-loop-amplitude}.  It is a \emph{BCJ}, or color--kinematics-dual,
collection when every vertex flip and every local Jacobi move obeys the same signed
linear relation as the corresponding color factors.
\end{definition}

Orientations can be chosen so that a flip at a cubic vertex gives
\begin{equation}
  c_{\bar g}=-c_g.
  \label{eq:color-antisymmetry}
\end{equation}
For any three graphs differing only at one internal edge, the color Jacobi identity is
\begin{equation}
  c_i+c_j+c_k=0.
  \label{eq:color-jacobi}
\end{equation}
The parallel kinematic relations are
\begin{equation}
  n_{\bar g}=-n_g,
  \qquad n_i+n_j+n_k=0.
  \label{eq:kinematic-jacobi}
\end{equation}
Thus ``the numerator'' in the rest of this paper means a coordinated assignment to all
relevant graphs, generated from a chosen set of master numerators.  A master numerator
is a generating graph numerator from which the others follow by these graph relations.
This object differs from the numerator of the fully combined amplitude.

\begin{figure}[H]
\centering
\begin{tikzpicture}[scale=.96,line cap=round,line join=round]
  \draw[bcjgray!45,dashed,thick]
    (-3.5,0)--(0,2.75)--(3.5,0)--cycle;
  \node[align=center,font=\small,fill=bcjlight!96,rounded corners=1.5pt,
    inner xsep=5pt,inner ysep=2pt] at (0,1.30)
    {local Jacobi plane};

  \begin{scope}[shift={(0,2.75)}]
    \coordinate (L) at (-.52,0); \coordinate (R) at (.52,0);
    \draw[very thick,bcjblue] (L)--(R);
    \draw[thick] (L)--(-1.28,.66) node[left,font=\scriptsize] {\(a\)};
    \draw[thick] (L)--(-1.28,-.66) node[left,font=\scriptsize] {\(b\)};
    \draw[thick] (R)--(1.28,.66) node[right,font=\scriptsize] {\(c\)};
    \draw[thick] (R)--(1.28,-.66) node[right,font=\scriptsize] {\(d\)};
    \fill (L) circle (1.8pt); \fill (R) circle (1.8pt);
    \node[font=\small\bfseries,bcjblue] at (0,1.05) {\(s\)-channel: \(n_s\)};
  \end{scope}

  \begin{scope}[shift={(-3.5,0)}]
    \coordinate (L) at (-.52,0); \coordinate (R) at (.52,0);
    \draw[very thick,bcjcyan!80!black] (L)--(R);
    \draw[thick] (L)--(-1.28,.66) node[left,font=\scriptsize] {\(b\)};
    \draw[thick] (L)--(-1.28,-.66) node[left,font=\scriptsize] {\(c\)};
    \draw[thick] (R)--(1.28,.66) node[right,font=\scriptsize] {\(a\)};
    \draw[thick] (R)--(1.28,-.66) node[right,font=\scriptsize] {\(d\)};
    \fill (L) circle (1.8pt); \fill (R) circle (1.8pt);
    \node[font=\small\bfseries,bcjcyan!70!black] at (0,-1.10)
      {\(t\)-channel: \(n_t\)};
  \end{scope}

  \begin{scope}[shift={(3.5,0)}]
    \coordinate (L) at (-.52,0); \coordinate (R) at (.52,0);
    \draw[very thick,bcjgold!70!black] (L)--(R);
    \draw[thick] (L)--(-1.28,.66) node[left,font=\scriptsize] {\(c\)};
    \draw[thick] (L)--(-1.28,-.66) node[left,font=\scriptsize] {\(a\)};
    \draw[thick] (R)--(1.28,.66) node[right,font=\scriptsize] {\(b\)};
    \draw[thick] (R)--(1.28,-.66) node[right,font=\scriptsize] {\(d\)};
    \fill (L) circle (1.8pt); \fill (R) circle (1.8pt);
    \node[font=\small\bfseries,bcjgold!55!black] at (0,-1.10)
      {\(u\)-channel: \(n_u\)};
  \end{scope}
\end{tikzpicture}
\caption{The three graphs in one Jacobi relation.  Cutting one internal edge leaves
four half-edges, which have three cubic reconnections.  With a fixed orientation, both
the three color factors and the three kinematic numerators sum to zero.}
\label{fig:jacobi-triangle}
\end{figure}
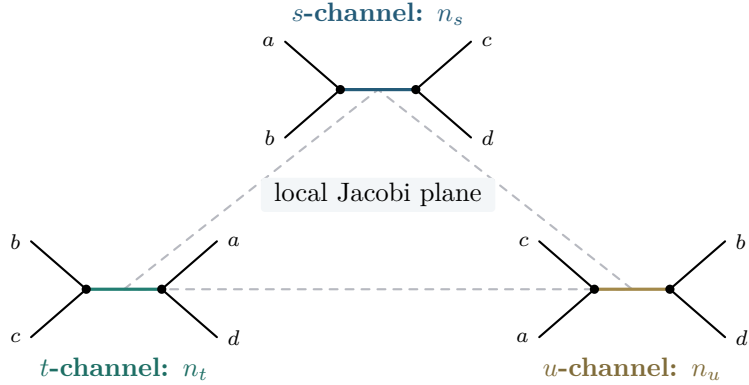

The original duality and tree-level amplitude relations were introduced in
\cite{Bern:2008qj}; loop-level constructions and the double copy were developed in
\cite{Bern:2010ue,Carrasco:2011mn,Bern:2013yya,Boels:2013bi}.  Reviews include
\cite{Bern:2019prr}.
Constructive tree-level algorithms provide useful points of comparison.  They have
produced relabeling-symmetric BCJ numerators through six points
\cite{Fu:2014pya}.  In the scattering-equation representation, a concrete reduction of
the CHY integrand makes color--kinematics duality manifest term by term
\cite{Bjerrum-Bohr:2016axv}.  These constructions select a numerator representation
through an explicit algebraic organization.  The present calculation complements them
by determining the complete affine family compatible with a chosen local ansatz and
cut map.

If \(n_g\) and \(\widetilde n_g\) are two duality-satisfying numerator sets, the formal
double copy is
\begin{equation}
  \M_n^{(1)}
  =i\left(\frac{\kappa}{2}\right)^n
  \int\!\frac{d^D\ell}{(2\pi)^D}
  \sum_g\frac{1}{S_g}
  \frac{n_g(\ell)\widetilde n_g(\ell)}{\prod_eD_e(\ell)}.
  \label{eq:loop-double-copy}
\end{equation}
The complete numerator collection organizes all graph crossings and color relations in
one representation and can be used directly in the gravity double copy.

The affine line in \cref{eq:fourpoint-affine-family} already showed that a cubic
representation can be redundant: distinct graph numerators may obey Jacobi and assemble
to the same amplitude.  For a general graph set, a deformation
\begin{equation}
  n_g\longmapsto n_g+\Delta_g
  \label{eq:gauge-shift}
\end{equation}
preserves the gauge-theory amplitude when
\begin{equation}
  \sum_g\frac{1}{S_g}\frac{c_g\Delta_g}{\prod_eD_e}=0
  \label{eq:generalized-gauge-condition}
\end{equation}
under the chosen notion of equality.  At loop level, four notions are used:
\begin{align}
  n\sim_{\rm alg}n'&:\quad n-n'\text{ vanishes by the listed on-shell and algebraic identities},
  \label{eq:eq-alg}\\
  n\sim_{\rm cut}n'&:\quad \Res_C(n-n')=0\text{ on every generalized cut used in the reconstruction},
  \label{eq:eq-cut}\\
  n\sim_{\cS}n'&:\quad n-n'\text{ vanishes under the observable-equivalence map }\cS,
  \label{eq:eq-s}\\
  n\sim_{\I}n'&:\quad \I[n-n']=0\text{ for the chosen integration, LSZ, and scheme prescription}.
  \label{eq:eq-int}
\end{align}
Algebraic identities act on local tensor expressions.  Cut equality compares their
generalized-cut residues.  The observable-equivalence map \(\cS\) combines fixed graph
routing, scaleless-integral identities, and an
LSZ projection.  Integrated equality is defined by a complete integration and
renormalization prescription.  Loop-momentum shifts, total derivatives,
evanescent terms, scaleless integrals, external-leg bubbles, and renormalization
conventions determine the maps between these levels.

The deterministic reduced row echelon form (RREF) rule selects
\(x^{(0)}+D_sx^{(1)}\), which we call the \emph{fixed representative}.  The full set in
\eqref{eq:main-affine-family} is the \emph{reconstructed numerator family}.  The
observable-equivalence calculation determines whether this family has one image under
\(\cS\).

Forward-limit constructions relate one-loop integrands to tree objects with two extra
legs carrying momenta \(\ell\) and \(-\ell\).  They have produced loop-level correlators,
relations, and BCJ numerators in several theories
\cite{He:2016mzd,He:2017spx,Mafra:2014oia,He:2022iqi}.  Cao, He, Zhang, and Zhu derived
loop-level double-copy relations from single cuts and gave local crossing-symmetric
Yang--Mills BCJ numerators through five points \cite{Cao:2025forward}.  Their formula
supplies a published numerator collection in the same coefficient space.  We first
reconstruct the affine family from the chosen basis and generalized cuts, and then
express that collection in the reconstructed coordinates:
\begin{equation}
  \begin{gathered}
  \underbrace{\text{basis}\to\text{cut equations}\to\text{exact solve}\to\text{checks}}
  _{\text{reconstruction}}
  \\[1mm]\Downarrow\\[-1mm]
  \underbrace{\text{published numerator}\to\text{coordinates in the solution family}}
  _{\text{coordinate comparison}}
  \end{gathered}.
  \label{eq:post-audit-boundary}
\end{equation}
From the scattering-equation viewpoint, the limiting operations in this comparison have
different solution behavior.  An isolated factorization channel produces a collapsing
subset of punctures, a generic soft limit does not require a singular solution, and the
forward limit contains two distinct collision scalings \cite{Feng:2021iik}.  These facts
describe how an integrand representation is obtained in singular kinematics.  The
coordinate comparison gives the kernel coordinates in
\cref{eq:known-coordinate-bridge}.  Fixed-routing equality and equality after LSZ
projection and integration are evaluated as two distinct levels below.

\subsection{Generalized cuts, dimensional regularization, and exact equations}

Generalized unitarity reconstructs loop integrands from products of on-shell trees
\cite{Bern:1994zx,Bern:1994cg,Britto:2004nc,Forde:2007mi,Giele:2008ve,Badger:2008cm}.
Other one-loop organizations use the same tree data in complementary ways.  The
Q-cut representation reduces channel by channel to the conventional unitarity-cut
integrand \cite{Huang:2015cwh}, while one-loop CHY differential operators factorize into
their tree-level counterparts on a unitarity cut \cite{Zhou:2021kzv}.  Here the cut is
used as an exact linear measurement on a common BCJ coefficient space.
For a cut \(C\) placing a set of propagators on shell,
\begin{equation}
  \left.\left(\prod_{e\in C}D_e\right)\I_n^{(1)}\right|_{D_e=0}
  =\sum_{\text{physical states}}
   \prod_{v\in C}\A_v^{\rm tree}.
  \label{eq:generalized-cut}
\end{equation}
Given a numerator ansatz linear in coefficients \(x_j\), every sampled cut produces a
linear equation
\begin{equation}
  \sum_j A_{C,ij}x_j=b_{C,i}.
  \label{eq:cut-linear-row}
\end{equation}
The matrix \(A\) maps numerator coefficients to cut data.  Its right null space (kernel)
contains coefficient changes that vanish on the sampled cut equations.  Its left null
space tests
whether the target vector belongs to the span of the chosen numerator terms:
\begin{align}
  v\in\ker A &\quad\Longrightarrow\quad
      x\text{ and }x+v\text{ have identical sampled cuts},
      \label{eq:right-kernel-meaning}\\
  \lambda^TA=0,\ \lambda^Tb\ne0
    &\quad\Longrightarrow\quad b\notin\im A.
      \label{eq:left-obstruction-meaning}
\end{align}
A nonzero \(\lambda^Tb\) identifies an inconsistency between the current ansatz and the
target cuts.  A nontrivial right kernel identifies directions that require additional cuts
or a later equivalence classification.  \Cref{sec:cut-bootstrap} constructs the cut
equations in a fixed order and applies these two tests.

The cut equations also retain the dimensional data required by the regularization
scheme.  We keep several dimensions conceptually separate:
\begin{equation}
  D_{\rm mom,ext},\qquad D_{\rm pol,ext},\qquad D_\ell,
  \qquad D_s.
  \label{eq:dimension-tuple}
\end{equation}
The external momenta span at most a four-dimensional subspace at five points.
\(D_{\rm mom,ext}\) and \(D_{\rm pol,ext}\) denote the dimensions used for external
momenta and polarizations.  The loop momentum is treated in \(D_\ell\) dimensions, and
\(D_s\) is the dimension used in the internal polarization-state sum.  In a
four-dimensional helicity setup one writes
\begin{equation}
  \ell^\mu=\bar\ell^\mu+\widetilde\ell^\mu,
  \qquad \widetilde\ell^{\,2}=-\mu^2,
  \qquad \bar\ell\in\mathbb R^{1,3}.
  \label{eq:loop-split}
\end{equation}
The internal completeness relation is evaluated in this spin space.  At one
loop in pure Yang--Mills, the reconstructed coefficient vector is affine in \(D_s\),
which motivates
\begin{equation}
  x_{\rm p}(D_s)=x^{(0)}+D_sx^{(1)}.
  \label{eq:ds-affinity}
\end{equation}
The four-dimensional helicity (FDH) and 't Hooft--Veltman (HV) schemes use different
specializations of the spin dimension relative to \(D=4-2\epsilon\).  The scheme is
specified when amplitudes are compared.

Working in \(D\) dimensions retains rational terms.  Four-dimensional cuts can miss
\(\mu^2\)-dependent contributions that integrate to finite rational functions.  The
integral-basis calculation therefore includes dimension-shifting identities and a
shifted pentagon.

The resulting linear equations are solved exactly over prime fields, a standard method
for reconstructing rational functions in amplitude calculations
\cite{Peraro:2016wsq,Peraro:2019svx,Abreu:2017xsl}.  Exact arithmetic gives unambiguous
ranks, null spaces, and residuals; several primes and independent kinematic samples
separate generic statements from accidental specializations.  Chinese remaindering and
rational reconstruction then lift the result to \(\Q\).  The arithmetic algorithms and
uniqueness bounds are recorded in \cref{app:exact-arithmetic}.  Using these graph,
equivalence, cut, and regulator conventions, the next section formulates the finite
five-point inverse problem.

\subsection{The finite five-point inverse problem}
\label{sec:problem}

The four-point calculation of \cref{sec:four-point-walkthrough} started from a fixed
matrix and solved it.  At five points, constructing that matrix is part of the
scientific problem.  The construction specifies a finite coefficient space, graph
relations in a common loop route, cut measurements, an equivalence relation for
residual differences, and independent validation tests.  These ingredients distinguish
the scientific target from its finite computational realization.

The \emph{fixed scientific specification} \(\Omega\) collects the physical and
mathematical inputs held constant across the calculation:
\begin{equation}
  \Omega=\left(
  \begin{array}{c}
  \text{pure Yang--Mills theory and the one-loop five-gluon process},\\
  \text{a BCJ numerator collection as the reconstruction target},\\
  \text{the chosen local parity-even ansatz envelope and graph identities},\\
  \text{the definition of generalized-cut targets and allowed conventions},\\
  \text{the specified \(S\)-equivalence, validation tests, and completion criteria}
  \end{array}\right).
  \label{eq:scientific-charter}
\end{equation}
The theory, target amplitude, graph identities, and equivalence relation remain fixed
while the finite basis, measurement set, and evaluation sequence are refined.

At iteration \(t\), a \emph{stage specification}
\(\sigma_t\in\Sigma(\Omega)\) gives the current finite realization.  It specifies the
active tensor and contact sectors, graph and routing construction, cut topologies and
samples, arithmetic choices, quotient calculation, and validation tests.  Diagnostics
update these choices within \(\Omega\): an obstruction can activate a contact sector in
the ansatz envelope, and a right-kernel diagnostic can add a cut topology from the
measurement family.

Together these choices define the finite reconstruction problem
\begin{equation}
  \boxed{\;
  \mathcal P_{\sigma}
  =\bigl(A_{\sigma},b_{\sigma};\,\cS_{\sigma};\,\cT_{\sigma}\bigr).
  \;}
  \label{eq:compiled-scientific-problem}
\end{equation}
Here \(A_\sigma x=b_\sigma\) reconstructs the cut-solution fiber in the active
coefficient space.  The map \(\cS_\sigma\) acts on differences of solutions and defines
observable equivalence.  The collection \(\cT_\sigma\) contains held-out cuts, integral
reduction, helicity amplitudes, and comparisons between specified integration
prescriptions.  Thus \(\cS_\sigma\) classifies ambiguity within the fiber, while
\(\cT_\sigma\) tests physical consequences of the reconstruction.

The corresponding compilation and evaluation maps are
\begin{equation}
  \sigma_t
  \xrightarrow{\ \mathrm{Compile}_{\Omega}\ }
  \bigl(A_t,b_t,\cS_t,\cT_t\bigr)
  \xrightarrow{\ \mathrm{Evaluate}\ }
  d_t,
  \label{eq:compile-evaluate-stage}
\end{equation}
with an exact diagnostic of the form
\begin{equation}
  d_t=\left(
  \rank A_t,\ \rank[A_t\mid b_t],\
  \ker A_t^T,\ \ker A_t,\ \cS_t(\ker A_t),\
  r_{\mathrm{heldout}}\right).
  \label{eq:complete-stage-diagnostic}
\end{equation}
Compilation selects the finite inverse problem, and Gaussian elimination solves it.
\Cref{sec:agent-assisted-framework} describes how the diagnostic is used to propose the
next finite specification.  The rest of this section defines the concrete objects
compiled for the five-point problem.

At five points, the unknowns are coefficients of loop-momentum and polarization tensors.
They parametrize a related set of pentagon, box, triangle, and lower-topology graphs,
and products of tree amplitudes on generalized cuts supply the physical data.  Once a
stage specification has been compiled, the inner calculation tests its exact linear
system for consistency and determines its complete affine solution.

All external momenta are outgoing,
\begin{equation}
  k_i^2=0,
  \qquad \sum_{i=1}^{5}k_i=0,
  \qquad k_i\cdot\varepsilon_i=0.
  \label{eq:fivepoint-kinematics}
\end{equation}
We use cyclic Mandelstam invariants
\begin{equation}
  s_{12},\ s_{23},\ s_{34},\ s_{45},\ s_{51},
  \qquad s_{ij}=(k_i+k_j)^2,
  \label{eq:cyclic-invariants}
\end{equation}
supplemented, when needed, by the parity-odd invariant
\begin{equation}
  \trfive=4i\,\epsilon_{\mu\nu\rho\sigma}
    k_1^\mu k_2^\nu k_3^\rho k_4^\sigma.
  \label{eq:tr5-definition}
\end{equation}
The parity-even ansatz is built without explicit \(\trfive\) denominators.  The
helicity checks retain both parity-conjugate kinematic solutions because
spinor-helicity evaluations distinguish the two roots of the Gram relation.

Choose the adjacent pentagon routing
\begin{equation}
  \ell_1=\ell,
  \qquad \ell_{i+1}=\ell_i+k_i,
  \qquad D_i=\ell_i^2.
  \label{eq:adjacent-routing}
\end{equation}
A routing table records the loop-momentum assignment of every graph and the shift needed
to express it in \cref{eq:adjacent-routing}.  All graph operations use this table, so
pointwise comparisons and cut equations use the same loop-momentum convention.  The
explicit routing and polarization conventions are collected in \cref{app:conventions}.

The graph family is specified by a small set of generating numerators.  A \emph{master
numerator} is a generating graph numerator from which other graph
numerators follow by symmetries and Jacobi identities.  A numerator obtained in this way
is called a \emph{graph descendant}.  This use of ``master'' is separate from the master
integrals used later in integral reduction.  We assign unknown functions to the
generating set and construct the remaining numerators from the graph relations.

Fix leg 1 and denote a pentagon master by
\begin{equation}
  N_{1|\rho}(\ell),\qquad \rho\in S_4.
  \label{eq:master-words}
\end{equation}
The word \(1|\rho\) records the cyclic order of the five external legs around the loop.
The \(4!=24\) master words generate all cubic graphs through antisymmetrization,
Jacobi brackets, and loop-momentum shifts.  For example, antisymmetrizing two adjacent
entries replaces their ordered pair by the box graph in the corresponding Jacobi
relation.  The resulting linear relation is
\begin{equation}
  n_g(\ell)=\sum_{\rho\in S_4}J_{g\rho}[\tau_{\Delta_{g\rho}}]
  N_{1|\rho}(\ell),
  \qquad \tau_\Delta f(\ell)=f(\ell+\Delta).
  \label{eq:descendant-map-problem}
\end{equation}
The map \(J_{g\rho}\) implements the nested antisymmetrizations, also called free-Lie
relations, together with the required momentum shifts \(\tau_\Delta\).  The offset
\(\Delta_{g\rho}\) is fixed by the graph and master word.  The kinematic
Jacobi identities therefore hold symbolically in the generated graph numerators before
the numerical rows are constructed.  The explicit graph-generation checks are recorded
in \cref{app:pseudocode,app:evidence-ledger}.

These routes and graph relations turn the functional reconstruction into a finite
coefficient problem.
Let \(V_{\rm raw}\) be the ordered span of local parity-even tensor monomials allowed by
the power-counting specification.  Let \(R_{\rm kin}\) be the exact relation matrix
generated by on-shell conditions, momentum conservation, transversality, and the chosen
algebraic reductions.  Then
\begin{equation}
  V_{\rm ansatz}=V_{\rm raw}/\im R_{\rm kin},
  \qquad \dim V_{\rm ansatz}=10{,}724.
  \label{eq:ansatz-quotient}
\end{equation}
The quotient notation identifies two tensor expressions when their difference follows
from these kinematic identities.  For example,
\(k_i^2=0\) removes monomials proportional to \(k_i^2\),
\(k_i\mathbin{\cdot}\varepsilon_i=0\) removes longitudinal factors, and momentum
conservation can replace contractions with \(k_5\) by contractions with
\(-k_1-\cdots-k_4\).  These reductions remove duplicate columns before the generalized
cuts are imposed.

To write an equivalence class as an explicit tensor, we choose one expression in each
class by a fixed linear map called a \emph{quotient section}.  This map converts every
basis vector of \(V_{\rm ansatz}\) back to a definite expression in \(V_{\rm raw}\).
Consequently every coefficient vector and graph numerator can be converted between the
reduced basis and the unreduced tensor notation.

Pentagon rotations, reflections, and graph-orientation signs give a second reduction,
\begin{equation}
  O:\Q^{1127}\longrightarrow V_{\rm ansatz}.
  \label{eq:orbit-map}
\end{equation}
Here the \(1127\) input coordinates are the independent coefficients left after these
graph symmetries are imposed.  The \(7100\)-dimensional algebraic kernel is the kernel
of the raw tensor reduction \(\Phi\); it has already been removed before the orbit map
\(O\) is built.  The \(207\)-dimensional right kernel appears after the cut equations
are solved and consists of numerator changes left free by those cuts.
Evaluating the basis-expanded
numerator produces one linear row.  In the five-point calculation this evaluation is
factored into four linear maps.  For a cut family \(C\), basis coefficients
pass through
\begin{equation}
  A_C=\mathsf{Cut}_C\,\mathsf{Prop}\,\mathsf{Jac}\,\mathsf{Build}.
  \label{eq:compiled-factorization}
\end{equation}
The map \(\mathsf{Build}\) constructs master numerators from the \(1127\)
graph-symmetry-reduced coefficients.  The map \(\mathsf{Jac}\) generates all related
graph numerators and applies their loop-momentum shifts.  The map \(\mathsf{Prop}\)
inserts the propagator weights, and \(\mathsf{Cut}_C\) takes the residue on cut \(C\) at
an on-shell point.  Their composition maps the coefficient vector directly to its cut
value.  The target
\begin{equation}
  b_C=\sum_{\rm states}\prod_{v\in C}\A_v^{\rm tree}
  \label{eq:raw-target}
\end{equation}
is the product of tree amplitudes on the same cut, summed over internal physical states.
It is computed independently of the numerator expression.

The reconstruction component of \(\mathcal P_\sigma\) is
\begin{equation}
  A_\sigma\,x(D_s)=b_\sigma(D_s),
  \qquad A_\sigma\in\Q^{m\times1127},
  \label{eq:full-inverse-problem}
\end{equation}
In this coefficient parametrization the graph-side matrix
\(A_\sigma\) is independent of \(D_s\); the internal-state sum enters through
\(b_\sigma(D_s)\).  The kernel is therefore common to the sampled spin dimensions,
while the particular solution is affine in \(D_s\).  Its exact finite-field evaluation
and rational reconstruction are described in \cref{app:exact-arithmetic}.  Once the
final specification is fixed, we suppress the subscript \(\sigma\) and write
\(A,b,\cS,\cT\).

The finite coefficient space and its cut map lead to three outputs.  First, solving
\(Ax=b\) gives one particular
solution and the complete right kernel of the cut matrix.  Second, the
observable-equivalence map \(\cS\) is defined for the color-ring projection
\(R_{12345}\), the coefficient with cyclic external ordering \((1,2,3,4,5)\).  It
classifies the resulting affine family after fixed loop-momentum routing,
scaleless-integral identities, and LSZ projection.  Third, graph
and color assembly followed by
integral reduction maps selected representatives to helicity amplitudes.  One comparison
tests the equation
\begin{equation}
  \mathsf H_{\rm presc}:\quad
  \I_{\rm DR}\!\left[\A^{\mathrm{FR}}-\A^{\mathrm{FL}}\right]=0.
  \label{eq:prescription-hypothesis}
\end{equation}
Here \(\A^{\mathrm{FR}}\) is the directly reconstructed fixed-routing integrand and
\(\A^{\mathrm{FL}}\) is the published forward-limit integrand.  The comparison uses a
fixed graph convention, routing, LSZ projection, integral basis, and
dimensional-regularization prescription.  The reconstruction and
observable-equivalence maps are evaluated for the complete affine family.  The
amplitude map is checked in \cref{sec:validation}.  Under the stated prescription, the
difference has the nonzero exact \(\Q(D)\) coefficient described in
\cref{sec:prescription-counterexample}.

\begin{figure}[t]
\centering
\begin{tikzpicture}[x=1cm,y=1cm,line cap=round,line join=round,
  every node/.append style={font=\small}]
  \node[pending,text width=22mm,inner xsep=1.5mm] (spec) at (1.35,.85)
    {current working\\problem \(\sigma\)};
  \node[evidence,text width=28mm,inner xsep=1.5mm] (problem) at (4.75,.85)
    {compiled problem\\\((A,b;\cS;\cT)\)};
  \node[evidence,text width=31mm,inner xsep=1.5mm] (fiber) at (9.00,.85)
    {cut-solution fiber\\\(\mathcal F_b=x_{\rm p}+\ker A\)};
  \node[evidence,text width=27mm,inner xsep=1.5mm] (quotient) at (13.55,.85)
    {observable quotient\\\(\mathcal F_b/\!\sim_{\cS}\)};
  \node[flow,text width=35mm,inner xsep=1.5mm] (validation) at (9.00,-1.25)
    {independent readouts \(\cT\)\\integration-by-parts reduction, helicity, prescriptions};

  \draw[arrow] (spec.east)--(problem.west);
  \draw[arrow] (problem.east)--(fiber.west);
  \draw[arrow] (fiber.east)--(quotient.west);
  \node[font=\scriptsize,fill=white,inner sep=1pt] at (3.05,1.58)
    {\(\mathrm{Compile}_{\Omega}\)};
  \node[font=\scriptsize,fill=white,inner sep=1pt] at (6.88,1.58)
    {solve \(Ax=b\)};
  \node[font=\scriptsize,fill=white,inner sep=1pt] at (11.28,1.58)
    {classify with \(\cS\)};
  \draw[arrow] (problem.south) to[out=-70,in=160] (validation.north west);
  \draw[arrow] (fiber.south)--(validation.north);
\end{tikzpicture}
\caption{The roles of the compiled objects.  The cut equations reconstruct an affine
fiber, \(\cS\) classifies differences inside that fiber, and \(\cT\) denotes independent
readouts and validation maps.  The three objects therefore represent reconstruction,
equivalence, and validation, respectively.}
\label{fig:claim-ladder}
\end{figure}
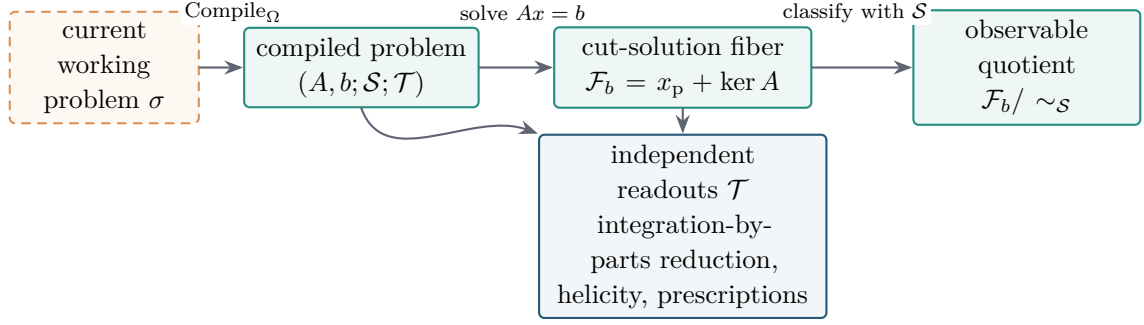

Exact-arithmetic details are given in
\cref{app:exact-arithmetic,app:evidence-ledger}.  The next section describes how
diagnostics propose and evaluate updates to \(\sigma_t\).  Subsequent sections construct
the finite tensor space and generalized-cut equations explicitly.

\section{Agent-assisted adaptive inverse-problem construction}
\label{sec:agent-assisted-framework}

The four-point example in \cref{sec:four-point-walkthrough} can be derived on paper in
one pass.  At five points, the same operations act on thousands of tensor structures
and many routed cut families.  It is useful to treat the calculation as a sequence of
finite inverse problems.  Each stage has an ordered coefficient space and an ordered
measurement set, and the exact result of that stage supplies the state from which the
next revision is chosen.

Proposal generation and exact evaluation occupy successive steps of the calculation.
A physicist, a fixed search policy, or a stateful agent may suggest the next candidate
revision.  In each case the same compiler constructs the finite problem and the same
evaluator calculates its ranks, residuals, and null spaces.  This section first
describes the diagnostic state, then introduces language models, agents, skills, and
harnesses in terms familiar to a physics calculation, and finally records the policies
used in the five-point construction.

\subsection{Exact diagnostics and candidate revisions}

The fixed scientific specification \(\Omega\), mutable working specification
\(\sigma_t\), compiled tuple
\(\mathcal P_{\sigma_t}=(A_t,b_t;\cS_t;\cT_t)\), and exact diagnostic \(d_t\) were
defined in \cref{sec:problem}.  The active numerator basis and cut schedule enter the
factorized map in \cref{eq:compiled-factorization}.  Its exact evaluation returns the
ranks of \(A_t\) and \([A_t\mid b_t]\), the left and right null spaces, the action of
\(\cS_t\) on the right kernel, and the residuals on \(\cT_t\).  Thus \(d_t\) is a
finite description of what the current ansatz can express, which coefficient
directions the active cuts leave unresolved, and how those directions behave on the
specified readouts.

Let \(L_t\) and \(K_t\) denote the ordered left- and right-null bases.  Their pairings
with candidate columns and rows give the information needed for a revision.

Three finite matrix calculations then answer three different questions:
\begin{equation}
  L_t^T C_{\rm cand}\ne0,
  \qquad
  \rank(R_{\rm cand}K_t)>0,
  \qquad
  \cS_tK_t=0.
  \label{eq:adaptive-revision-tests}
\end{equation}
A candidate column block \(C_{\rm cand}\) is useful when it can change an inconsistent
left-null combination.  A candidate row block \(R_{\rm cand}\) is informative when it
measures a direction in the current right kernel.  Once the measurement coverage is
complete, the last equality states that every remaining cut-preserving deformation has
the same specified readout.  These exact tests evaluate candidates supplied by the
proposal policies described below.

This organization continues the inverse-problem logic of inverse CHY construction
\cite{Li:2026inverseCHY}.  In that setting factorization channels give exact integer
constraints and residual information is propagated through a subset lattice.  Here
local tensor columns and generalized-cut rows give a rational linear problem, and
left- and right-null data guide revisions of its finite realization.  In both cases the
physical target is fixed before the constraint problem is adapted.

\subsection{Language models, agents, skills, and harnesses}

A \emph{large language model} (LLM) maps a supplied context to a continuation, which
may contain prose, code, or a structured tool call.  Persistent calculation state
enters through an \emph{agent} and its control loop: it reads the current observation,
chooses an action, receives the result, and may act again.  In this paper the
observation is the exact record \(d_t\), and the action space consists of typed
revisions of \(\sigma_t\), such as activating a numerator sector, sampling a cut family,
or scheduling an independent test.

A \emph{skill} is the domain implementation of one such action.  It specifies its input
type, preconditions, transformation, output order, and local checks.  A \emph{harness}
is the execution layer around the skills and scientific programs.  It normalizes an
input, invokes the selected programs, passes their outputs to the next step, and records
the resulting dependency graph.  The \emph{compiler} turns a working specification
into \((A_t,b_t;\cS_t;\cT_t)\), while the \emph{evaluator} performs the exact algebra on
those objects.  Thus the proposal policy chooses a candidate operation, and the
remaining layers give that operation a definite mathematical meaning.

If \(\alpha_t\) names the selected operation and \(F_{\alpha_t}\) is its skill, one
iteration is
\begin{equation}
  \widehat\sigma_{t+1}=F_{\alpha_t}(\sigma_t,d_t),
  \qquad
  \sigma_{t+1}
  =\operatorname{Normalize}_{\Omega}(\widehat\sigma_{t+1}),
  \qquad
  d_{t+1}=\operatorname{Evaluate}
  \!\left(\operatorname{Compile}_{\Omega}(\sigma_{t+1})\right).
  \label{eq:adaptive-stage-transition}
\end{equation}
The normalized proposal has the same form whether it came from a person, an explicit
search rule, or an LLM-based agent.  The exact record \(d_{t+1}\) then determines
whether the requested improvement occurred.

The roles can be summarized in calculation language:
\begin{center}
\small
\begin{tabularx}{\textwidth}
  {@{}>{\raggedright\arraybackslash}p{.17\textwidth}
      >{\raggedright\arraybackslash}p{.34\textwidth}X@{}}
\toprule
component & input and operation & output used by the next component \\
\midrule
proposal policy & current specification, exact diagnostic, and allowed actions;
selects a candidate revision & operation name and typed parameters \\
skill & applies one domain transformation & normalized candidate basis, cut schedule,
or test schedule \\
harness & runs the specified programs and records their dependencies & ordered program
outputs and provenance \\
compiler & constructs the maps in \cref{eq:compiled-factorization} & explicit finite problem
\((A_t,b_t;\cS_t;\cT_t)\) \\
exact evaluator & performs elimination, null-space calculations, reconstruction, and
residual tests & diagnostic and validation results \(d_t\) \\
\bottomrule
\end{tabularx}
\end{center}

Model architecture and external orchestration address different scales of this
workflow.  PJ-RoPE, for example, organizes relative-position responses inside a model
through a learnable Fourier--jet--affine representation \cite{Zhang:2026pjrope}.  At the
workflow scale, the agent state, skills, harness, compiler, and evaluator organize a
sequence of model calls and exact scientific programs.

\begin{figure}[H]
\centering
\begin{tikzpicture}[x=1cm,y=1cm,scale=.96,transform shape,
  line cap=round,line join=round,
  every node/.append style={font=\footnotesize}]
  \node[evidence,text width=142mm,minimum height=10mm,inner xsep=1.5mm,
    font=\scriptsize] (omega) at (7.75,3.15)
    {fixed scientific specification \(\Omega\): physical target, representation class,
     conventions, exact field, equivalence, and independent tests};

  \node[pending,text width=21mm,inner xsep=1.5mm] (sigmat) at (1.35,1.40)
    {working stage\\\(\sigma_t\)};
  \node[flow,text width=22mm,inner xsep=1.5mm] (harness) at (4.30,1.40)
    {harness\\normalize\\run and record};
  \node[flow,text width=29mm,inner xsep=1.5mm] (compile) at (7.85,1.40)
    {compiler \(\operatorname{Compile}_{\Omega}\)\\
     \((A_t,b_t,\cS_t,\cT_t)\)};
  \node[evidence,text width=25mm,inner xsep=1.5mm] (evaluate) at (11.45,1.40)
    {exact evaluator\\row reduction\\and tests};
  \node[evidence,text width=20mm,inner xsep=1.5mm] (diag) at (14.35,1.40)
    {diagnostic \(d_t\)\\ranks and\\null spaces};

  \node[pending,text width=24mm,inner xsep=1.5mm] (propose) at (14.10,-.70)
    {proposal policy\\person, rule, or agent};
  \node[flow,text width=27mm,inner xsep=1.5mm] (skill) at (10.45,-.70)
    {selected skill\\typed\\transformation};
  \node[pending,text width=22mm,inner xsep=1.5mm] (sigmanext) at (6.75,-.70)
    {revised stage\\\(\sigma_{t+1}\)};

  \draw[arrow] (sigmat.east)--(harness.west);
  \draw[arrow] (harness.east)--(compile.west);
  \draw[arrow] (compile.east)--(evaluate.west);
  \draw[arrow] (evaluate.east)--(diag.west);
  \draw[arrow] (diag.south) to[out=-90,in=85] (propose.north);
  \draw[arrow] (propose.west)--(skill.east);
  \draw[arrow] (skill.west)--(sigmanext.east);
  \draw[-{Stealth[length=2.4mm]},thick,densely dashed,draw=bcjorange!80!black]
    (sigmanext.north west) to[out=145,in=-75]
    node[pos=.50,left=2pt,font=\scriptsize,align=center,fill=white,inner sep=1pt]
    {next recorded\\iteration} (harness.south);

  \draw[-{Stealth[length=2mm]},thin,densely dashed,draw=bcjcyan!75!black]
    ($(omega.south)+(-6.40,0)$) -- (sigmat.north);
  \draw[-{Stealth[length=2mm]},thin,densely dashed,draw=bcjcyan!75!black]
    ($(omega.south)+(0.10,0)$) -- (compile.north);
  \draw[-{Stealth[length=2mm]},thin,densely dashed,draw=bcjcyan!75!black]
    ($(omega.south)+(3.70,0)$) -- (evaluate.north);

  \node[evidence,text width=41mm,minimum height=14mm,inner xsep=1.5mm] (n1) at (2.75,-3.25)
    {basis transition\\square-free trial\\\(+300\) repeated contacts};
  \node[evidence,text width=41mm,minimum height=14mm,inner xsep=1.5mm] (n2) at (7.75,-3.25)
    {row transition\\\(\rank:916\longrightarrow920\)\\
     \(\nullity:211\longrightarrow207\)};
  \node[evidence,text width=41mm,minimum height=14mm,inner xsep=1.5mm] (n3) at (12.75,-3.25)
    {quotient transition\\\(\cS k_a=0\), \(a=1,\ldots,207\)\\
     \(646=350+296\)};
\end{tikzpicture}
\caption{A common interface for adaptive inverse-problem construction.  The fixed
scientific specification determines the meaning of the calculation.  A proposal
policy revises the working stage, and the harness, compiler, and exact evaluator return
the next diagnostic.  The lower strip lists the three recorded transitions in the
five-point case study.}
\label{fig:llm-agent-skill-harness}
\end{figure}
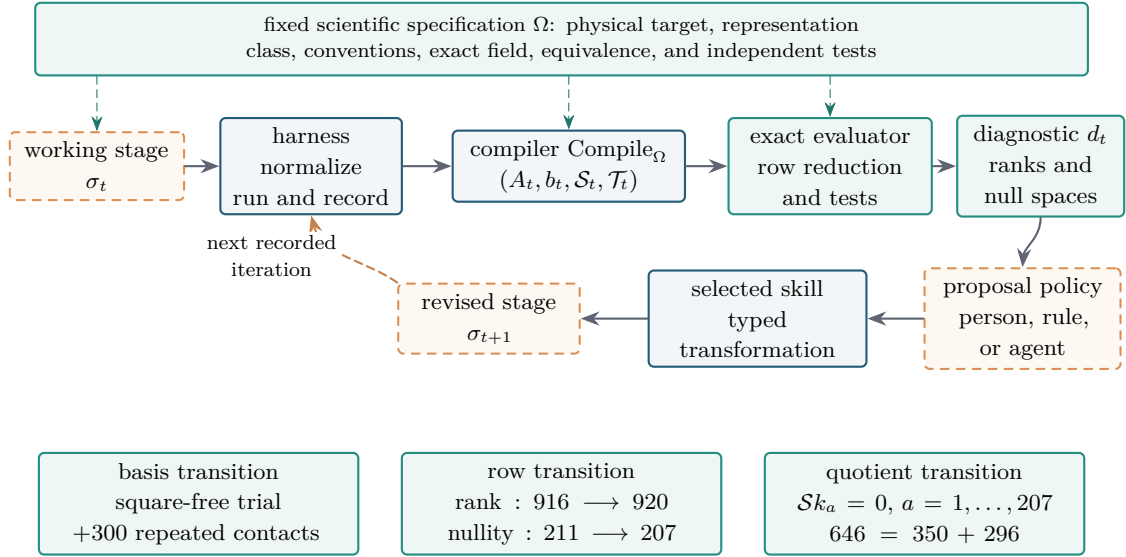

The public implementation provides this common interface.  Its five-point records
contain fixed rules for basis completion and cut-row selection, and the same interface
can receive proposals from a language-model agent.  The compiler, exact evaluator, and
held-out tests then assess the resulting stage.  The concrete schemas, skill contracts,
and record graph are given in
\cref{sec:agent-harness,app:pseudocode,app:evidence-ledger}.

\subsection{Five-point transitions and related work}

The five-point calculation contains three transitions, shown in the lower part of
\cref{fig:llm-agent-skill-harness}.

\paragraph{Completing the numerator basis.}
The first basis trial used the square-free products \(Y_iY_jN_{ij}\) with \(i<j\).
There are \(300\) distinct nonzero quotient directions involving \(Y_i^2\).  Extending
the label range to \(i\leq j\) adds the five repeated-contact tags times the \(60\)
associated tensor structures, hence \(300\) columns.  The construction and count are derived in
\cref{sec:ansatz-compiler}.

\paragraph{Completing the cut measurements.}
Maximal, box, and triple cuts gave rank \(916\) and a \(211\)-dimensional right kernel.
Candidate double-cut rows were evaluated in a fixed round-robin topology schedule.  The
first four independent rows raised the rank to \(920\), so the kernel dimension fell to
\(207\).  Once every physical topology had at least three successful trials, the next
\(64\) rows were dependent and consistent, so the schedule stopped.  The rank gain is the row test
\(\rank(R_{\rm double}K)=4\) in \cref{eq:adaptive-revision-tests}; its physical
visibility mechanism is analyzed in \cref{sec:q-visibility}.

\paragraph{Classifying the remaining freedom.}
Finally, the specified observable map is applied to the complete ordered kernel.  All
\(207\) vectors satisfy \(\cS k_a=0\).  The \(646\) source
conditions split into \(350\) exact coefficient identities and \(296\) fixed-route
scaleless sources, with the LSZ projection recorded separately.  This is a
classification of the final cut-solution fiber, described in
\cref{sec:stage-s,eq:stage-s-row-split,eq:stage-s-direction-theorem}.

Tool-using language-model systems provide general designs for proposal loops.  ReAct
organizes model calls as action--observation iterations, while Toolformer studies how
models select and populate external tool calls \cite{Yao:2023react,Schick:2023toolformer}.
SWE-agent and OpenHands place software tasks inside an agent--computer interface and an
executable environment, and the AI Scientist extends the loop across idea generation,
code execution, experiments, writing, and review
\cite{Yang:2024sweagent,Wang:2025openhands,Lu:2024aiscientist}.

High-energy-physics applications make the role of a scientific harness especially
clear.  Supervisor--coder workflows have been coupled to Snakemake for open-data
analysis; HEPTAPOD exposes schema- and run-card-controlled Monte Carlo operations;
ColliderAgent connects a hierarchical agent layer to an execution backend; and
HepScript supplies a restricted analysis language shared by users and agents
\cite{GendreauDistler:2025llmhep,HEPTAPOD:2025,Qiu:2026collideragent,Jiao:2026hepscript}.
Complete analysis chains have also been exercised on ALEPH, DELPHI, CMS, and BESIII
data, while Collider-Bench compares reproduced collider analyses through event yields
and histograms
\cite{Badea:2026agentic,Moreno:2026jfc,He:2026drsai,Faroughy:2026colliderbench}.

For exact theoretical calculations, executable evaluators can be built from algebraic
reductions and formal checks.  Generalized integration-by-parts equations made the
rapidity-divergent integrals in an \(N^3\)LO quark transverse-parton-distribution
calculation algorithmically reducible \cite{Luo:2019szz}.  FunSearch-based work uses
completed integration-by-parts reductions to evaluate proposed priority functions, and
HepLean provides a machine-checkable setting for formalized high-energy-physics
statements \cite{vonHippel:2025ibpml,Song:2025ibpoptimization,ToobySmith:2025heplean}.
Tree search with automated numerical feedback has also selected analytic methods for a
cosmic-string radiation integral \cite{Brenner:2026aiphysics}.  In the present problem,
the corresponding executable feedback is supplied by the BCJ compiler in
\cref{eq:compiled-factorization} and by the exact null-space diagnostic in
\cref{eq:complete-stage-diagnostic}.

The following sections construct the completed tensor basis, the generalized-cut rows,
the affine solution fiber, and its observable and validation maps.  The provider-neutral
implementation interface is given in \cref{sec:agent-harness}, and the reproducibility
record is described in \cref{sec:reproducibility,app:evidence-ledger}.

\section{Constructing the local numerator basis}
\label{sec:ansatz-compiler}

The columns of the reconstruction matrix define the coefficient directions available
to the compiled model.  Their construction is part of the scientific specification.
This section builds the final local tensor space.  The initial square-free two-contact
layer omits \(300\) repeated-contact directions; adding them gives the final enumerated
basis within the declared parity-even polynomial ring and power-counting bounds.
Graph symmetries and Jacobi maps then turn each retained coefficient into a coordinated
collection of cubic-graph numerators.

\subsection{Contact sectors and repeated-contact terms}

Let
\begin{equation}
  Y_i=\ell_i^2=(\ell+\Delta_i)^2,
  \qquad \Delta_i:=\ell_i-\ell,
  \qquad i=1,\ldots,5,
  \label{eq:contact-generators}
\end{equation}
be the five canonical inverse propagators associated with
\cref{eq:adjacent-routing}.  We write the representative pentagon numerator (the master
numerator) as a sum organized by the number of explicit inverse-propagator factors,
\begin{equation}
  N=N_0+\sum_{i=1}^{5}Y_iN_i
    +\sum_{1\leq i\leq j\leq5}Y_iY_jN_{ij}.
  \label{eq:contact-ansatz}
\end{equation}
This organization records how many propagators a numerator term can cancel.  For example,
inside a pentagon contribution,
\begin{equation}
  \frac{Y_iN_i}{Y_1Y_2Y_3Y_4Y_5}
  =\frac{N_i}{\prod_{j\ne i}Y_j},
  \label{eq:one-contact-cancels-propagator}
\end{equation}
so a term with one factor \(Y_i\) contributes to a box propagator topology.  A product
\(Y_iY_j\) with \(i\ne j\) cancels two distinct propagators and contributes to a
triangle topology.  When
\(i=j\), one copy cancels the propagator and the other remains as an inverse-propagator
factor in the lower-topology numerator.  We therefore include both products on two
different edges and repeated factors on one edge.

The residual mass dimensions and loop-rank bounds in these three layers are
\begin{equation}
  (N_0,N_1,N_2):\qquad (5,3,1).
  \label{eq:layer-power-counting}
\end{equation}
Each inverse propagator \(Y_i\) has mass dimension two.  Since the complete five-point
numerator has dimension five, the factors left in \(N_0\), \(N_i\), and \(N_{ij}\) must
have dimensions five, three, and one, respectively.  The same numbers are imposed as
upper bounds on their residual powers of the loop momentum.  Representative allowed
monomials are, schematically,
\begin{align}
  N_0:&\quad
  s_{12}(\varepsilon_1\mathbin{\cdot}\varepsilon_2)
  (\varepsilon_3\mathbin{\cdot}\varepsilon_4)
  (\varepsilon_5\mathbin{\cdot}\ell)(\ell\mathbin{\cdot}k_1),
  \notag\\
  N_i:&\quad
  (\varepsilon_1\mathbin{\cdot}\varepsilon_2)
  (\varepsilon_3\mathbin{\cdot}\varepsilon_4)
  (\varepsilon_5\mathbin{\cdot}\ell)(\ell\mathbin{\cdot}k_1),
  \notag\\
  N_{ij}:&\quad
  (\varepsilon_1\mathbin{\cdot}\varepsilon_2)
  (\varepsilon_3\mathbin{\cdot}\varepsilon_4)
  (\varepsilon_5\mathbin{\cdot}\ell),
  \label{eq:layer-monomial-examples}
\end{align}
with leg permutations and all other allowed contractions included in the enumeration.
The displayed examples illustrate how the residual dimension decreases from one layer
to the next.

Before algebraic duplicates are removed across the three layers, the enumeration is
\begin{center}
\small
\begin{tabular}{@{}lrrrr@{}}
\toprule
layer & residual dimension & structures per label & inverse-propagator labels & raw total\\
\midrule
\(N_0\) & 5 & 10,724 & 1 & 10,724\\
\(N_1\) & 3 & 1,240 & 5 & 6,200\\
\(N_2\) & 1 & 60 & 15 (\(i\le j\)) & 900\\
\midrule
total &&&& 17,824\\
\bottomrule
\end{tabular}
\end{center}
Algebraic identities can make two expressions in this list represent the same kinematic
tensor.  Removing these duplicates reduces the three layer totals to
\(5124,4700,900\).

Every external polarization appears linearly.  The list contains parity-even scalar
contractions and is constructed for generic \(D_s\).  Its reduction uses generic-
dimensional relations and excludes strictly four-dimensional Levi--Civita and Gram
relations.

\begin{remark}[Completing the repeated-contact basis]
\label{rem:repeated-contact-repair}
The initial two-contact basis used \(i<j\) and was therefore square-free.  In the
chosen polynomial ring and quotient, each direction \(Y_i^2N_{ii}\) represents a
nonzero quotient class distinct from the retained square-free directions.  There are
\(60\) allowed residual structures for each of the five repeated labels, giving
\begin{equation}
  5\times 60=300
  \label{eq:repeated-contact-extension-count}
\end{equation}
additional repeated-contact directions.  We therefore extend the label range from
\(i<j\) to \(i\leq j\), add these \(300\) columns, and rebuild the quotient, graph
descendants, routes, and cut rows.
The resulting two-contact layer contains \(15\times60=900\) structures.

This completes the enumerated two-contact sector within the chosen parity-even polynomial ring,
quotient, and power-counting bounds.  Cut-system consistency is tested after compilation
by evaluating \(\lambda^Tb\) for left-null vectors satisfying \(\lambda^TA=0\).  The use
of exact residual data to extend the candidate space parallels the discrete inverse CHY construction
\cite{Li:2026inverseCHY}.
\end{remark}

The preceding count treats every generated tensor expression as a separate item and
produces
\begin{equation}
  \dim V_{\rm raw}=17{,}824.
  \label{eq:raw-basis-dim}
\end{equation}
Some of these expressions are equal after momentum conservation, transversality, and
the other algebraic relations specified in \cref{sec:problem}.  Let \(\Phi\)
be the exact map that reduces an unreduced expression to a basis modulo those relations.
Its exact rank and nullity are
\begin{equation}
  \rank\Phi=10{,}724,
  \qquad \dim\ker\Phi=7100.
  \label{eq:phi-rank-nullity}
\end{equation}
The standard-basis layer counts are
\begin{equation}
  \dim B_0=5124,
  \qquad \dim B_1=4700,
  \qquad \dim B_2=900.
  \label{eq:standard-layer-counts}
\end{equation}
The quotient is described by two maps.  The projection
\(\pi:V_{\rm raw}\to V_{\rm ansatz}\) removes algebraic duplicates.  The map
\(s:V_{\rm ansatz}\to V_{\rm raw}\) chooses one explicit unreduced expression for each
basis vector.  They obey
\begin{equation}
  \pi\circ s=\id_{V_{\rm ansatz}},
  \label{eq:quotient-section}
\end{equation}
so every reduced vector can be expanded back into the specified monomials.  This chosen
representative makes the expansion unique for the calculation and fixes its algebraic
coordinates.  Physical equivalence is classified later by \(\cS\).

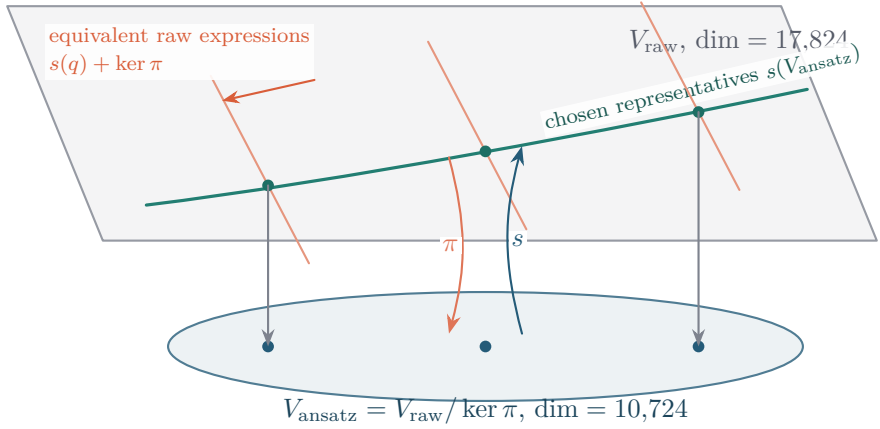
\begin{figure}[H]
\centering
\begin{tikzpicture}[x=1.15cm,y=1cm,line cap=round,line join=round]
  % Low-dimensional projection of the raw presentation space.
  \path[fill=bcjgray!7,draw=bcjgray!65,thick]
    (-.35,1.25)--(8.55,1.25)--(7.45,4.35)--(-1.45,4.35)--cycle;
  \node[anchor=north east,font=\small,bcjgray!90!black] at (8.35,4.18)
    {\(V_{\rm raw}\), \(\dim=17{,}824\)};

  % A transverse deterministic section and several presentation fibers.
  \draw[very thick,bcjcyan!80!black]
    (.15,1.72).. controls (2.35,2.03) and (5.45,2.70) .. (7.75,3.25);
  \node[font=\scriptsize,bcjcyan!55!black,rotate=13,fill=white,
    fill opacity=.90,text opacity=1,inner sep=1pt] at (6.55,3.27)
    {chosen representatives \(s(V_{\rm ansatz})\)};

  \coordinate (sone) at (1.55,1.98);
  \coordinate (stwo) at (4.05,2.43);
  \coordinate (sthree) at (6.50,2.95);
  \foreach \p in {sone,stwo,sthree}{
    \draw[bcjred!65,thick] ($(\p)+(-.66,1.45)$)--($(\p)+(.47,-1.03)$);
    \fill[bcjcyan!70!black] (\p) circle (2.2pt);
  }
  \node[align=left,font=\scriptsize,bcjred,anchor=west,fill=white,
    fill opacity=.90,text opacity=1,inner sep=1.5pt] (fibercallout) at (-1.02,3.75)
    {equivalent raw expressions\\\(s(q)+\ker\pi\)};
  \draw[-{Stealth[length=2.3mm]},bcjred,thick]
    (fibercallout.south east)--(1.02,3.10);

  % Quotient base.
  \fill[bcjblue!8] (4.05,-.15) ellipse[x radius=3.65,y radius=.72];
  \draw[bcjblue!80,thick] (4.05,-.15) ellipse[x radius=3.65,y radius=.72];
  \node[font=\small,bcjblue!75!black] at (4.05,-1.00)
    {\(V_{\rm ansatz}=V_{\rm raw}/\ker\pi\), \(\dim=10{,}724\)};

  \coordinate (qone) at (1.55,-.15);
  \coordinate (qtwo) at (4.05,-.15);
  \coordinate (qthree) at (6.50,-.15);
  \foreach \p in {qone,qtwo,qthree}{\fill[bcjblue] (\p) circle (2.2pt);}
  \draw[-{Stealth[length=2.3mm]},bcjgray!80,thick] (sone)--(qone);
  \draw[-{Stealth[length=2.3mm]},bcjgray!80,thick] (sthree)--(qthree);
  \draw[-{Stealth[length=2.5mm]},bcjblue,thick]
    (4.47,.02) to[bend left=15]
    node[right,font=\small,fill=white,inner sep=1pt] {\(s\)} (4.47,2.50);
  \draw[-{Stealth[length=2.5mm]},bcjred!85,thick]
    (3.63,2.35) to[bend left=15]
    node[left,font=\small,fill=white,inner sep=1pt] {\(\pi\)} (3.63,.02);
\end{tikzpicture}
\caption{Removing algebraically duplicate expressions.  Each red line represents raw
expressions that differ by an identity in \(\ker\pi\) and therefore reduce to the same
blue point.  The map \(s\) chooses one explicit raw expression for each reduced basis
vector, with \(\pi\circ s=\id\).  The drawing shows a low-dimensional slice.}
\label{fig:quotient-section-geometry}
\end{figure}

The declared ansatz is the quotient of the enumerated parity-even scalar contractions
with mass dimension five, at most two explicit inverse-propagator factors, the
loop-momentum bounds in \cref{eq:layer-power-counting}, and one power of every external
polarization.  Parity-odd structures lie outside this ansatz.  These conditions define
the finite tensor family used throughout the reconstruction.  The explicit
enumeration and quotient algorithms are given in \cref{app:pseudocode}, and the exact
dimension checks are indexed in \cref{app:evidence-ledger}.

\subsection{Symmetry reduction and graph descendants}

Rotations and reflections relate different labelings of the same pentagon ansatz.
After relabeling momenta and polarizations, the related expressions share one
coefficient; a reflection can also introduce a minus sign because it reverses cubic
vertex orientations.  An \emph{orbit} is the complete group of basis expressions tied
together in this way.  One coefficient represents each signed orbit, incorporating
all of its symmetry relations directly into the variables.

The pentagon master numerator is invariant under cyclic relabeling and changes sign
under the reflection implied by cubic-vertex antisymmetry.  Let
\(D_5=\langle C,R\mid C^5=R^2=1,RCR=C^{-1}\rangle\), and let \(\rho\) denote this signed
action on the standard tensor basis.  The sign assigned to each symmetry operation is a
sign rule, or in group-theory terms a one-dimensional character,
\begin{equation}
  \chi(C^k)=+1,\qquad \chi(RC^k)=-1,
  \label{eq:required-pentagon-character}
\end{equation}
because a rooted reflection reverses the five cubic vertices.  We identify all basis
elements related by this action and keep one coefficient for each signed orbit.  If an
expression is mapped to itself with a minus sign, its coefficient must be zero.  These
two steps reduce the number of unknown coefficients to
\begin{equation}
  \dim V_{\rm orbit}=1127.
  \label{eq:orbit-dimension}
\end{equation}
The symmetry analysis also finds \(53\) additional candidate orbits that map to their
own negatives and therefore have zero coefficient.  The dimension in
\cref{eq:orbit-dimension} also follows from the character projector
\begin{equation}
  P_\chi=\frac1{10}\sum_{g\in D_5}\chi(g)\rho(g).
  \label{eq:signed-character-projector}
\end{equation}
The signed traces are \(\operatorname{Tr}\rho(1)=10{,}724\),
\(\operatorname{Tr}\rho(C^k)=4\) for the four nontrivial rotations, and
\(\operatorname{Tr}\rho(RC^k)=-106\) for the five reflections.  Hence
\begin{equation}
  \dim\operatorname{im} P_\chi
  =\frac{10{,}724+4\cdot4-5(-106)}{10}
  =1127.
  \label{eq:signed-burnside-dimension}
\end{equation}
In this formula, the factor \(-1\) from \(\chi(RC^k)\) multiplies the trace
\(-106\) of the signed reflection action.  The \(53\) forced-zero coordinates come from
orbits that contain an operation mapping a basis element to its own negative.

After these symmetry identifications, the master ansatz is
\begin{equation}
  N_{1|\rho}(\ell;x)=\sum_{a=1}^{1127}x_a\,
  B_{a,1|\rho}(\ell),
  \qquad B_{a,1|\rho}\in V_{\rm ansatz}.
  \label{eq:orbit-master-ansatz}
\end{equation}
We fix an order for these \(1127\) coefficients.  The same order is used for the columns
of the cut equations and for the kernel vectors reported below.

The symmetry-reduced master ansatz determines the remaining graph numerators.  A
free-Lie bracket is a compact way to record the antisymmetric combination associated
with a cubic subgraph.  Starting from the pentagon master numerators, bracket trees
construct the numerator of every related cubic graph.  For example,
\begin{equation}
  N_{[1,2]|3|4|5}(\ell)
  =N_{1|2,3,4,5}(\ell)-N_{2|1,3,4,5}(\ell),
  \label{eq:simple-descendant}
\end{equation}
Further brackets are built recursively.  The second term begins with leg 2, whereas the
master convention roots the ordering at leg 1.  A cyclic or reflected relabeling returns
it to that convention, accompanied by the loop-momentum shift that restores the
canonical route.

Antisymmetry and Jacobi follow from
\begin{equation}
  [a,b]=-[b,a],
  \qquad [[a,b],c]+[[b,c],a]+[[c,a],b]=0.
  \label{eq:free-lie-relations}
\end{equation}
The word expansion is combined with the routing shift \(\tau_\Delta\).  Formally,
\begin{equation}
  J:\{N_{1|\rho}\}_{\rho\in S_4}
  \longrightarrow\{n_g\}_{g\in\Gamma^{(1)}_{5,3}},
  \qquad
  n_g=\sum_{\rho}J_{g\rho}[\tau_{\Delta_{g\rho}}]N_{1|\rho}.
  \label{eq:free-lie-compiler}
\end{equation}

\begin{proposition}[Jacobi identities for all constructed graph numerators]
If all graph numerators are produced by the free-Lie and routing map
\eqref{eq:free-lie-compiler}, every Jacobi triple in the graph list obeys
\(n_i+n_j+n_k=0\) identically in the ansatz coefficients and loop momentum.
\end{proposition}

\begin{proof}[Proof sketch]
A graph Jacobi triple differs only in the bracketing of one local four-point subtree.
Its signed word expansion is the free-Lie Jacobi relation.  The graph list assigns a
common canonical loop route, so matching words receive matching affine shifts.  The
alternating sum therefore vanishes before any kinematic value is substituted.
\end{proof}

Thus every coefficient vector in the ansatz satisfies the Jacobi identities before the
sampled kinematic rows are constructed.

The complete graph list contains \(297\) cubic graphs.  Their color factors span \(12\)
independent ring orderings, in which the external legs are arranged cyclically around
the loop.  The map from graph color factors to these ring orderings has rank \(12\), and
the subsequent map to the trace basis also has rank \(12\).  Graph antisymmetry, routed
Jacobi relations, cyclic and reflected relabelings, and crossings are preserved through
these changes of representation.

Schematically,
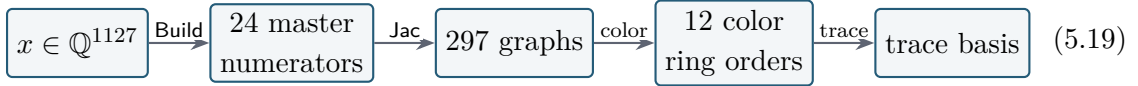
\begin{equation}
\begin{tikzpicture}[baseline=-2pt,node distance=8mm]
  \node[flow,minimum width=18mm] (x) {\(x\in\Q^{1127}\)};
  \node[flow,right=of x,minimum width=20mm] (m) {24 master\\numerators};
  \node[flow,right=of m,minimum width=20mm] (g) {297 graphs};
  \node[flow,right=of g,minimum width=20mm] (r) {12 color\\ring orders};
  \node[flow,right=of r,minimum width=20mm] (t) {trace basis};
  \draw[arrow] (x)-- node[above=1pt,font=\scriptsize,fill=white,inner sep=.7pt]
    {\(\mathsf{Build}\)} (m);
  \draw[arrow] (m)-- node[above=1pt,font=\scriptsize,fill=white,inner sep=.7pt]
    {\(\mathsf{Jac}\)} (g);
  \draw[arrow] (g)-- node[above=1pt,font=\scriptsize,fill=white,inner sep=.7pt]{color} (r);
  \draw[arrow] (r)-- node[above=1pt,font=\scriptsize,fill=white,inner sep=.7pt]{trace} (t);
\end{tikzpicture}
\label{eq:graph-color-chain}
\end{equation}
This chain separates the kinematic construction from the later choice of color basis.
Its explicit graph registry, routing rules, and exact checks are described in
\cref{app:pseudocode,app:evidence-ledger}.

The resulting \(1127\)-component vector now parametrizes every graph numerator.  The next
section evaluates those graphs on generalized cuts and turns each evaluation into a row
of the constraint matrix.

\section{Adaptive row design: cuts and topology-dependent visibility}
\label{sec:cuts-and-visibility}
\label{sec:cut-bootstrap}

\subsection{Physical cut targets and the measurement filtration}

Once the coefficient space has been fixed, each row specifies a measurement of that
space.  The row-design problem is to add physically defined measurements that act on
directions still invisible to the current system.  For each generalized cut, the target
is the physical-state sum
\begin{equation}
  b_C=\sum_{\{h_e\}}\prod_{v\in C}
  A_v^{(0)}(\{k,\varepsilon\}_v;\{h_e\}),
  \label{eq:tree-side-target}
\end{equation}
evaluated on exact on-shell kinematics with a \(D_s\)-dimensional physical-state
projector.  The tree amplitudes obey their Ward identities, so the target depends only
on physical internal states.\footnote{The calculation evaluates the same target by
Berends--Giele recursion and by an explicit color-ordered vertex expansion.  The two
constructions and their exact comparison are recorded in
\cref{app:pseudocode,app:evidence-ledger}.}

Putting a propagator on shell removes it as an off-shell connection and exposes tree
amplitudes on the two sides of the cut.  We use four groups of equations, denoted by
\(M,B,T,D\).  They put five, four, three, and two propagators on shell, respectively;
the letters stand for maximal, box, triple, and double cuts.  Maximal cuts constrain the
terms with no explicit inverse propagator.  Cuts with fewer on-shell propagators retain
more contact terms and can therefore constrain coefficient combinations that vanish on
maximal cuts.  The linear system is built by adding these four groups in that order.

The way cut depth exposes successive contact layers can be seen directly from one
canonical pentagon contribution.  Suppressing graph labels, write
\begin{equation}
  \I_{\rm pent}(\ell)
  =\frac{N_0+\sum_iY_iN_i+\sum_{i\le j}Y_iY_jN_{ij}}
  {Y_1Y_2Y_3Y_4Y_5}.
  \label{eq:pentagon-contact-integrand-example}
\end{equation}
On the fivefold maximal cut \(C_M\), multiply by all five propagators and impose
\(Y_1=\cdots=Y_5=0\).  Every explicit contact factor vanishes, leaving
\begin{equation}
  \Res_{C_M}\I_{\rm pent}=N_0(\ell^{(q)}),
  \qquad
  b_{C_M}^{(q)}=\sum_{\rm states}\prod_{v=1}^{5}A^{\rm tree}_{3,v}.
  \label{eq:explicit-maximal-cut-row}
\end{equation}
Here \(\Res_C\) means that the cut propagators have been multiplied out and their
on-shell conditions imposed, and \(\ell^{(q)}\) is one exact solution of those
conditions.  The full color-dressed equation sums all graphs compatible with the same
cut, including their signs and loop-momentum routes.

Now cut \(Y_1,\ldots,Y_4\) and keep \(Y_5\) uncut.  Direct cancellation in
\cref{eq:pentagon-contact-integrand-example} gives
\begin{equation}
  \Res_{C_B}\I_{\rm pent}
  =\left[\frac{N_0}{Y_5}+N_5+Y_5N_{55}\right]_{Y_1=\cdots=Y_4=0}.
  \label{eq:explicit-box-contact-row}
\end{equation}
The box equation therefore contains both the single contact on the uncut edge and the
repeated contact \(Y_5^2N_{55}\), while the maximal equation contains neither.
Expanding each displayed \(N\) in the orbit basis gives, for example,
\begin{equation}
  A_{(C_B,q),a}
  =\left[\frac{B_{0a}}{Y_5}+B_{5a}+Y_5B_{55,a}\right]_{\ell=\ell^{(q)}}
  \label{eq:explicit-box-matrix-entry}
\end{equation}
where \(B_{0a},B_{5a},B_{55,a}\) are the functions multiplying the unknown coefficient
\(x_a\) in the three corresponding numerator layers.  The expression becomes a known exact
matrix entry, and the product of trees on that box cut supplies the
corresponding \(b_{C_B}^{(q)}\).  This is the loop-level version of the equation-row
construction in \cref{eq:fourpoint-two-column-row-example}.

Rows are grouped by topology and introduced cumulatively:
\begin{equation}
  \cC_{\le M}\subset\cC_{\le B}\subset
  \cC_{\le T}\subset\cC_{\le D}.
  \label{eq:cut-filtration}
\end{equation}
Let \(A_{\le q}\) denote the matrix containing all equations through cut group \(q\).
For a numerator subspace \(W\subseteq V_{\rm orbit}\), define the cumulative visible
dimension
\begin{equation}
  \nu_W(q)=\rank(A_{\le q}|_W).
  \label{eq:visibility-barcode-general}
\end{equation}
The first cut group at which \(\nu_W\) grows is the first one that constrains a component
of \(W\).  Thus the question of which cuts constrain a numerator structure is answered
by an exact rank calculation.

For the full system, each added cut group reduces the subspace left unconstrained by the
previous equations, as shown in \cref{fig:cut-filtration-geometry}.

\begin{figure}[t]
\centering
\begin{tikzpicture}[
  x=1.08cm,y=.90cm,line cap=round,line join=round,
  every node/.style={font=\small}]
  % A schematic tube whose cross-sections are the nested right kernels.
  \path[fill=bcjgray!7]
    (0,2.35).. controls (.75,2.16) and (1.38,1.84) .. (2.15,1.67)
    -- (2.15,-1.67).. controls (1.38,-1.84) and (.75,-2.16) .. (0,-2.35)--cycle;
  \path[fill=bcjblue!9]
    (2.15,1.67).. controls (2.85,1.56) and (3.58,1.36) .. (4.30,1.25)
    -- (4.30,-1.25).. controls (3.58,-1.36) and (2.85,-1.56) .. (2.15,-1.67)--cycle;
  \path[fill=bcjcyan!11]
    (4.30,1.25).. controls (5.00,1.18) and (5.73,1.02) .. (6.45,.96)
    -- (6.45,-.96).. controls (5.73,-1.02) and (5.00,-1.18) .. (4.30,-1.25)--cycle;
  \path[fill=bcjgold!18]
    (6.45,.96).. controls (7.14,.95) and (7.86,.94) .. (8.60,.94)
    -- (8.60,-.94).. controls (7.86,-.94) and (7.14,-.95) .. (6.45,-.96)--cycle;

  \draw[bcjgray!55,thick]
    (0,2.35).. controls (.75,2.16) and (1.38,1.84) .. (2.15,1.67)
    .. controls (2.85,1.56) and (3.58,1.36) .. (4.30,1.25)
    .. controls (5.00,1.18) and (5.73,1.02) .. (6.45,.96)
    .. controls (7.14,.95) and (7.86,.94) .. (8.60,.94);
  \draw[bcjgray!55,thick]
    (0,-2.35).. controls (.75,-2.16) and (1.38,-1.84) .. (2.15,-1.67)
    .. controls (2.85,-1.56) and (3.58,-1.36) .. (4.30,-1.25)
    .. controls (5.00,-1.18) and (5.73,-1.02) .. (6.45,-.96)
    .. controls (7.14,-.95) and (7.86,-.94) .. (8.60,-.94);

  \fill[bcjgray!16] (0,0) ellipse[x radius=.23,y radius=2.35];
  \draw[bcjgray!75,thick] (0,0) ellipse[x radius=.23,y radius=2.35];
  \fill[bcjblue!27] (2.15,0) ellipse[x radius=.22,y radius=1.67];
  \draw[bcjblue!80,thick] (2.15,0) ellipse[x radius=.22,y radius=1.67];
  \fill[bcjcyan!31] (4.30,0) ellipse[x radius=.21,y radius=1.25];
  \draw[bcjcyan!70!black,thick] (4.30,0) ellipse[x radius=.21,y radius=1.25];
  \fill[bcjgold!42] (6.45,0) ellipse[x radius=.20,y radius=.96];
  \draw[bcjgold!55!black,thick] (6.45,0) ellipse[x radius=.20,y radius=.96];
  \fill[bcjred!28] (8.60,0) ellipse[x radius=.20,y radius=.94];
  \draw[bcjred!85,very thick] (8.60,0) ellipse[x radius=.20,y radius=.94];
  \draw[-{Stealth[length=2.8mm]},bcjgray!75,thick] (-.38,0)--(9.05,0);

  % Draw the magnifier guides before the rank arrows, so the guides recede behind
  % the data-bearing marks rather than cutting through them.
  \draw[bcjgray!22,densely dotted] (6.45,.70)--(9.47,1.58);
  \draw[bcjgray!22,densely dotted] (8.60,.70)--(9.47,.36);

  \draw[-{Stealth[length=2.2mm]},bcjblue,thick]
    (.42,2.64)--(1.78,2.05)
    node[midway,above=2pt,sloped,font=\scriptsize] {\(M:\ \Delta r=547\)};
  \draw[-{Stealth[length=2.2mm]},bcjcyan!70!black,thick]
    (2.48,1.92)--(4.00,1.48)
    node[midway,above=2pt,sloped,font=\scriptsize] {\(B:\ \Delta r=265\)};
  \draw[-{Stealth[length=2.2mm]},bcjgold!55!black,thick]
    (4.62,1.46)--(6.14,1.15)
    node[midway,above=2pt,sloped,font=\scriptsize] {\(T:\ \Delta r=104\)};
  \draw[-{Stealth[length=2.2mm]},bcjred,thick]
    (6.78,1.16)--(8.28,1.11)
    node[midway,above=4pt,sloped,font=\scriptsize] {\(D:\ \Delta r=4\)};

  \node[align=center,font=\scriptsize] at (0,-2.72)
    {\(V_{\rm orbit}\)\\\(n=1127\)};
  \node[align=center,font=\scriptsize,bcjblue!85!black] at (2.15,-1.98)
    {\(\mathcal N_M\)\\\(n=580\)};
  \node[align=center,font=\scriptsize,bcjcyan!55!black] at (4.30,-1.56)
    {\(\mathcal N_B\)\\\(n=315\)};
  \node[align=center,font=\scriptsize,bcjgold!42!black] at (6.45,-1.30)
    {\(\mathcal N_T\)\\\(n=211\)};
  \node[align=center,font=\scriptsize,bcjred!90!black] at (8.60,-1.28)
    {\(\mathcal N_D\)\\\(n=207\)};

  % Magnify the small but exact T-to-D contraction.
  \fill[white] (10.28,.98) circle (1.18);
  \draw[bcjgray!55,thick] (10.28,.98) circle (1.18);
  \fill[bcjgold!24] (10.28,.98) ellipse[x radius=.82,y radius=.55];
  \draw[bcjgold!55!black,thick] (10.28,.98) ellipse[x radius=.82,y radius=.55];
  \fill[bcjred!24] (10.28,.98) ellipse[x radius=.70,y radius=.45];
  \draw[bcjred!85,very thick] (10.28,.98) ellipse[x radius=.70,y radius=.45];
  \node[font=\scriptsize,align=center] at (10.28,2.42)
    {\(T\to D\) detail};
  \node[font=\scriptsize,align=center,bcjred!90!black] at (10.28,.98)
    {\(\Delta r=4\)};
  \node[font=\scriptsize,align=center,bcjred!90!black] at (10.28,-.42)
    {directions resolved};

  \node[font=\scriptsize,anchor=west] at (1.10,-3.38)
    {\(\mathcal N_D\subset\mathcal N_T\subset\mathcal N_B\subset\mathcal N_M
       \subset V_{\rm orbit},\qquad \mathcal N_q:=\ker A_{\le q}\)};
\end{tikzpicture}
\caption{How successive cut groups reduce the unconstrained coefficient space.  Adding
\(M/B/T/D\) equations raises the cumulative rank
\(0\to547\to812\to916\to920\) and contracts the nullity
\(1127\to580\to315\to211\to207\).  The magnified final slice emphasizes that the
double-cut group constrains four additional directions.  Cross-sectional areas indicate
subspace inclusion in a schematic two-dimensional projection.}
\label{fig:cut-filtration-geometry}
\end{figure}
The final group adds four pivot directions and reduces the nullity from \(211\) to
\(207\).

\subsection{Information gain on the current right kernel}

For a system \(Ax=b\) over \(\F_p\), the ranks of \(A\) and \([A|b]\) distinguish
inconsistency, residual freedom, and uniqueness.  The right null space records
coefficient combinations invisible to the accumulated cuts, while the left null space
records compatibility conditions among the cut targets:

\begin{center}
\begin{tabularx}{\textwidth}{@{}lX@{}}
\toprule
algebraic result & meaning \\ \midrule
\(\rank A<\rank[A|b]\) & a vector \(\lambda\) exists with
\(\lambda^TA=0\) and \(\lambda^Tb\ne0\); the target lies outside the column space of
the chosen ansatz; \\
\(\rank A=\rank[A|b]<n\) & the equations define an affine solution family whose
directions are the vectors in \(\ker A\); \\
\(\rank A=\rank[A|b]=n\) & the coefficient vector is uniquely determined. \\
\bottomrule
\end{tabularx}
\end{center}

These null spaces also identify which change can add genuinely new information.  If
\(\lambda^TA=0\) and \(\lambda^Tb\ne0\), a candidate basis element gives a new column
\(c\), and its overlap with the obstruction is
\begin{equation}
  s_{\rm col}(c)=\lambda^Tc.
  \label{eq:obstruction-column-score}
\end{equation}
A nonzero value changes the incompatible combination of equations.  Conversely, let
the columns of \(K\) form an ordered basis of \(\ker A\), and let \(R\) be a block of
candidate measurement rows.  The restriction \(RK\) records the response of every
candidate measurement on every currently blind coefficient direction.  Its exact
information gain is
\begin{equation}
  \Delta_{\rm info}(R\mid A)
  :=\rank(RK)
  =\dim\ker A-\dim(\ker A\cap\ker R).
  \label{eq:kernel-row-score}
\end{equation}
For a single row this rank is zero or one; for a cut family it counts the independent
blind directions resolved by the whole block.  Left-null information repairs the
column model, while right-kernel information designs the measurements.  Directions
remaining after the specified cut-topology coverage is complete are passed to the
observable quotient \(\cS\).  The analogous use of exact residual information in inverse
CHY constructions is discussed in \cite{Li:2026inverseCHY}.

For the present system, maximal, box, and triple cuts give a matrix \(A_{\le T}\) of
rank \(916\) with a \(211\)-dimensional kernel.  The candidate double-cut block obeys
\begin{equation}
  \Delta_{\rm info}(R_D\mid A_{\le T})
  =\rank(R_DK_{\le T})=4.
  \label{eq:double-cut-information-gain}
\end{equation}
After those four directions are resolved, the selected maximal, box, triple, and double
coverage has stable rank \(920\) and a \(207\)-dimensional right kernel.

Independent validation evaluates the resulting affine family on new cut equations.  If
\((A_h,b_h)\) denotes such an equation, the two conditions are
\begin{equation}
  A_h\bigl(x^{(0)}+D_sx^{(1)}\bigr)=b_h,
  \qquad A_hk_a=0\quad(a=1,\ldots,207).
  \label{eq:sealed-validation-equations}
\end{equation}
The first condition checks the particular solution and the second checks the complete
right kernel.  Equivalently, both residuals
\(A_h(x^{(0)}+D_sx^{(1)})-b_h\) and \(A_hk_a\) vanish on independently generated
equations.  The arithmetic and
row-building algorithm are described in
\cref{app:exact-arithmetic,app:pseudocode}; the validation-record index and its
v0.1/v0.2 distribution boundary are given in
\cref{app:evidence-ledger}.  The following worked
example isolates a small tensor sector in which the need for a new cut topology can be
derived directly.

\subsection{A three-polarization example of topology-dependent visibility}
\label{sec:q-visibility}

The cut system constructed above contains directions that different topologies measure
differently.  The complete solve in \cref{sec:reconstruction} leaves a
\(207\)-dimensional kernel, and \cref{sec:stage-s} classifies its observable content.
Before giving that full solution, we isolate a smaller example in which the origin of a
cut kernel can be seen directly.  We restrict to tensors containing one loop-momentum
contraction with each of three external polarizations and call this the
three-polarization \(Q\)-sector.  A route reflection splits the sector into even and odd
parts, and the corner structure of a cut determines which part it measures.

For the five loop-momentum routes, define
\begin{equation}
  A_i=\ell_i\cdot\varepsilon_1,
  \qquad B_i=\ell_i\cdot\varepsilon_2,
  \qquad C_i=\ell_i\cdot\varepsilon_3,
  \qquad i=1,\ldots,5.
  \label{eq:q-routed-variables}
\end{equation}
The \(5^3=125\) products \(A_aB_bC_c\) obey \(61\) linear relations from the on-shell,
momentum-conservation, and transversality conditions.  Hence
\begin{equation}
  125_{\rm routed}=61_{\rm relations}+64_{\rm independent}.
  \label{eq:q-125-balance}
\end{equation}

An adjacent-route reflection \(\sigma\) acts on each local four-dimensional route space
with signature \(3_+\oplus1_-\).  More generally, for an \(m\)-fold tensor product of a
reflection with signature \((r-1)_+\oplus1_-\), the two parity dimensions are
\begin{equation}
  \dim V_\pm=\frac{r^m\pm(r-2)^m}{2}.
  \label{eq:reflection-parity-count}
\end{equation}
Setting \(r=4\) and \(m=3\) gives
\begin{equation}
  64=36_{\sigma+}\oplus28_{\sigma-}.
  \label{eq:q-even-odd-split}
\end{equation}
This decomposition is fixed by reflection symmetry before any cut equation is
evaluated.  A useful set of reflection-odd covariants begins with
\begin{align}
  \delta_A&=2A_5-A_1-A_4=(k_4-k_5)\cdot\varepsilon_1,\nonumber\\
  \delta_B&=2B_5-B_1-B_4,\qquad
  \delta_C=2C_5-C_1-C_3,
  \label{eq:q-deltas}
\end{align}
where the external forms of \(\delta_B\) and \(\delta_C\) follow by relabeling and
transversality.  Each \(\delta_X\) changes sign under \(\sigma\) and is independent of
the loop momentum.

On an adjacent box cut, descendant graphs occur in reflection pairs with equal cut
weight.  The cut map therefore factors through the even projector:
\begin{equation}
  R_{\Box}=M_{\Box}(1+\sigma)
  =2M_{\Box}P_+,
  \qquad P_+=\frac{1+\sigma}{2}.
  \label{eq:box-even-factorization}
\end{equation}
It follows that
\begin{equation}
  R_{\Box}|_{V_{\sigma-}}=0,
  \qquad \rank R_{\Box}\le36.
  \label{eq:box-odd-blind}
\end{equation}
The exact rank is \(36\): the box cuts measure the entire even subspace and none of the
\(28\) odd directions.

The balanced triple cuts have corner-size patterns
\begin{equation}
  (1,2,2),\qquad(2,1,2),\qquad(2,2,1).
  \label{eq:balanced-triples}
\end{equation}
At a corner containing two external legs, the legs either lie along a cubic path or
first join in a bracket.  For a reflection-odd tensor \(Q^-\), these two local forms have
equal spectator and state-sum factors and opposite signs:
\begin{equation}
  R_{\rm path}^{de}[Q^-]
  =+\eta_\alpha(\varepsilon_d\cdot\varepsilon_e)Q^-_{de},
  \qquad
  R_{\rm bracket}^{de}[Q^-]
  =-\eta_\alpha(\varepsilon_d\cdot\varepsilon_e)Q^-_{de}.
  \label{eq:spectator-local-pair}
\end{equation}
The pair cancels locally, before the other corners are summed.  Every nonzero odd
placement in a balanced triple contains such a two-leg corner, and therefore
\begin{equation}
  R_{\rm balanced}|_{V_{\sigma-}}=0,
  \qquad R_{\rm balanced}=R_{\rm balanced}P_+.
  \label{eq:balanced-odd-blind}
\end{equation}
Every box and balanced-triple row factors through the same even projector.  Additional
rows from these families therefore leave the odd kernel unchanged.

The odd directions become visible when the cut contains a three-leg corner, where the
local path--bracket pairing above is absent.  The
anchored filtration contains the following \(13\) topologies:
\begin{center}
\begin{tabular}{@{}cll@{}}
\toprule
indices & family & corner sizes\\ \midrule
0 & maximal & \((1,1,1,1,1)\)\\
1--4 & adjacent boxes & cyclic anchored placements of \((2,1,1,1)\)\\
5--7 & balanced triples & \((1,2,2),(2,1,2),(2,2,1)\)\\
8--10 & three-leg-corner triples & \((1,1,3),(1,3,1),(3,1,1)\)\\
11--12 & doubles & \((2,3),(3,2)\)\\
\bottomrule
\end{tabular}
\end{center}
In this order, the cumulative rank is
\begin{equation}
  (0,36,36,36,36,36,36,36,64,64,64,64,64),
  \label{eq:q-visibility-barcode}
\end{equation}
with the refined parity ranks
\begin{equation}
  \nu_{\sigma+}=(0,36,\ldots,36),
  \qquad
  \nu_{\sigma-}=(0,\ldots,0,28,\ldots,28).
  \label{eq:q-refined-barcodes}
\end{equation}
The first three-leg-corner triple has rank \(28\) on the odd subspace and therefore
measures every direction that the preceding cut families miss.

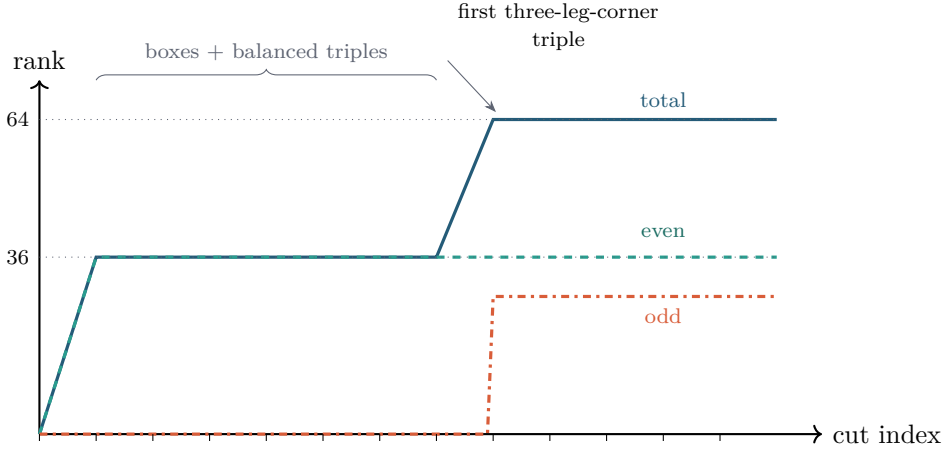
\begin{figure}[H]
\centering
\begin{tikzpicture}[x=.75cm,y=.065cm]
  \draw[->,thick] (0,0)--(13.8,0) node[right,font=\small]{cut index};
  \draw[->,thick] (0,0)--(0,72) node[above,font=\small]{rank};
  \draw[very thick,bcjblue]
    (0,0)--(1,36)--(7,36)--(8,64)--(13,64);
  \draw[very thick,bcjcyan,dashed]
    (0,0)--(1,36)--(13,36);
  \draw[very thick,bcjred,dash dot]
    (0,0)--(7.9,0)--(8,28)--(13,28);
  \foreach \x in {0,...,12}{\draw (\x,0)--++(0,-1.4);}
  \draw[bcjgray,dotted] (0,36)--(13,36);
  \draw[bcjgray,dotted] (0,64)--(13,64);
  \node[left,font=\scriptsize] at (0,36) {36};
  \node[left,font=\scriptsize] at (0,64) {64};
  \node[font=\scriptsize,bcjblue] at (11,68) {total};
  \node[font=\scriptsize,bcjcyan!70!black] at (11,41) {even};
  \node[font=\scriptsize,bcjred] at (11,24) {odd};
  \draw[decorate,decoration={brace,amplitude=4pt},bcjgray]
    (1,72)--(7,72) node[midway,above=3pt,font=\scriptsize]{boxes + balanced triples};
  \node[font=\scriptsize,align=center,anchor=south] (triplelabel) at (9.15,76)
    {first three-leg-corner\\triple};
  \draw[-{Stealth[length=1.6mm]},thin,bcjgray]
    (triplelabel.south west) -- (8.05,65.5);
\end{tikzpicture}
\caption{Cumulative rank in the three-polarization \(Q\)-sector.  The flat segment
reflects the fact that box and balanced-triple rows act only on the reflection-even
subspace.}
\label{fig:q-visibility}
\end{figure}

The fixed-routing representative gives an explicit reflection-odd tensor on a
\(22\)-monomial support, denoted by \(Q_{\mathrm{FR}}\).  On this support,
\begin{equation}
  \rank(1+\sigma)=12,
  \qquad \dim\ker(1+\sigma)=10.
  \label{eq:q-support-rank}
\end{equation}
A convenient odd basis is
\begin{align}
 &\delta_A B_1C_3,\ \delta_A B_2C_1,\ \delta_A B_4C_3,\
   \delta_A B_2C_3,\nonumber\\
 &\delta_B A_1C_3,\ \delta_B A_4C_1,\nonumber\\
 &\delta_C A_1B_1,\ \delta_C A_1B_4,\ \delta_C A_1B_2,\
   \delta_C A_4B_2.
  \label{eq:q-ten-odd-basis}
\end{align}
The first odd-visible triple fixes the coefficients
\begin{equation}
  (1,1,-1,1,-1,-1,-1,1,1,1),
  \label{eq:q-ten-coefficients}
\end{equation}
and hence
\begin{align}
  Q_{\mathrm{FR}}={}&
  \delta_A(B_1C_3+B_2C_1-B_4C_3+B_2C_3)
  -\delta_B(A_1C_3+A_4C_1)\nonumber\\
  &+\delta_C(-A_1B_1+A_1B_4+A_1B_2+A_4B_2),
  \label{eq:q-strict-bracket}
\end{align}
with \(\sigma Q_{\mathrm{FR}}=-Q_{\mathrm{FR}}\).  Its leading loop-quadratic term is
\begin{equation}
  Q_{\mathrm{FR}}\big|_{\ell^2}
  =2\left[
  \delta_A(\ell\cdot\varepsilon_2)(\ell\cdot\varepsilon_3)
  -\delta_B(\ell\cdot\varepsilon_1)(\ell\cdot\varepsilon_3)
  +\delta_C(\ell\cdot\varepsilon_1)(\ell\cdot\varepsilon_2)
  \right].
  \label{eq:q-leading-covariant}
\end{equation}

This sector gives a complete local model of topology-dependent visibility.  Symmetry
identifies a \(28\)-dimensional space annihilated by the box and balanced-triple rows,
and the three-leg-corner topology supplies a rank-\(28\) measurement of that space.  In
the full reconstruction, the same test restricts a candidate cut block \(R\) to the
current kernel \(K\).  A topology has positive information gain precisely when
\(\rank(RK)>0\).  A family that vanishes on the symmetry-protected kernel contributes
zero information gain at every additional sample, so the next informative measurement
uses a different topology.  Applying this test to the full
\(1127\)-coordinate system produces the cut matrix whose affine solution is derived in
the next section.

\section{The exact solution fiber and observable quotient}
\label{sec:solution-equivalence}
\label{sec:reconstruction}

\subsection{Obstructions, deformations, and the exact solution fiber}

Set \(\F=\Q(D_s)\), and extend the symmetry-reduced coefficient space to
\begin{equation}
  V=V_{\rm orbit}\otimes_{\Q}\F\cong\F^{1127}.
\end{equation}
Let \(W\) be the \(\F\)-vector space containing the collected maximal, box, triple,
and double-cut data.  The cut construction is a linear map
\begin{equation}
  A:V\longrightarrow W,
  \qquad x\longmapsto Ax,
  \label{eq:cut-linear-map}
\end{equation}
and the inverse problem is to find a preimage of the specified data vector \(b\in W\).

\begin{definition}[Obstruction and deformation spaces]
\label{def:inverse-obstruction-deformation}
For a linear inverse problem \(Ax=b\), define
\begin{equation}
  \mathrm{Ob}(A):=\operatorname{coker}A
     =W/\operatorname{im}A,
  \qquad
  \mathrm{Def}(A):=\ker A.
  \label{eq:obstruction-deformation-spaces}
\end{equation}
The obstruction carried by the data is the class
\(\operatorname{ob}_A(b):=[b]\in\mathrm{Ob}(A)\).  It vanishes exactly when the
inverse problem has a solution.  Once a solution exists, \(\mathrm{Def}(A)\) is the
vector space of changes that preserve all cut data.
\end{definition}

The same definitions can be organized as the cohomology of a two-term complex.  In
that language, adding ansatz columns and adding measurement rows produce two exact
sequences whose connecting maps are precisely the obstruction-repair and
kernel-reduction operators used in the construction.  A self-contained account is
given in \cref{app:homological-view}.

\begin{theorem}[Exact cut solution]
For the coefficient space obtained from the ansatz of \cref{sec:ansatz-compiler} after
cyclic and reflection symmetries are imposed, the exact cut system obeys
\begin{equation}
  \rank A=920,
  \qquad \dim\ker A=207,
  \qquad \rank A=\rank[A|b].
  \label{eq:rank-nullity-theorem}
\end{equation}
Consequently,
\begin{equation}
  \operatorname{ob}_A(b)=0,
  \qquad
  \dim_{\F}\mathrm{Def}(A)=207.
  \label{eq:exact-obstruction-deformation-result}
\end{equation}
\end{theorem}

The vanishing obstruction establishes existence.  Choosing one particular solution
\(x_{\rm p}(D_s)\), the deformation space gives the complete solution fiber,
\begin{equation}
\begin{aligned}
  \mathcal F_b(D_s)
  &:=A^{-1}(b)
    =\{x\in\Q(D_s)^{1127}:Ax=b\},\\
  &=x_{\rm p}(D_s)+\mathrm{Def}(A)\\
  &=x_{\rm p}(D_s)+\Span_{\F}\{k_1,\ldots,k_{207}\},
  \qquad \dim_{\F}\mathcal F_b=207.
\end{aligned}
  \label{eq:solution-affine-space}
\end{equation}
Non-uniqueness is a structural feature of inverse representation problems.  The inverse
CHY construction can admit several integral edge assignments realizing the same channel
data \cite{Li:2026inverseCHY}; here the corresponding freedom is linear and is resolved
into one particular solution plus the \(207\)-dimensional right kernel.

To give reproducible coordinates to this affine family, we use the fixed coefficient
ordering of \cref{eq:orbit-master-ansatz} and put the matrix in reduced row-echelon form
(RREF).  Let \(I_{\rm piv}\) and \(I_{\rm free}\) be the ordered index sets of pivot and
free coordinates, with \(|I_{\rm piv}|=920\) and \(|I_{\rm free}|=207\).  After ordering
the columns by \((I_{\rm piv},I_{\rm free})\), an invertible row-operation matrix \(U\)
gives
\begin{equation}
  UA_{(I_{\rm piv},I_{\rm free})}
  =\begin{pmatrix}I_{920}&R\\0&0\end{pmatrix},
  \qquad
  Ub(D_s)=\begin{pmatrix}c(D_s)\\0\end{pmatrix}.
  \label{eq:canonical-rref-blocks}
\end{equation}
The chosen RREF particular solution is defined by
\begin{equation}
  (x_{\rm p})_{I_{\rm free}}=0,
  \qquad (x_{\rm p})_{I_{\rm piv}}=c(D_s).
  \label{eq:canonical-particular}
\end{equation}
For each free coordinate \(f_a\in I_{\rm free}\), the corresponding normalized kernel
vector obeys, with \(e_a\) the \(a\)th standard basis vector of \(\F^{207}\),
\begin{equation}
  (k_a)_{f_b}=\delta_{ab},
  \qquad (k_a)_{I_{\rm piv}}=-R e_a,
  \label{eq:canonical-kernel-vector}
\end{equation}
Thus \(K_{I_{\rm free}}=I_{207}\): kernel vector \(k_a\) has value one in its own free
coordinate and zero in the other free coordinates.  This convention fixes the order and
normalization of all \(207\) kernel vectors.

The particular solution is affine in the internal spin dimension and the kernel is
independent of it:
\begin{equation}
  x_{\rm p}(D_s)=x^{(0)}+D_sx^{(1)},
  \qquad K(D_s)=K.
  \label{eq:particular-kernel-ds}
\end{equation}
The coefficient of each kernel vector can be any rational function of \(D_s\), so the
general solution is
\begin{equation}
  x(D_s,\boldsymbol\beta)
  =x^{(0)}+D_sx^{(1)}+\sum_{a=1}^{207}\beta_a(D_s)k_a.
  \label{eq:general-ds-affine-family}
\end{equation}
All vectors \(x^{(0)},x^{(1)}\), and \(k_a\) have rational entries.  Dependence on
\(D_s\) enters through the displayed particular solution and through the freely chosen
functions \(\beta_a(D_s)\).

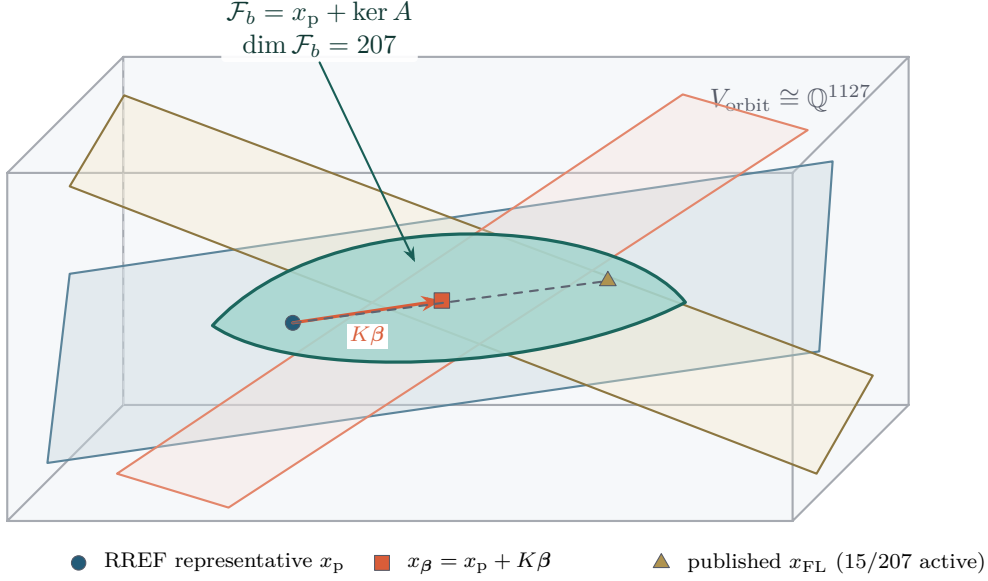
\begin{figure}[t]
\centering
\begin{tikzpicture}[scale=1.18,line cap=round,line join=round,
  every node/.style={font=\small}]
  % Ambient coefficient space, drawn as a perspective wire-frame volume.
  \coordinate (A) at (-.15,-.35); \coordinate (B) at (8.65,-.35);
  \coordinate (C) at (9.95,.95);  \coordinate (D) at (1.15,.95);
  \coordinate (At) at (-.15,3.55); \coordinate (Bt) at (8.65,3.55);
  \coordinate (Ct) at (9.95,4.85); \coordinate (Dt) at (1.15,4.85);
  \path[fill=bcjlight!75] (A)--(B)--(C)--(Ct)--(Dt)--(At)--cycle;
  \draw[bcjgray!55,thick]
    (A)--(B)--(C)--(D)--cycle (At)--(Bt)--(Ct)--(Dt)--cycle
    (A)--(At) (B)--(Bt) (C)--(Ct) (D)--(Dt);
  \draw[bcjgray!25,dashed] (D)--(Dt);
  \node[anchor=north east,bcjgray!90!black] at (9.65,4.65)
    {\(V_{\rm orbit}\cong\Q^{1127}\)};

  % Representative cut-row hyperplanes.
  \path[fill=bcjblue!24,fill opacity=.42,draw=bcjblue!75,thick]
    (.30,.30)--(8.95,1.55)--(9.10,3.68)--(.55,2.42)--cycle;
  \path[fill=bcjgold!35,fill opacity=.38,draw=bcjgold!60!black,thick]
    (.55,3.40)--(8.92,.18)--(9.55,1.28)--(1.16,4.42)--cycle;
  \path[fill=bcjred!20,fill opacity=.34,draw=bcjred!75,thick]
    (1.08,.18)--(7.42,4.43)--(8.82,4.03)--(2.33,-.20)--cycle;

  % Their common affine intersection.
  \path[fill=bcjcyan!42,fill opacity=.88,draw=bcjcyan!65!black,very thick]
    (2.15,1.84).. controls (3.05,1.22) and (5.92,1.30) .. (7.45,2.10)
    .. controls (6.72,3.05) and (3.47,3.28) .. (2.15,1.84)--cycle;
  \node[font=\small,bcjcyan!35!black,align=center,fill=white,fill opacity=.94,
    text opacity=1,rounded corners=1.5pt,inner sep=2pt] (flabel) at (3.35,5.16)
    {\(\mathcal F_b=x_{\rm p}+\ker A\)\\[-1pt]
     \(\dim\mathcal F_b=207\)};
  \draw[-{Stealth[length=2mm]},bcjcyan!60!black,thick]
    (flabel.south)--(4.42,2.55);

  \coordinate (xp) at (3.05,1.87);
  \node[circle,draw=bcjgray!90!black,fill=bcjblue,minimum size=5.6pt,
    inner sep=0pt] at (xp) {};

  \coordinate (xarb) at (4.72,2.12);
  \draw[-{Stealth[length=2.7mm]},very thick,bcjred] (xp)--(xarb)
    node[midway,below=4pt,font=\scriptsize,fill=white,inner sep=1pt]
    {\(K\boldsymbol\beta\)};
  \node[rectangle,draw=bcjgray!90!black,fill=bcjred,minimum size=5.8pt,
    inner sep=0pt] at (xarb) {};

  \coordinate (fl) at (6.58,2.34);
  \draw[dashed,thick,bcjgray] (xp)--(fl);
  \node[regular polygon,regular polygon sides=3,draw=bcjgray!90!black,
    fill=bcjgold!75!black,minimum size=7.2pt,inner sep=0pt] at (fl) {};

  % External key: long semantic labels stay outside the geometric intersection.
  \node[circle,draw=bcjgray!90!black,fill=bcjblue,minimum size=5.4pt,
    inner sep=0pt] at (.65,-.82) {};
  \node[anchor=west,font=\scriptsize] at (.82,-.82)
    {RREF representative \(x_{\rm p}\)};
  \node[rectangle,draw=bcjgray!90!black,fill=bcjred,minimum size=5.6pt,
    inner sep=0pt] at (4.05,-.82) {};
  \node[anchor=west,font=\scriptsize] at (4.22,-.82)
    {\(x_{\boldsymbol\beta}=x_{\rm p}+K\boldsymbol\beta\)};
  \node[regular polygon,regular polygon sides=3,draw=bcjgray!90!black,
    fill=bcjgold!75!black,minimum size=7pt,inner sep=0pt] at (7.18,-.82) {};
  \node[anchor=west,font=\scriptsize] at (7.36,-.82)
    {published \(x_{\rm FL}\) (15/207 active)};
\end{tikzpicture}
\caption{A projection of the coefficient-space geometry.  Each cut equation defines a
hyperplane \(H_r\), and their common intersection is the \(207\)-dimensional affine
solution family.  RREF selects the point \(x_{\rm p}\).  The previously published
forward-limit numerator is another point in the same family, reached by adding kernel
vectors.  The drawing is schematic.}
\label{fig:affine-geometry}
\end{figure}

\subsection{Exact rational realization and the published representative}
\label{sec:known-representative}

The coefficient vectors in \cref{eq:particular-kernel-ds} are rational and have a common
ordering at every value of \(D_s\).  Their exact status is summarized by the following
statement.

\begin{proposition}[Exact rational spin-dimension lift]
The two coefficient vectors \(x^{(0)},x^{(1)}\) defining the particular solution and the ordered kernel
\(k_1,\ldots,k_{207}\) are uniquely reconstructed within the specified rational bounds.
They solve all reconstruction equations exactly, and the same kernel vectors apply at
every value of \(D_s\).
\end{proposition}

Independent cut equations then test the whole affine family: the particular-solution
residual and all \(207\) kernel residuals in
\cref{eq:sealed-validation-equations} vanish.  The finite-field calculation, rational
reconstruction bounds, and validation procedure are described in
\cref{app:exact-arithmetic,app:pseudocode}; the calculation-record index and its
v0.1/v0.2 distribution boundary are given in \cref{app:evidence-ledger}.

The same coordinates locate a previously published formula inside the solution family.
That formula is expressed first in the
\(72\)-term tensor basis used to write it, then in the \(10{,}724\)-dimensional standard tensor
basis, and finally in the \(1127\) symmetry-reduced coordinates used here.  Solving for
its position relative to \((x_{\rm p},K)\) gives the exact relation
\begin{equation}
  x_{\rm FL}(D_s)-x_{\rm p}(D_s)
  =K\bigl(\boldsymbol\alpha_0+D_s\boldsymbol\alpha_1\bigr).
  \label{eq:known-coordinate-bridge}
\end{equation}
Exactly \(15\) of the \(207\) pairs
\((\alpha_{0,a},\alpha_{1,a})\) are nonzero.  Nine have
\(\alpha_{1,a}\ne0\), so their coefficients depend on \(D_s\).

\begin{corollary}[The published representative lies in the solution family]
The published forward-limit five-point numerator is an exact member of the
reconstructed affine family.
\end{corollary}

The published formula is therefore a compact analytic representative of the reconstructed
fiber.  Its difference from the RREF representative lies entirely in the span of the
ordered kernel matrix \(K\).  The fiber dimension is \(207\) by
\cref{eq:rank-nullity-theorem}.  The next step determines how many of these deformations
remain distinct under the stated amputated observable by evaluating its quotient map on
the full ordered kernel.

\subsection{Observable differences and closedness modulo the readout}
\label{sec:stage-s}

The specified map \(\cS\) first projects the graph expansion onto the color-ring
coefficient \(R_{12345}\).  All quotient statements in this subsection refer to this fixed
color projection.  To compare two points in
\(\mathcal F_b=x_{\rm p}+\ker A\), let \(\cQ_S\) be the vector space of canonical
scalar-source records after the
specified LSZ projection and dimensional-regularization identities.  Color projection,
common routing, graph expansion, coefficient collection, LSZ projection, and passage to
this source
space are all linear.  Their composite therefore defines an ambient readout
\begin{equation}
  \cS:V\longrightarrow\cQ_S.
  \label{eq:stage-s-map}
\end{equation}
Every graph is first written in the same loop-momentum routing, so terms with identical
propagators and numerator powers can be compared.  The LSZ projector removes graphs with
a bubble on an external line.  The remaining expressions are collected as coefficients
of scalar integral sources, specified by the powers of the five propagators; negative
powers represent irreducible numerator factors.  Sources with no physical scale are
then zero in the codomain.  The exact implementation, compact v0.1 certificate, and
v0.2 production-record boundary for this map are described in
\cref{app:pseudocode,sec:reproducibility,app:evidence-ledger}.

The readout defines an equivalence relation on the solution fiber by acting on
differences:
\begin{equation}
  x\sim_{\cS}y
  \quad\Longleftrightarrow\quad
  \cS(x-y)=0,
  \qquad x,y\in\mathcal F_b.
  \label{eq:observable-equivalence-relation}
\end{equation}
The deformation directions that remain visible to this readout form the quotient
\begin{equation}
  \mathrm H^{\cS}(A)
  :=\frac{\mathrm{Def}(A)}{\mathrm{Def}(A)\cap\ker\cS}
  =\frac{\ker A}{\ker A\cap\ker\cS}
  \cong\operatorname{im}\!\left(\cS|_{\ker A}\right).
  \label{eq:observable-deformation-space}
\end{equation}
The final isomorphism is the natural one supplied by the first isomorphism theorem.

\begin{definition}[Closed inverse problem modulo \(\cS\)]
\label{def:closed-inverse-problem-mod-s}
The inverse problem \((A,b,\cS)\) is \emph{closed modulo \(\cS\)} when
\begin{equation}
  \operatorname{ob}_A(b)=0,
  \qquad
  \mathrm H^{\cS}(A)=0.
  \label{eq:closed-inverse-problem-conditions}
\end{equation}
The first condition gives existence of a numerator satisfying the cuts.  The second says
that the cuts leave no deformation visible to the stated readout.
\end{definition}

\begin{remark}[Algebraic closure and full validation]
\label{rem:algebraic-versus-charter-closure}
\Cref{def:closed-inverse-problem-mod-s} gives the algebraic core of closure.  For the
five-point result, it is supplemented by complete coverage of the selected cut
topologies, zero held-out residuals for the particular solution and ordered kernel,
successful tests in \(\cT\), and reproducible execution.
\end{remark}

\begin{remark}[The readout acts on differences of solutions]
\label{rem:s-equivalence-on-differences}
The relation \(\sim_{\cS}\) compares two solutions through \(\cS(x-y)\).  The
reconstruction equation remains \(Ax=b\), and the common physical readout is
\(\cS(x_{\rm p})\).  Appending this value gives
\begin{equation}
  \begin{bmatrix}A\\ \cS\end{bmatrix}x
  =\begin{bmatrix}b\\ \cS(x_{\rm p})\end{bmatrix},
  \qquad
  \ker\!\begin{bmatrix}A\\ \cS\end{bmatrix}
  =\ker A\cap\ker\cS,
  \label{eq:correct-readout-stacked-system}
\end{equation}
When \(\ker A\subseteq\ker\cS\), the stacked kernel equals \(\ker A\).  The appended
readout therefore records the common observable value of all \(207\) cut-preserving
directions and leaves the choice of representative free.
\end{remark}

\begin{proposition}[Observable quotient and factorization through cut data]
\label{prop:observable-factorization-through-cuts}
Suppose \(b\in\operatorname{im}A\).  Then
\(\mathcal F_b/{\sim_{\cS}}\) is an affine torsor modeled on
\(\mathrm H^{\cS}(A)\).  Moreover, the following statements are equivalent:
\begin{enumerate}[label=(\roman*)]
\item \(\mathrm H^{\cS}(A)=0\);
\item \(\ker A\subseteq\ker\cS\);
\item there is a unique linear map
  \(\overline{\cS}:\operatorname{im}A\to\cQ_S\) such that
  \begin{equation}
    \cS=\overline{\cS}\circ A_{\rm im},
    \qquad
    A_{\rm im}:V\twoheadrightarrow\operatorname{im}A,
    \label{eq:observable-factorization-through-image}
  \end{equation}
  where \(A_{\rm im}\) is \(A\) with its codomain restricted to its image.
\end{enumerate}
When these conditions hold,
\begin{equation}
  \left|\mathcal F_b/{\sim_{\cS}}\right|=1,
  \qquad
  \cS(x)=\overline{\cS}(b)
  \quad\text{for every }x\in\mathcal F_b.
  \label{eq:cut-determined-observable-readout}
\end{equation}
\end{proposition}

\begin{proof}
Translation by \(\mathrm{Def}(A)=\ker A\) acts freely and transitively on
\(\mathcal F_b\).  Under \(\sim_{\cS}\), the translations that become trivial are
exactly \(\ker A\cap\ker\cS\).  The quotient fiber is therefore a torsor for the
quotient in \cref{eq:observable-deformation-space}.

The equivalence of (i) and (ii) follows directly from that quotient.  If (ii) holds,
define
\(\overline{\cS}(Ax):=\cS(x)\).  If \(Ax=Ay\), then
\(x-y\in\ker A\subseteq\ker\cS\), so this definition is independent of the chosen
preimage.  It is linear and unique because \(A_{\rm im}\) is surjective, proving (iii).
Conversely, (iii) gives \(\cS(k)=\overline{\cS}(Ak)=0\) for every \(k\in\ker A\).
Finally, if \(Ax=b\), the factorization gives
\(\cS(x)=\overline{\cS}(b)\), and the quotient fiber has one element.
\end{proof}

The map \(\overline{\cS}\) in
\cref{eq:observable-factorization-through-image} has the natural domain
\(\operatorname{im}A\).  Consistent data lie in this image.  Values on a complementary
subspace of \(W\) are undetermined when \(\operatorname{coker}A\) is nonzero and play no
role in the readout.

The LSZ part of the map is an explicit projector.  At this post-assembly stage it acts on
a \(56\)-object LSZ/source registry, not on the complete \(297\)-graph cubic registry used
to build the numerator ansatz.  This post-assembly registry contains \(31\) regular
objects and \(25\) formal objects with an external-line bubble, and
\begin{equation}
  P_{\rm LSZ}^2=P_{\rm LSZ},
  \qquad \rank P_{\rm LSZ}=31,
  \qquad \dim\ker P_{\rm LSZ}=25.
  \label{eq:lsz-projector}
\end{equation}
The LSZ projector removes external-line bubbles before the integral identities are
applied.  It acts before external helicities are chosen, so its action on the numerator
family is helicity independent within this graph space.  A different renormalized-field
prescription supplies the corresponding LSZ and \(\delta Z\) contribution at this step.

The remaining vanishing sources have no scale after the prescribed loop-momentum shift.
Typical identities are
\begin{equation}
  \int d^D\ell\,(\ell^2)^r=0,
  \qquad
  \int\frac{d^D\ell\,P(\ell)}{(\ell+\Delta)^2}=0
  \label{eq:scaleless-examples}
\end{equation}
when the shifted integral contains no nonzero external invariant.  Canonical routing is
therefore part of the definition of \(\cS\): it makes the absence of a scale a property
of a stated scalar source, with graph-dependent momentum labels converted to the
canonical route.

\subsection{The exact kernel certificate and observable quotient}

Applying \(\cS\) to the graph expansion of the full kernel gives \(646\) distinct
coefficient conditions.  They divide as
\begin{equation}
  646=350_{\rm coefficient\ identity}+296_{\rm route\ scaleless}.
  \label{eq:stage-s-row-split}
\end{equation}
For each of the first \(350\) conditions \(r\), each kernel direction \(a\), and each of
the \(243\) tensor monomials \(m\), exact collection of the sparse polynomial terms gives
an integer coefficient
\begin{equation}
  c_{ram}=0,
  \qquad
  r=1,\ldots,350,\qquad
  a=1,\ldots,207,\qquad
  m=1,\ldots,243.
  \label{eq:generic-component-zero}
\end{equation}
These coefficient identities hold before integration, independently of the
scaleless-integral rules.

The other \(296\) sources vanish after their canonical loop-momentum shifts:
\begin{equation}
\begin{array}{c|r}
\text{reason the integral is scaleless}&\text{count}\\ \hline
\text{polynomial with no loop denominator}&21\\
\text{massless tadpole after a loop shift}&175\\
\text{adjacent massless external bubble after a loop shift}&100
\end{array}
\label{eq:scaleless-class-counts}
\end{equation}
The three classes exhaust the \(296\) sources.  Thus each condition is removed either by
an exact coefficient identity or by a stated dimensional-regularization identity; the
external-line graphs are removed independently by \(P_{\rm LSZ}\).

The sparse term counts, source-condition summary, and versioned distribution of the
row-level records are given in
\cref{app:pseudocode,sec:reproducibility,app:evidence-ledger}.

The preceding cancellations apply to every vector in the complete ordered kernel.

\begin{theorem}[Observable equivalence of all directions left free by the cuts]
For every ordered kernel vector \(k_a\) of the reconstructed five-point affine family,
\begin{equation}
  \cS(k_a)=0,
  \qquad a=1,\ldots,207.
  \label{eq:stage-s-direction-theorem}
\end{equation}
Equivalently,
\begin{equation}
  \ker A\subseteq\ker\cS,
  \qquad
  \mathrm H^{\cS}(A)=0.
  \label{eq:kernel-in-readout-kernel}
\end{equation}
\end{theorem}

\begin{proof}
The normalized vectors \(k_1,\ldots,k_{207}\) form a basis of \(\ker A\).  For each
basis vector, the readout produces the \(646\) conditions in
\cref{eq:stage-s-row-split}.  The first \(350\) vanish coefficient by coefficient by
\cref{eq:generic-component-zero}.  The remaining \(296\) vanish in \(\cQ_S\) by the
exhaustive source classification in \cref{eq:scaleless-class-counts}, while the
external-line contributions are removed separately by the projector in
\cref{eq:lsz-projector}.  Hence \(\cS(k_a)=0\) for every basis vector.  Linearity gives
\(\ker A\subseteq\ker\cS\), and
\cref{eq:observable-deformation-space} then gives \(\mathrm H^{\cS}(A)=0\).
\end{proof}

\begin{corollary}[The inverse problem is closed modulo \(\cS\)]
\label{cor:five-point-closed-mod-s}
The reconstructed five-point inverse problem is closed modulo \(\cS\), and
\begin{equation}
  \dim_{\F}\mathcal F_b=207,
  \qquad
  \left|\mathcal F_b/{\sim_{\cS}}\right|=1.
  \label{eq:fiber-dimension-versus-quotient-size}
\end{equation}
For any coefficient functions
\(\boldsymbol\beta(D_s)\in\Q(D_s)^{207}\),
\begin{equation}
  \cS\!\left(x_{\rm p}(D_s)+K\boldsymbol\beta(D_s)\right)
  =\cS\!\left(x_{\rm p}(D_s)\right)
  =\overline{\cS}(b).
  \label{eq:stage-s-family-equivalence}
\end{equation}
In particular, the published representative of
\cref{eq:known-coordinate-bridge} and the reduced-row-echelon representative define the
same observable-equivalence class.
\end{corollary}

\begin{proof}
The exact cut solution gives \(\operatorname{ob}_A(b)=0\) and
\(\dim_{\F}\mathcal F_b=207\).  The theorem gives
\(\mathrm H^{\cS}(A)=0\).  Thus the two closedness conditions in
\cref{def:closed-inverse-problem-mod-s} hold, and
\cref{prop:observable-factorization-through-cuts} gives the quotient size and the
cut-determined readout.  The coordinate bridge in \cref{eq:known-coordinate-bridge}
uses the same ordered kernel matrix \(K\), so it also places the published representative
in this unique \(\cS\)-equivalence class.
\end{proof}

\begin{figure}[H]
\centering
\begin{tikzpicture}[x=1.10cm,y=1cm,line cap=round,line join=round]
  % The reconstructed affine leaf.
  \path[fill=bcjcyan!12,draw=bcjcyan!75!black,very thick]
    (.15,.55).. controls (1.55,.08) and (4.20,.42) .. (5.65,1.08)
    .. controls (5.10,2.58) and (2.15,3.18) .. (.72,2.30)--cycle;
  \draw[bcjcyan!45,thin]
    (.65,1.02).. controls (2.05,.72) and (4.08,.94) .. (5.22,1.42);
  \draw[bcjcyan!45,thin]
    (.78,1.55).. controls (2.35,1.18) and (4.28,1.45) .. (5.15,1.88);
  \draw[bcjcyan!45,thin]
    (1.02,2.08).. controls (2.50,1.73) and (4.15,1.92) .. (4.82,2.30);
  \node[font=\small,bcjcyan!45!black,fill=white,fill opacity=.90,
    text opacity=1,inner sep=1.5pt] at (2.65,3.18)
    {cut solution family \(\mathcal F_b=x_{\rm p}+\ker A\)};

  \coordinate (p0) at (1.25,1.25);
  \coordinate (pb) at (3.05,1.65);
  \coordinate (pfl) at (4.72,1.72);
  \fill[bcjblue] (p0) circle (2.6pt)
    node[below left=2pt,font=\scriptsize] {\(x_{\rm p}\)};
  \fill[bcjred] (pb) circle (2.6pt)
    node[below=3pt,font=\scriptsize] {\(x_{\rm p}+K\boldsymbol\beta\)};
  \fill[bcjgold!70!black] (pfl) circle (2.7pt)
    node[below right=3pt,font=\scriptsize,fill=white,inner sep=1pt]
    {\(x_{\rm FL}\)};
  \draw[-{Stealth[length=2.2mm]},bcjred,thick] (p0)--(pb)
    node[midway,above=2pt,font=\scriptsize] {\(K\boldsymbol\beta\)};

  % Quotient map collapses the whole leaf to one equivalence class.
  \draw[-{Stealth[length=4mm]},line width=2.2pt,bcjblue]
    (5.75,1.62)--(7.12,1.62)
    node[midway,above=5pt,font=\large] {\(\pi_{\cS}\)};

  \fill[bcjblue!6] (8.55,1.62) ellipse[x radius=1.42,y radius=1.05];
  \draw[bcjblue!70,thick] (8.55,1.62) ellipse[x radius=1.42,y radius=1.05];
  \draw[bcjblue!25] (8.55,1.62) ellipse[x radius=.92,y radius=.66];
  \fill[bcjblue] (8.55,1.62) circle (3.3pt);
  \node[font=\small,bcjblue!80!black,fill=white,inner sep=1.5pt]
    at (8.55,2.92) {\(\mathcal F_b/{\sim_{\cS}}\)};
  \node[below=7pt,font=\small,bcjblue!80!black] at (8.55,1.62)
    {\([x_{\rm p}]_{\cS}\)};

  % Independent proof mechanisms remain visually and logically distinct.
  \draw[very thick,bcjblue] (.65,-.08)--(1.25,-.08);
  \node[anchor=west,font=\footnotesize] at (1.35,-.08)
    {350 exact coefficient identities};
  \draw[very thick,dashed,bcjgold!70!black] (.65,-.43)--(1.25,-.43);
  \node[anchor=west,font=\footnotesize] at (1.35,-.43)
    {296 scaleless integral sources};
  \draw[very thick,dash dot,bcjred] (.65,-.78)--(1.25,-.78);
  \node[anchor=west,font=\footnotesize] at (1.35,-.78)
    {explicit \(P_{\rm LSZ}\)};
\end{tikzpicture}
\caption{The quotient of the reconstructed fiber.  Every direction \(k_a\) lies in
\(\ker\cS\), so the quotient projection \(\pi_{\cS}\) sends the full affine fiber to one
equivalence class.  The three certificate mechanisms are exact coefficient
cancellation, scaleless integration after a loop-momentum shift, and LSZ projection.}
\label{fig:stage-s-flow}
\end{figure}
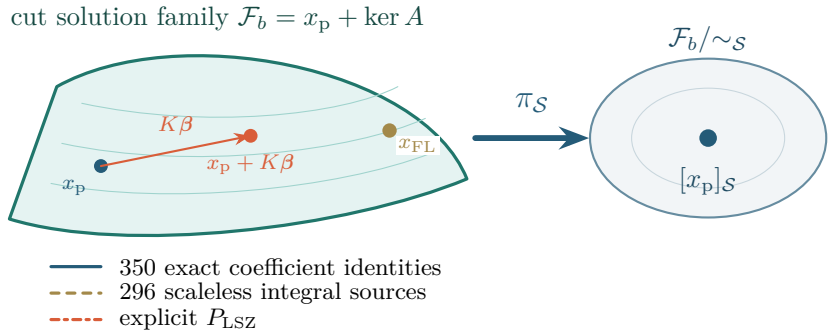

The observable quotient in \cref{eq:fiber-dimension-versus-quotient-size} has the sole
element
\begin{equation}
  [x_{\rm p}]_{\cS}
  =\left[x_{\rm p}+K\boldsymbol\beta(D_s)\right]_{\cS},
  \qquad \boldsymbol\beta(D_s)\in\Q(D_s)^{207}.
  \label{eq:stage-s-quotient-class}
\end{equation}
Thus \(207\) is the affine dimension of the exact numerator fiber, whereas \(1\) is the
number of its \(\cS\)-equivalence classes.  The explicit kernel coordinates remain useful
for choosing sparse, local, or integration-adapted representatives, while the readout is
fixed by the cut data.  The next section evaluates selected representatives through
integral reduction and helicity amplitudes.

\section{Downstream validation and comparison of prescription maps}
\label{sec:validation}

The validation collection \(\cT\) tests consequences of the reconstructed numerator
family using information independent of the cut equations and the quotient \(\cS\).  It assembles graph
and color factors, reduces loop integrals to a common master basis, compares
four-dimensional helicity amplitudes, and evaluates two specified prescription maps.
These readouts are independent of the equations used to determine the numerator fiber.

\subsection{Independent master-integral reduction}

After graph and color projection, every numerator contribution is written in terms of
scalar-source integrals
\begin{equation}
  I_{\bm a}^{(D)}
  =\int\frac{d^D\ell}{i\pi^{D/2}}
  \frac{1}{D_1^{a_1}\cdots D_5^{a_5}},
  \qquad \bm a\in\Z^5.
  \label{eq:source-integrals}
\end{equation}
Negative entries represent irreducible numerator powers.  The regular graph sector
produces \(472\) source exponents.  The observable-equivalence analysis of
\cref{sec:stage-s} sets \(296\) of them to zero, leaving \(176\) sources in the regular
cut-visible quotient.  External-line graphs are removed by the LSZ projector before
this scalar-source representation is compared with an amputated amplitude.

Integration-by-parts identities follow from total derivatives,
\begin{equation}
  0=\int d^D\ell\,
  \frac{\partial}{\partial\ell^\mu}
  \left(\frac{v^\mu(\ell,k)}{D_1^{a_1}\cdots D_5^{a_5}}\right),
  \label{eq:ibp-identity}
\end{equation}
and reduce each source to
\begin{equation}
  I_{\bm a}^{(D)}=\sum_{m=1}^{11}
  R_{\bm a m}(D,s_{ij})\,I_m.
  \label{eq:ibp-reduction}
\end{equation}
The same-\(D\) scalar rank is \(11\).  We use the basis
\begin{equation}
  \mathcal B_{\rm MI}
  =\left\{P_5^{(D+2)},\widehat B_{4,1},\ldots,\widehat B_{4,5},
  U_{2,1},\ldots,U_{2,5}\right\},
  \label{eq:eleven-master-basis}
\end{equation}
where the shifted pentagon is related to the same-\(D\) integrals by an exact dimension
recurrence \cite{Bern:1992em}.  Its coefficient may begin at \(O(\epsilon)\), and its
finite contribution is obtained by combining that valuation with the Laurent order of
\(P_5^{(6-2\epsilon)}\).  Kira and FIRE independently give the same rank, master set, and
rational reduction coefficients for the \(21\) same-dimensional unit-power anchor
integrals.  Kira supplies the full map \(R_{\bm a m}\) for the \(472\)-source table,
including its non-unit-power sources.

There are two natural routes from the graph representation to the master coefficients.
The first expands the graph numerators directly into scalar sources and applies
\cref{eq:ibp-reduction}.  The second reconstructs the \(176\) cut-visible source
coefficients and applies the same master-basis map.  Their relation is
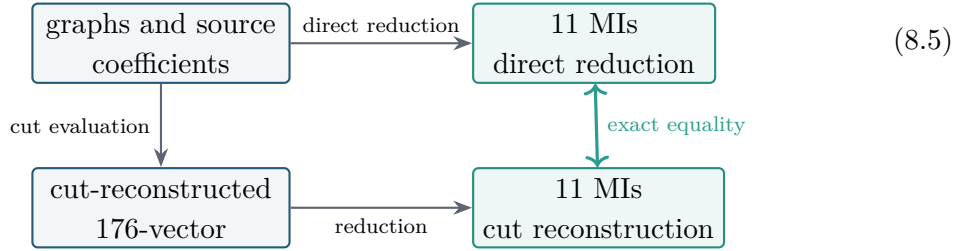
\begin{equation}
\begin{tikzpicture}[baseline=-2pt,node distance=11mm and 24mm]
  \node[flow,minimum width=34mm] (graphs) {graphs and source\\coefficients};
  \node[evidence,right=of graphs,minimum width=32mm] (mia) {11 MIs\\direct reduction};
  \node[flow,below=of graphs,minimum width=34mm] (blind) {cut-reconstructed\\176-vector};
  \node[evidence,right=of blind,minimum width=32mm] (mib) {11 MIs\\cut reconstruction};
  \draw[arrow] (graphs)-- node[above,font=\scriptsize]{direct reduction} (mia);
  \draw[arrow] (graphs)-- node[left,font=\scriptsize]{cut evaluation} (blind);
  \draw[arrow] (blind)-- node[below,font=\scriptsize]{reduction} (mib);
  \draw[<->,very thick,bcjcyan] (mia)-- node[right,font=\scriptsize]{exact equality} (mib);
\end{tikzpicture}
\label{eq:path-ab-square}
\end{equation}
The two routes agree exactly component by component.  The \(296\) cut-invisible sources
vanish independently under the scaling test and the
symbolic integral reduction.  The source tables, reduction programs, and independent
comparison data are indexed in \cref{sec:reproducibility,app:evidence-ledger}.

\subsection{Helicity-amplitude benchmarks and factorization targets}

For four-dimensional external states, the five-point tree MHV amplitude is
\begin{equation}
  A_5^{(0)}(1^-,2^-,3^+,4^+,5^+)
  =i\frac{\langle12\rangle^4}
  {\langle12\rangle\langle23\rangle\langle34\rangle
   \langle45\rangle\langle51\rangle},
  \label{eq:parke-taylor-five}
\end{equation}
while
\begin{equation}
  A_5^{(0)}(+,+,+,+,+)=0,
  \qquad A_5^{(0)}(-,+,+,+,+)=0.
  \label{eq:tree-zero-sectors}
\end{equation}
The all-plus and single-minus one-loop amplitudes therefore isolate rational terms.  The
MHV amplitudes test the infrared poles, scheme conversion, logarithms, and finite
functions.  We compare with the analytic expressions of
\cite{Bern:1993mq,Bern:1994zx,Bern:2005hs}; the dimension-shift relation to self-dual
and supersymmetric amplitudes is described in \cite{Bern:1996ja}.

In the all-plus sector, the expected pole and logarithmic structures vanish and the
shifted-pentagon term has the expected valuation.  For the leading single-trace color
ordering, define the normalized, Parke--Taylor-stripped finite object
\begin{equation}
  F_{+++++}^{\rm PT}
  :=\langle12\rangle\langle23\rangle\langle34\rangle\langle45\rangle\langle51\rangle
  \left[\frac{A_{5;1}^{(1)}(+,+,+,+,+)}{i c_\Gamma}\right]_{\epsilon^0},
  \label{eq:all-plus-pt-stripped-object}
\end{equation}
where \(A_{5;1}^{(1)}\) is the leading single-trace partial amplitude and
\(c_\Gamma=(4\pi)^\epsilon\Gamma(1+\epsilon)\Gamma^2(1-\epsilon)/
[16\pi^2\Gamma(1-2\epsilon)]\).  Write
\((s_1,s_2,s_3,s_4,s_5)=(s_{12},s_{23},s_{34},s_{45},s_{51})\), with cyclic indices
modulo five.  The reconstruction specifies that \(F_{+++++}^{\rm PT}\) is a homogeneous
mass-dimension-four polynomial invariant under cyclic relabeling, with no additional Mandelstam
denominators, in the ordered basis
\begin{equation}
  \left(
  \sum_{i=1}^5s_i^2,
  \sum_{i=1}^5s_is_{i+1},
  \sum_{i=1}^5s_is_{i+2},
  \trfive
  \right).
  \label{eq:all-plus-c5-basis}
\end{equation}
In this four-term function class, its coefficient vector is
\begin{equation}
  (0,\tfrac16,0,-\tfrac16).
  \label{eq:all-plus-reconstructed-coefficients}
\end{equation}
The adjacent-MHV amplitudes reproduce the universal double pole, all five logarithmic
single-pole coefficients, and the finite HV--FDH scheme shift.  Their finite Euclidean
values agree with the analytic amplitudes at independent Euclidean kinematics.

The single-minus comparison satisfies the Ward, cyclic, parity, little-group, pole, and
logarithmic relations.  The nonadjacent-MHV amplitudes agree at independent kinematic
points, helicity families, schemes, and fixed analytic-continuation branches.
Complex-ball evaluation gives overlapping prediction and reference intervals in every
case.  The number of points, the exact rational reconstruction, and the branch choices
are given in \cref{app:exact-arithmetic,sec:reproducibility,app:evidence-ledger}.

The standard one-loop factorization relations provide additional functional validation
targets.  For adjacent collinear momenta, the target relation is
\begin{equation}
 A_n^{(1)}\xrightarrow{a\parallel b}
 \sum_{\lambda}\left[
 \operatorname{Split}^{(0)}_{-\lambda}(a,b)A_{n-1}^{(1)}(P^\lambda)
 +\operatorname{Split}^{(1)}_{-\lambda}(a,b)A_{n-1}^{(0)}(P^\lambda)
 \right],
 \label{eq:one-loop-collinear}
\end{equation}
The corresponding target relation for multiparticle channels includes tree--loop
factorization and the one-loop factorization function of \cite{Bern:1995ix}.  A
positive-helicity soft leg has leading factor
\begin{equation}
  S^{(0)}(a,s^+,b)=
  \frac{\langle ab\rangle}{\langle as\rangle\langle sb\rangle}.
  \label{eq:positive-soft-factor}
\end{equation}
These relations state how the amplitude comparison extends beyond isolated kinematic
points.  The calculation records used here validate the finite-kinematics benchmarks above.
The exact adjacent-collinear and soft-limit comparisons are separate pending validation
steps.

\subsection{Comparing the fixed-routing and forward-limit prescriptions}
\label{sec:prescription-counterexample}

An equality statement is defined relative to its readout map.  Let
\(x_{\rm FR}\) and \(x_{\rm FL}\) denote the fixed-routing and forward-limit points in
the common coefficient space.  The observable quotient establishes
\begin{equation}
  x_{\rm FR}\sim_{\cS} x_{\rm FL}
  \quad\Longleftrightarrow\quad
  \cS(x_{\rm FR}-x_{\rm FL})=0.
  \label{eq:s-equivalence-versus-prescriptions}
\end{equation}
The fixed-routing and forward-limit calculations define distinct downstream maps
\(P_{\rm FR}\) and \(P_{\rm FL}\).  Their output comparison is a separate statement:
\begin{equation}
  \cS(x_{\rm FR}-x_{\rm FL})=0
  \quad\not\Longrightarrow\quad
  P_{\rm FR}(x_{\rm FR})=P_{\rm FL}(x_{\rm FL}).
  \label{eq:map-relative-nonimplication}
\end{equation}
The map \(\cS\) first puts a difference in its canonical routing, applies its LSZ
projector, and sets exact cancellations and scaleless sources to zero.  The maps
\(P_{\rm FR}\) and \(P_{\rm FL}\) retain the raw source assignments associated with
their respective prescriptions and express both outputs in one master-integral basis.
Surface, external-field, and scheme terms must be specified to complete a comparison of
the two maps.  The calculation below identifies a nonzero difference between the raw
maps; it does not construct this completion.

The directly reconstructed fixed-routing representative
\(\A^{\mathrm{FR}}\) and the published forward-limit representative
\(\A^{\mathrm{FL}}\) of \cite{Cao:2025forward} are two points in the affine numerator family.
Three independent calculations of the prescription-specific source and master
coefficients give
\begin{equation}
  C_{\mathrm{FR}}-C_{\mathrm{FL}}=C_{\mathrm{FR-FL}}
  \label{eq:fr-fl-coefficient-identity}
\end{equation}
component by component.

At the preregistered exact kinematic point
\begin{equation}
  (s_{12},s_{23},s_{34},s_{45},s_{51})=(-172,-79,7,-231,320),
  \label{eq:p2-witness-kinematics}
\end{equation}
consider the color-ring ordering \(R_{12345}\), the all-plus helicity assignment, and
\(D_s=3\).  The selected difference in the coefficient multiplying the master integral
\(\widehat B_{4,1}\) is the rational function below.  Here \(D\) is the loop-integration
dimension introduced in \cref{eq:source-integrals} and is kept distinct from the internal
spin dimension \(D_s\).  In this witness \(D_s=3\) is fixed while \(D\) remains symbolic.
\begin{equation}
 C_{\mathrm{FR-FL}}(D)
 =\frac{
 \frac{1591723521}{22400}D^2
 -\frac{667592581}{2800}D
 +\frac{33527363}{200}}
 {D^2-4D+3}
 \ne0.
 \label{eq:p2-nonzero-witness}
\end{equation}
The master integral itself is not evaluated in this coefficient-level witness.
The detailed component maps and the historical filenames used for this comparison are
listed in \cref{app:evidence-ledger}.

Thus the two raw prescription maps have different coefficients in this component of the
stated master basis.  This is the explicit nonzero test of
\cref{eq:prescription-hypothesis}.  The result demonstrates that equivalence is
map-relative: the cut reconstruction has a unique readout modulo \(\cS\), while equality
between \(P_{\rm FR}\) and \(P_{\rm FL}\) requires a separately specified surface,
external-field, and scheme completion.  Reconstruction, quotient, and downstream
prescriptions therefore remain distinct parts of \(\mathcal P_\sigma\).

\section{Summary and outlook}
\label{sec:discussion}

We treat BCJ numerator construction as a finite exact inverse problem.  The specification
\(\Omega\) fixes the physical target, representation space, graph identities,
equivalence map, and validation tests.  At each stage, \(\sigma_t\) selects a finite
tensor basis and cut set.  Compilation yields
\(\mathcal P_{\sigma_t}=(A_t,b_t;\cS_t;\cT_t)\), and its exact ranks and null spaces
guide the next change.

For the one-loop five-gluon problem, the final local parity-even ansatz begins with
\(17{,}824\) tensor monomials.  Kinematic identities reduce them to \(10{,}724\)
independent structures, and signed graph symmetries leave \(1127\) coefficient
variables.  Maximal, box, triple, and double cuts have cumulative ranks
\(547,812,916,920\).  Their complete solution over \(\Q(D_s)\) is the affine fiber
\begin{equation}
  \mathcal F_b
  =\left\{
  x^{(0)}+D_sx^{(1)}+K\boldsymbol\beta(D_s)
  :\boldsymbol\beta(D_s)\in\Q(D_s)^{207}
  \right\},
  \qquad \dim\mathcal F_b=207.
  \label{eq:discussion-main-result}
\end{equation}
The reduced-row-echelon convention selects one point in this fiber, and the ordered
columns of \(K\) give every coefficient deformation invisible to the specified cuts.

The observable map \(\cS\) for the \(R_{12345}\) color-ring projection classifies the
residual freedom.  It sends all \(207\) kernel directions to zero through exact coefficient identities,
fixed-routing scaleless relations, and LSZ amputation of external-line bubbles.
Consequently
\begin{equation}
  \mathrm H^{\cS}(A)
  =\frac{\ker A}{\ker A\cap\ker\cS}=0,
  \qquad
  \left|\mathcal F_b/\!\sim_{\cS}\right|=1.
  \label{eq:discussion-closure}
\end{equation}
In the final calculation, the cut system is consistent, the selected topology coverage
is complete, all held-out residuals vanish, the remaining kernel is classified by
\(\cS\), and the completed finite-kinematics and coefficient-level downstream validations
pass.  The factorization relations in \cref{sec:validation} remain pending validation
targets.  The cuts determine a
\(207\)-dimensional numerator fiber, and \(\cS\) assigns the same readout to every point
in that fiber.

Exact algebra identifies two necessary additions.  The repeated-contact sector
\(Y_i^2N_{ii}\) adds \(300\) basis directions.  Maximal, box, and triple cuts then leave
\(211\) right-kernel directions; the double-cut block has information gain
\(\rank(R_DK_{\le T})=4\), leaving the final \(207\)-dimensional kernel.  The
three-polarization example explains why a new topology was needed: symmetry forced
repeated box and balanced-triple measurements to annihilate the same odd subspace.  In
general the left null space tests whether the ansatz can reproduce the target, the
right kernel guides additional measurements, and \(\cS\) determines whether the
remaining freedom changes the specified readout.

The downstream collection \(\cT\) provides independent evidence.  A direct
scalar-source reduction and a separate route based on cuts give the same coefficients
in a basis of \(11\) master integrals.  All-plus, single-minus, and MHV amplitudes
reproduce the expected rational terms, infrared poles, scheme conversion, analytic
branches, and finite-kinematics values.  The published forward-limit numerator is an exact point
in \(\mathcal F_b\), using \(15\) kernel directions.  Its raw forward-limit
prescription map and the fixed-routing prescription map differ in one stated
master-integral coefficient.  This nonzero component shows that equality is relative
to the chosen map: uniqueness modulo \(\cS\) and equality under two different
integration prescriptions are separate statements.

The public implementation exposes the same diagnostic state and typed operations to a
physicist, a fixed search policy, or an agent.  The five-point calculation uses fixed
rules for the repeated-contact extension and double-cut row
search.  The compact v0.1 release contains the proposal interface, affine vectors, and
selected certificates; the primitive production inputs are assigned to the v0.2
payload described in \cref{sec:reproducibility,app:evidence-ledger}.

The v0.1 compact affine fiber supports further work without repeating the cut
reconstruction.  Its coordinates can be optimized for sparsity, double-copy use, or
integration cost; the difference between prescription maps can be followed through
surface and LSZ completions; and additional observables can be added to \(\cT\).
At higher multiplicity, the same framework begins with a fixed scientific
specification, compiles a
finite formulation, and uses exact obstruction, deformation, quotient, and validation
data to decide whether the formulation is closed or which candidate revision should be
tested next.

\clearpage
\appendix
\section{Homological view of adaptive inverse-problem construction}
\label{app:homological-view}

The linear algebra used in the main text has a compact formulation in homological
algebra.  It organizes four objects from the construction: the obstruction to
solving the cut equations, the deformations left after they are solved, the effect of
adding ansatz terms, and the effect of adding measurements.  In this formulation the
observable quotient classifies the remaining freedom after the cut equations are
solved.

Throughout this appendix, let \(X\) be a finite-dimensional coefficient space, let
\(Y\) be a finite-dimensional data space, and let
\begin{equation}
  A:X\longrightarrow Y
  \label{eq:homological-forward-map}
\end{equation}
be the compiled forward map.  The ground field can be \(\Q\), a finite field used in an
intermediate calculation, or the rational-function field \(\Q(D_s)\) used for the final
reconstruction.  All statements below hold over any field.

\subsection{The two-term complex: obstructions and deformations}

Place the coefficient space in degree \(-1\), the data space in degree \(0\), and the
forward map between them:
\begin{equation}
  \mathsf C_A^\bullet:
  \qquad
  0\longrightarrow X
  \xrightarrow{\ A\ }Y
  \longrightarrow0.
  \label{eq:two-term-forward-complex}
\end{equation}
This is a two-term cochain complex.  Its two cohomology spaces are
\begin{equation}
  H^{-1}(\mathsf C_A^\bullet)=\ker A,
  \qquad
  H^0(\mathsf C_A^\bullet)=\operatorname{coker}A
  =Y/\operatorname{im}A.
  \label{eq:two-term-cohomology}
\end{equation}
Thus the right kernel is the degree-\(-1\) deformation cohomology, while the cokernel is
the degree-\(0\) obstruction cohomology.  This is exactly the pair
\(\mathrm{Def}(A)\) and \(\mathrm{Ob}(A)\) defined in
\cref{def:inverse-obstruction-deformation}.

The target \(b\in Y\) determines a class
\begin{equation}
  [b]_A\in H^0(\mathsf C_A^\bullet).
  \label{eq:target-obstruction-class}
\end{equation}
The equation \(Ax=b\) is solvable precisely when this class vanishes.  If it vanishes
and \(x_{\rm p}\) is one solution, every solution is \(x_{\rm p}+k\) with
\(k\in H^{-1}(\mathsf C_A^\bullet)\).  The solution set is therefore an affine torsor
over \(H^{-1}\): subtraction of two solutions gives a canonical element of the vector
space \(H^{-1}\); choosing a particular solution supplies an origin for this affine set.

The left-null-space test is the dual form of the same statement.  In finite dimensions,
\begin{equation}
  H^0(\mathsf C_A^\bullet)^*
  \cong (Y/\operatorname{im}A)^*
  \cong\ker A^*.
  \label{eq:cokernel-left-null-duality}
\end{equation}
After bases and the standard coordinate pairing are chosen, \(A^*\) is represented by
\(A^T\).  A left-null vector \(\lambda\) evaluates the obstruction class by
\(\lambda^Tb\).  A nonzero value proves that \(b\notin\operatorname{im}A\).  A basis of
the left null space gives coordinates on the dual obstruction space and therefore a
complete consistency test.

\subsection{Adding ansatz directions and the forward correction operator}

Suppose that the current coefficient space \(X\) is enlarged by a candidate ansatz
sector \(U\).  Compiling those new structures gives a column map
\begin{equation}
  C:U\longrightarrow Y,
\end{equation}
and the enlarged forward map is
\begin{equation}
  A_C:X\oplus U\longrightarrow Y,
  \qquad
  A_C(x,u)=Ax+Cu.
  \label{eq:augmented-column-map}
\end{equation}
Let \(\pi_A:Y\to Y/\operatorname{im}A\) be the quotient map.  The induced map
\begin{equation}
  \mathfrak f_C:=\pi_A\circ C:
  U\longrightarrow H^0(\mathsf C_A^\bullet),
  \qquad
  u\longmapsto[Cu]_A,
  \label{eq:forward-correction-operator}
\end{equation}
will be called the \emph{forward correction operator}.  It propagates a proposed
ansatz direction through the compiled physics map and projects the result onto the
cokernel of the current ansatz.  It is therefore the relevant operator for
repairing an ansatz obstruction.  The enlarged equation \(Ax+Cu=b\) is solvable if and
only if
\begin{equation}
  [b]_A\in\operatorname{im}\mathfrak f_C.
  \label{eq:column-repair-solvability}
\end{equation}

This construction is summarized by the exact sequence
\begin{equation}
\begin{split}
  0\longrightarrow\ker A
  &\longrightarrow\ker A_C
  \longrightarrow U
  \xrightarrow{\ \mathfrak f_C\ }
  \operatorname{coker}A\\
  &\longrightarrow\operatorname{coker}A_C
  \longrightarrow0.
\end{split}
  \label{eq:column-extension-exact-sequence}
\end{equation}
The map \(\ker A_C\to U\) forgets the old coefficient \(x\), and the map between the two
cokernels is the natural quotient.  Exactness at \(U\) says that a combination of new
columns can be absorbed by old columns exactly when its obstruction class is zero.
Exactness at \(\operatorname{coker}A\) says that the obstruction classes removed by the
enlargement are precisely those reached by \(\mathfrak f_C\).

The sequence follows directly from the short exact sequence of
two-term complexes
\begin{equation}
  0\longrightarrow\mathsf C_A^\bullet
  \longrightarrow\mathsf C_{A_C}^\bullet
  \longrightarrow (0\longrightarrow U\longrightarrow0)
  \longrightarrow0,
  \label{eq:column-short-exact-complexes}
\end{equation}
where \(U\) occupies degree \(-1\).  Its long exact cohomology sequence is
\cref{eq:column-extension-exact-sequence}; the connecting homomorphism is
\(\mathfrak f_C\).  Equivalently, one can verify exactness element by element from
\cref{eq:augmented-column-map}.

Two dimension identities make the tradeoff visible:
\begin{equation}
\begin{aligned}
  \dim\ker A_C-\dim\ker A
    &=\dim\ker\mathfrak f_C,\\
  \dim\operatorname{coker}A-\dim\operatorname{coker}A_C
    &=\rank\mathfrak f_C.
\end{aligned}
  \label{eq:column-extension-dimensions}
\end{equation}
New columns remove \(\rank\mathfrak f_C\) independent obstruction classes.  Combinations
in \(\ker\mathfrak f_C\) simultaneously introduce deformation directions whose compiled
effect is cancelled by old columns.  Ansatz enlargement therefore changes existence and
uniqueness through two separately measured subspaces.

In left-null coordinates, the forward correction operator is represented by the
overlap matrix \(L^TC\), where the columns of \(L\) span \(\ker A^T\).  For a witness
\(\lambda\), a candidate column \(c\) with \(\lambda^Tc\ne0\) can change the obstruction
detected by \(\lambda\), so the overlap gives a direct score for selecting useful
columns.  Complete repair by the full block \(C\) is certified by
\cref{eq:column-repair-solvability}.

The repeated-contact step in the five-point construction completes the enumerated
two-contact sector.  The square-free choice \(i<j\) omits the nonzero directions
\(Y_i^2N_{ii}\) in the chosen polynomial quotient.  With \(60\) residual structures
for each of five repeated
labels, this gives \(5\times60=300\) omitted directions.  Enlarging the label range to
\(i\leq j\) adds all \(300\) columns.  Once
compiled into the forward map, this new sector is an instance of an added coefficient
space \(U\), and its cut data are governed by the correction operator
\(\mathfrak f_C\).

The repeated-contact calculation establishes that the chosen two-contact tensor space
contains \(300\) directions beyond the square-free ansatz.  After these columns are
added, cut-system consistency is tested separately through the cokernel of the compiled
cut map.  For the completed consistent system, every left-null vector satisfies
\(\lambda^TA=0\) and \(\lambda^Tb=0\).  A nonzero value of \(\lambda^Tb\) would instead
be a certificate that the current system remains inconsistent.

\subsection{Adding measurements and classifying the remaining freedom}

Now keep the coefficient space fixed and add a block of candidate measurements
\(R:X\to Z\).  The stacked forward map is
\begin{equation}
  A_R:X\longrightarrow Y\oplus Z,
  \qquad
  A_Rx=(Ax,Rx).
  \label{eq:augmented-row-map}
\end{equation}
There is a second exact sequence,
\begin{equation}
\begin{split}
  0\longrightarrow\ker A_R
  &\longrightarrow\ker A
  \xrightarrow{\ R|_{\ker A}\ } Z
  \longrightarrow\operatorname{coker}A_R\\
  &\longrightarrow\operatorname{coker}A
  \longrightarrow0.
\end{split}
  \label{eq:row-extension-exact-sequence}
\end{equation}
It is the long exact sequence associated with the projection of two-term complexes
\(\mathsf C_{A_R}^\bullet\to\mathsf C_A^\bullet\); the kernel complex contains \(Z\) in
degree \(0\).  The connecting map from \(H^{-1}(\mathsf C_A^\bullet)=\ker A\) to \(Z\)
is simply the restriction \(R|_{\ker A}\).  Consequently the number of old deformation
directions distinguished by the new measurements is
\begin{equation}
  \dim\ker A-\dim\ker A_R
  =\rank(R|_{\ker A})=\rank(RK),
  \label{eq:row-information-gain-homological}
\end{equation}
where the columns of \(K\) form any basis of \(\ker A\).  This is the information-gain
criterion used for cut selection in the main text.

Rows can also introduce a new compatibility question.  Suppose \(Ax=b\) is solvable,
the desired new data are \(d\in Z\), and \(x_{\rm p}\) is any solution of the old
system.  Define the row residual
\begin{equation}
  \rho_R=d-Rx_{\rm p}.
\end{equation}
Changing \(x_{\rm p}\) by \(k\in\ker A\) changes \(\rho_R\) by an element of
\(\operatorname{im}(R|_{\ker A})\).  Hence the class
\begin{equation}
  [\rho_R]\in
  \operatorname{coker}(R|_{\ker A})
  =Z/\operatorname{im}(R|_{\ker A})
  \label{eq:row-compatibility-obstruction}
\end{equation}
is independent of the chosen old solution.  It vanishes exactly when some point of the
old solution fiber also satisfies \(Rx=d\).  Thus added rows may both remove deformation
directions and reveal incompatible target data.  In the five-point calculation, the
double-cut block obeyed \(\rank(R_DK_{\le T})=4\), reducing the kernel dimension from
\(211\) to \(207\); the corresponding enlarged system remained compatible.

The observable map \(\cS:X\to\cQ_S\) acts on the degree-\(-1\) cohomology after the cut
problem has been solved:
\begin{equation}
  \cS|_{H^{-1}}:\ker A\longrightarrow\cQ_S.
  \label{eq:observable-on-deformation-cohomology}
\end{equation}
Its observable deformation space is
\begin{equation}
  \mathrm H^{\cS}(A)
  =\frac{\ker A}{\ker A\cap\ker\cS}
  \cong\operatorname{im}(\cS|_{\ker A}).
  \label{eq:homological-observable-deformation-space}
\end{equation}
This quotient retains the part of the deformation cohomology that changes the
specified readout.  When it vanishes, \(\cS\) is constant on every affine solution fiber;
equivalently, it factors through the cut image as stated precisely in
\cref{prop:observable-factorization-through-cuts}.  Observable equivalence acts on
differences of solutions through this restricted map.  For the final five-point system,
\(\dim H^{-1}(\mathsf C_A^\bullet)=207\) and
\(\mathrm H^{\cS}(A)=0\): the cut equations leave \(207\) representation directions,
all in one specified observable class.

The roles of the main operators can be read from the following summary:
\begin{center}
\begin{tabular}{@{}p{0.25\textwidth}p{0.66\textwidth}@{}}
\toprule
Object & Meaning in the inverse problem \\
\midrule
\([b]_A\in H^0\) & Ansatz obstruction; evaluated by left-null witnesses. \\
\(H^{-1}=\ker A\) & Cut-invisible deformation space; probed by \(RK\). \\
\(\mathfrak f_C:U\to H^0\) & Forward correction supplied by candidate ansatz columns. \\
\(\cS|_{H^{-1}}\) & Part of the remaining freedom visible to the specified readout. \\
\(\cT\) & Independent validation maps, external to the cohomology of \(A\). \\
\bottomrule
\end{tabular}
\end{center}
These cohomology spaces and exact sequences describe the algebraic core of the adaptive
construction.  Cut coverage, held-out tests, and physical validations complete the
calculation described in the main text.

\section{Conventions and kinematic inputs}
\label{app:conventions}

This appendix collects the conventions used by the theoretical construction and by its
computer implementation.  Fixing them in one place makes the routing maps, finite-field
samples, and amplitude comparisons refer to the same kinematic data.

\subsection{Metric, propagators, and loop-momentum routing}

We use a mostly-minus metric and take all external momenta to be outgoing.  Outputs from
scalar-integral programs that use another metric convention are converted to this one.
The amplitude-level integration measure is
\begin{equation}
  \int_\ell\equiv\int\frac{d^D\ell}{(2\pi)^D},
  \label{eq:measure-convention}
\end{equation}
while integral-reduction programs may use \(d^D\ell/(i\pi^{D/2})\).  The conversion
factor is stored with the integral-basis record.  Propagators have the Feynman prescription
\begin{equation}
  D_i=(\ell+\Delta_i)^2+i0.
  \label{eq:propagator-convention}
\end{equation}
Finite-field reconstruction uses the algebraic propagators without the formal \(+i0\).
The prescription is restored on the physical complex path before numerical scalar
evaluation.

\paragraph{Five-point routing.}

With cyclic labels understood modulo five,
\begin{equation}
  \Delta_1=0,
  \qquad \Delta_i=\sum_{j=1}^{i-1}k_j,
  \qquad \ell_i=\ell+\Delta_i,
  \qquad Y_i=\ell_i^2.
  \label{eq:appendix-route}
\end{equation}
Momentum conservation gives \(\Delta_6=0\).  A cyclic relabeling acts simultaneously on
external labels, master words, routing indices, and the stored graph list.  A reflection also
includes the signs induced by cubic-vertex orientation.

\subsection{External states and dimensional conventions}

The external polarization vectors obey
\begin{equation}
  k_i\cdot\varepsilon_i=0,
  \qquad \varepsilon_i\sim\varepsilon_i+\alpha_i k_i.
  \label{eq:polarization-gauge}
\end{equation}
Linearized field strengths are
\begin{equation}
  f_i^{\mu\nu}=k_i^\mu\varepsilon_i^\nu-\varepsilon_i^\mu k_i^\nu.
  \label{eq:linearized-field-strength}
\end{equation}
The raw basis contains exactly one factor of each external polarization.  Ward tests
replace \(\varepsilon_i\to k_i\) after graph/color assembly or on independently generated
tree targets, depending on the calculation being checked.

\paragraph{Momentum, polarization, and regulator dimensions.}

The implementation records
\begin{equation}
  (D_{\rm mom,ext},D_{\rm pol,ext},D_\ell,D_s)
  \label{eq:dimension-data-model}
\end{equation}
as four distinct entries.  Five massless momenta subject to conservation span at most
four dimensions.  Higher-dimensional samples therefore vary the loop momenta and
transverse polarizations while keeping the external momentum span at four dimensions.

For amplitude comparisons, external momenta and helicities are four dimensional and
\(D=4-2\epsilon\).  We use
\begin{equation}
  u=D_s-2
  \label{eq:u-variable}
\end{equation}
for the number of internal gluon polarization states in integer-dimensional samples.
The four-dimensional-helicity (FDH) and 't~Hooft--Veltman (HV) specializations are
applied with the amplitude convention stated for each comparison.  A formula written as
affine in \(u\) can equivalently be written affine in \(D_s\); the file index states which
variable is used in each recorded result.

\subsection{Spinor-helicity and comparison conventions}

For four-dimensional massless momenta,
\begin{equation}
  k_{\alpha\dot\alpha}=\lambda_\alpha\widetilde\lambda_{\dot\alpha},
  \quad
  \langle ij\rangle=\epsilon^{\alpha\beta}\lambda_{i\alpha}\lambda_{j\beta},
  \quad
  [ij]=\epsilon^{\dot\alpha\dot\beta}
  \widetilde\lambda_{i\dot\alpha}\widetilde\lambda_{j\dot\beta},
  \label{eq:spinor-conventions}
\end{equation}
and \(s_{ij}=\langle ij\rangle[ji]\) in the fixed sign convention.  Exact rational
kinematic families include paired parity sheets.  Targets and predictions are compared
after applying the same color, coupling, loop-measure, \(r_\Gamma\), and
Parke--Taylor-stripping conventions.

\paragraph{Levels of comparison.}

The paper distinguishes the following four levels of comparison:
\begin{description}[leftmargin=27mm,style=nextline]
  \item[fixed-routing pointwise.] The adjacent routing and raw cut state sums are compared
  before applying loop shifts or an integration quotient.
  \item[observable equivalence.] The scalar-source, scaleless-integral, and LSZ
  quotient of \cref{sec:stage-s}.  Some computational filenames use the internal label
  \texttt{Stage-S} for this map.
  \item[amputated benchmark.] The \(31\)-graph image of \(P_{\rm LSZ}\) matched to the
  external analytic benchmark convention.
  \item[integrated amplitude.] A stated regulator, scheme, analytic-continuation branch,
  integral basis, color projection, and external-state convention.
\end{description}

Each comparison in the paper states which of these levels it uses.

\section{Exact arithmetic, rational reconstruction, and validation statistics}
\label{app:exact-arithmetic}

This appendix gives the arithmetic details behind the exact affine solution in
\cref{sec:reconstruction}.  It describes the finite-field solve, the lift to rational
functions of \(D_s\), and the recorded checks on the lifted result.

\subsection{Finite-field solution and \texorpdfstring{\(D_s\)}{Ds} interpolation}

Let \(p\) be an odd prime chosen so that all sampled denominators are nonzero.  We compute
the normalized RREF
\begin{equation}
  U A=\begin{pmatrix}I_r&R\\0&0\end{pmatrix}
  \label{eq:rref-form}
\end{equation}
in the specified coordinate order.  Applying the same row operations to \(b\) exposes
inconsistency if a zero matrix row has a nonzero target.  Otherwise the chosen
reduced-row-echelon particular solution and kernel basis are
\begin{equation}
  x_{\rm p}=\binom{(Ub)_{1:r}}{0},
  \qquad
  K=\binom{-R}{I_{n-r}},
  \label{eq:rref-solution-form}
\end{equation}
after undoing the pivot/free column placement.  The solver output records this
permutation.

An independent implementation using Python and FLINT repeats the solve on exact systems
with the dimensions and sparse layout of the full calculation.  The two implementations
compare the rank, pivots, particular vector, and normalized kernel.

\paragraph{Dependence on \(D_s\).}

The same solve also determines the dependence on the internal-state dimension.  At a
fixed prime and kinematic sample set, solutions are computed at several integer \(D_s\)
values.  Affinity requires
\begin{equation}
  x(D_s+2)-2x(D_s+1)+x(D_s)=0
  \label{eq:ds-second-difference}
\end{equation}
componentwise, after aligning the same free variables.  The calculation record shows that
all second differences vanish and that all cells yield one common intercept, slope, and
kernel.  An independent sample set at \(D_s=12\) tests extrapolation beyond the values
used for interpolation.

\subsection{Rational reconstruction and finite-field validation}

Suppose a rational number \(a/b\) in lowest terms obeys
\begin{equation}
  |a|\le A,
  \qquad 0<b\le B,
  \qquad a\equiv rb\pmod M.
  \label{eq:ratrec-conditions}
\end{equation}
If
\begin{equation}
  M>2AB,
  \label{eq:ratrec-condition-appendix}
\end{equation}
there is at most one such rational number.  Indeed, if \(a/b\) and \(a'/b'\) both map to
\(r\), then \(M\mid ab'-a'b\), while
\begin{equation}
  |ab'-a'b|\le2AB<M,
  \label{eq:ratrec-uniqueness-proof}
\end{equation}
forcing equality.  The extended Euclidean algorithm constructs the candidate, and direct
modular remapping verifies it.

For the affine lift,
\begin{equation}
  M=1100178405778335937,
  \qquad A=B=10^8,
  \qquad M>2AB.
  \label{eq:actual-crt-bounds}
\end{equation}
Every one of the \(235{,}543\) coordinates satisfies the bounds \(A=B=10^8\) stated
above.

\paragraph{Interpretation of finite-field samples.}

If a polynomial \(f\) of total degree \(d\) is nonzero over \(\F_p\), a uniformly random
sample in \(\F_p^N\) is a zero with probability at most \(d/p\).  This
Schwartz--Zippel bound quantifies the generic finite-field checks.  Separately, the
observable-equivalence identity for the kernel coefficients is computed in
characteristic zero by exact integer accumulation.

The sampling rule rejects a point when a required denominator, minor, spinor bracket, or
Gram condition vanishes.  Such a rejection stops the whole batch.  The production
calculation record includes the rejection count and the random-seed namespace.

\subsection{Recorded calculation scale and numerical integral checks}

The following table records the size of the exact reconstruction and the number of
independent checks applied to it.  These counts describe implementation coverage; the
mathematical interpretation of the rank and kernel is given in
\cref{sec:reconstruction,sec:stage-s}.

\begin{center}
\begin{longtable}{@{}>{\raggedright\arraybackslash}p{.35\textwidth}r
  >{\raggedright\arraybackslash}p{.44\textwidth}@{}}
\toprule
quantity & value & interpretation\\ \midrule
\endhead
raw tensor list & 17,824 & tensor expressions before algebraic reduction\\
independent tensor structures & 10,724 & basis after algebraic reduction\\
algebraically duplicate directions & 7,100 & relations removed from the raw list\\
active orbit coordinates & 1,127 & signed-dihedral solve coordinates\\
final cut rank & 920 & cumulative M/B/T/D measurement rank\\
final cut nullity & 207 & complete right kernel in orbit coordinates\\
kinematic sample groups & 12 & independent groups with common pivot columns\\
sample--\(D_s\) combinations & 36 & inputs used for affine reconstruction\\
rational vectors & 209 & intercept, slope, and 207 kernel vectors\\
rational coordinates & 235,543 & entries subjected to bounded reconstruction\\
nonzero rational entries & 15,686 & v0.1 sparse-vector payload\\
validation runs/sample sets & 8/24 & independent primes and disjoint seeds\\
fresh validation groups & 208 & separate samples, including tests at \(D_s=12\)\\
particular-solution equation checks & 624 & all exact zero residual\\
kernel checks & 43,056 & all exact zero residual\\
observable-equivalence records & 207 & one per ordered kernel direction\\
observable-equivalence conditions & 646 & 350 coefficient identities + 296 scaleless\\
\bottomrule
\end{longtable}
\end{center}

\paragraph{Complex scalar integrals.}

The exact coefficient checks are supplemented by numerical tests of the scalar
integrals.  These finite-part checks use complex paths with a stated \(+i0\) orientation.  Exact
polynomial gcd or root isolation verifies that the relevant functions remain nonzero
along each path.  Scalar values are enclosed in arbitrary-precision complex balls, and
the comparison requires the prediction and target balls to overlap under a criterion
fixed before evaluation.  Independent QCDLoop and Wolfram/analytic evaluations are used
for selected scalar checks.  These calculations test analytic continuation and
normalization; the coefficients themselves are reconstructed exactly.

\section{Algorithm outlines}
\label{app:pseudocode}

This appendix translates the linear maps and exact checks defined in the main text into
compact pseudocode.  The listings expose the order of operations and the asserted
identities; implementation-specific data formats are described in
\cref{sec:agent-harness,sec:reproducibility}.

\subsection{Constructing the basis and one cut equation}

The first routine enumerates the specified covariants, removes their exact algebraic
relations, and constructs the signed dihedral orbit coordinates.  Its assertions match
the basis dimensions used in \cref{sec:ansatz-compiler}.

\begin{lstlisting}[language=Python,caption={Construction of the declared finite basis.}]
def build_complete_basis(spec):
    raw = enumerate_covariants(
        polarizations="multilinear",
        mass_dimension=5,
        parity="even",
        contact_degree=range(0, 3),
        allow_repeated_contacts=True,
        loop_rank=spec.layer_loop_ranks)

    relations = exact_kinematic_relations(
        on_shell=True,
        momentum_conservation=True,
        transversality=True,
        exclude_4d_gram_relations=True)

    projection, section = quotient(raw, relations)
    assert projection @ section == identity(10724)
    assert len(raw) == 17824
    assert rank(projection) == 10724
    assert nullity(projection) == 7100

    orbit_map = signed_D5_orbit_projection(section)
    assert orbit_map.coordinate_count == 1127
    independent_backend_audit(raw, relations, orbit_map)
    return BasisArtifact(raw, projection, section, orbit_map)
\end{lstlisting}

\newpage
\paragraph{One cut equation.}

Once the basis artifact has fixed the coordinate maps, a cut row is obtained by routing
the master columns through the graph registry and comparing their residue with an
independently computed tree-amplitude target.

\begin{lstlisting}[language=Python,caption={One cut equation with two tree-amplitude calculations.}]
def compile_cut_row(cut, point, basis, spec):
    masters = basis.build_master_columns(point)
    graphs = free_lie_routing_compiler(masters, spec.graph_registry)
    lhs = residue(cut, propagator_assembly(graphs), point)

    target_T1 = berends_giele_state_sum(cut, point, spec.D_s)
    target_T2 = explicit_vertex_state_sum(cut, point, spec.D_s)
    assert target_T1 == target_T2
    assert ward_audit_T1(cut, point)
    assert ward_audit_T2(cut, point)

    return lhs, target_T1
\end{lstlisting}

The equality of the two target calculations and their Ward tests is checked before the
row is admitted to the linear system.

\subsection{Completing the system with left and right null spaces}

The next routines implement the two null-space diagnostics described in
\cref{sec:cut-bootstrap}.  A left-null obstruction selects ansatz columns that can change
an inconsistent image condition.

\begin{lstlisting}[language=Python,caption={Select columns that act on an inconsistent left-null direction.}]
def respond_to_inconsistency(A, b, candidate_columns):
    assert rank(A) < rank(augment(A, b))
    lam = left_null_obstruction(A, b)
    assert transpose(lam) @ A == 0
    assert transpose(lam) @ b != 0

    scored = [(transpose(lam) @ col, col)
              for col in candidate_columns]
    useful = [col for score, col in scored if score != 0]
    return rank_increment_selection(A, useful)
\end{lstlisting}

\paragraph{Selecting rows with the right kernel.}

For a consistent system, candidate cut families are ranked by their action on
the current right kernel.

\begin{lstlisting}[language=Python,caption={Rank candidate row families on the current kernel.}]
def respond_to_ambiguity(A, b, candidate_row_families):
    particular, K = exact_affine_solve(A, b)
    assert A @ particular == b
    assert A @ K == 0

    scores = {}
    for family in candidate_row_families:
        R = compile_fresh_rows(family)
        scores[family.id] = rank(R @ K)
    return fixed_tiebreak(scores)
\end{lstlisting}

\paragraph{Solving the ordered cut groups.}

The layered solve applies these diagnostics while maximal, box, triple, and double cut
groups are appended in their stored order.

Here \texttt{topology\_coverage\_complete} means that every physical topology has at
least three successful trials.  The frozen \texttt{closure\_criterion\_complete}
requires \(64\) consecutive rows that are dependent and consistent.  The StageSpec does
not prescribe the final rank, nullity, or pivot set.

\begin{lstlisting}[language=Python,caption={Adding cut groups with the frozen topology-coverage and dependent-row stop rule.}]
def layered_solve(calculation_spec):
    A, b = empty_system(columns=1127)
    audit = []

    for group in calculation_spec.round_robin_groups:
        rows = compile_group(group)
        A, b = append_exact_rows(A, b, rows)
        state = exact_solver_state(A, b)
        audit.append((group.id, state.rank, state.pivots,
                      state.dependent_consistent_rows))

        if state.inconsistent:
            return certified_left_obstruction(A, b)
        if topology_coverage_complete(audit) and \
           closure_criterion_complete(audit):
            break

    x_p, K = canonical_rref_affine_solution(A, b)
    assert A @ x_p == b and A @ K == 0
    # Rank and nullity are outputs of the topology-coverage criterion.
    return ModularAffineSolution(A, b, x_p, K, audit)
\end{lstlisting}

\subsection{Rational reconstruction and independent validation}

After the modular ranks and pivot order have stabilized, the modular vectors are aligned
and lifted to a common affine function of \(D_s\).

\begin{lstlisting}[language=Python,caption={Bounded exact lift.}]
def lift_QDs(modular_artifacts, A_bound=10**8, B_bound=10**8):
    align_canonical_pivot_and_free_signatures(modular_artifacts)
    intercept, slope = interpolate_affine_Ds(modular_artifacts)
    K = assert_common_Ds_independent_kernel(modular_artifacts)

    M, residues = chinese_remainder([intercept, slope, *K])
    assert M > 2 * A_bound * B_bound
    vectors = bounded_rational_reconstruct(
        residues, M, A_bound, B_bound)
    assert len(vectors) == 209

    for modular in modular_artifacts:
        assert remap(vectors, modular.prime) == modular.vectors
    return RationalAffineData(vectors, gauge="RREF-free-zero")
\end{lstlisting}

\paragraph{Validation on separate samples.}

The lifted particular solution and every kernel vector are then evaluated on primes and
fresh seed namespaces reserved for validation.

\begin{lstlisting}[language=Python,caption={Validation of the whole affine space.}]
def validate_affine_family(affine_data, validation_spec):
    assert validation_spec.primes_are_disjoint_from_reconstruction
    assert validation_spec.seed_namespaces_are_disjoint

    for profile in validation_spec.profiles:
        A_h, b_h = compile_fresh_group(profile)
        x = affine_data.intercept + profile.D_s * affine_data.slope
        assert A_h @ x == b_h
        for k in affine_data.kernel_vectors:
            assert A_h @ k == 0

    return ExactValidationResult(all_profiles_checked=True)
\end{lstlisting}

\subsection{Reducing the kernel and locating a published representative}

The last pair of routines applies the observable-equivalence map direction by direction
and expresses the published representative in the resulting affine coordinates.

\begin{lstlisting}[language=Python,caption={Direction-by-direction quotient proof.}]
def prove_observable_equivalence_zero(kernel_vectors, equivalence_spec):
    records = []
    for a, k in enumerate(kernel_vectors):
        support = expand_graph_source_support(k)
        regular, formal = apply_LSZ_graph_projector(support)
        conservative, route_scaleless = classify_routes(regular)

        coeffs = sparse_integer_polynomial_accumulate(conservative)
        assert every_component_is_exactly_zero(coeffs)

        proofs = [canonical_scaleless_proof(row)
                  for row in route_scaleless]
        assert all(proof.valid for proof in proofs)
        records.append(DirectionReduction(a, coeffs, proofs, formal))

    assert exact_order_isomorphism(records, kernel_vectors)
    return ObservableEquivalenceResult(records)
\end{lstlisting}

\paragraph{Coordinates of the published representative.}

The coordinate calculation begins after the independent validation of the affine family
and verifies the two kernel equations exactly.

\begin{lstlisting}[language=Python,caption={Exact coordinates of a comparison representative.}]
def coordinates_of_representative(affine_data, representative):
    require(affine_data.independent_validation_passed)
    K = affine_data.kernel_matrix
    x_rep = embed_standard_orbit_coordinates(representative)
    delta0 = x_rep.intercept - affine_data.intercept
    delta1 = x_rep.slope - affine_data.slope
    alpha0 = solve_in_kernel_basis(delta0)
    alpha1 = solve_in_kernel_basis(delta1)
    assert K @ alpha0 == delta0 and K @ alpha1 == delta1
    return CoordinateBridge(alpha0, alpha1)
\end{lstlisting}

\section{Agent-compatible implementation workflow}
\label{app:agent-harness}
\label{sec:agent-harness}

The main text defines the fixed scientific specification \(\Omega\), the mutable
working specification \(\sigma_t\), and the compiled problem
\(\mathcal P_{\sigma_t}=(A_t,b_t;\cS_t;\cT_t)\).  This appendix describes the software
interface that connects an exact diagnostic to the next finite calculation.  The
interface is provider-neutral: a physicist, an explicit search policy, or a
language-model agent can read the same diagnostic state and return the same typed
proposal.  The compiler and evaluator follow one execution path for all three cases.

At stage \(t\), the proposer receives \(\sigma_t\), its exact diagnostic \(d_t\), and
the list of available operations.  It returns a named operation \(\alpha_t\) and its
typed parameters.  The corresponding scientific skill applies that operation, after
which normalization fixes ordering and checks compatibility with \(\Omega\):
\begin{equation}
\begin{aligned}
  d_t&=\operatorname{Evaluate}
  \!\left(\operatorname{Compile}_{\Omega}(\sigma_t)\right),\\
  \widehat\sigma_{t+1}&=F_{\alpha_t}(\sigma_t,d_t),\\
  \sigma_{t+1}&=\operatorname{Normalize}_{\Omega}
  (\widehat\sigma_{t+1}).
\end{aligned}
  \label{eq:normalized-state-transition}
\end{equation}
The compiler output, exact diagnostic, and held-out tests form the diagnostic and
validation record for the normalized stage.

\subsection{Stage specifications and exact diagnostic states}

A machine-readable \texttt{StageSpec} contains the information needed to reconstruct
one ordered finite problem:
\begin{equation}
\begin{split}
  \texttt{StageSpec}=\{&
  \text{scientific-specification and coefficient-field identifiers},\\
  &\text{active basis and contact sectors},\\
  &\text{compiler, graph, and routing-table versions},\\
  &\text{cut, quotient, and held-out-test schedules},\\
  &\text{ordered input bindings},\\
  &\text{primes, seeds, \(D_s\) samples, and resource limits}\}.
\end{split}
  \label{eq:stagespec-fields}
\end{equation}
The normalized record fixes the coordinate order, graph names, row order, arithmetic
profile, and serialization format.  Compilation then constructs
\begin{equation}
  A_C=\mathsf{Cut}_C\,\mathsf{Prop}\,\mathsf{Jac}\,\mathsf{Build},
  \qquad A_Cx=b_C,
  \label{eq:harness-compiled-state}
\end{equation}
together with the maps needed to lift a reduced coefficient vector back to the raw
tensor representation.

The exact evaluator returns
\begin{equation}
  d_t=\left(
  \rank A_t,\ \rank[A_t\mid b_t],\
  L_t,\ L_t^Tb_t,\ K_t,\ \cS_tK_t,\
  \rho^{\cT}_{p,t},\ \rho^{\cT}_{K,t}
  \right).
  \label{eq:exact-diagnostic-state}
\end{equation}
Here \(L_t\) and \(K_t\) are ordered bases of the left and right null spaces.
The two ranks and the product \(L_t^Tb_t\) test consistency; the retained basis \(L_t\)
also scores candidate columns.  The right kernel records every coefficient
direction left unresolved by the active rows.  The products with \(\cS_t\) and
\(\cT_t\) classify these directions and test them on information outside the solve.

A reusable operation, called a skill in the implementation, has the following interface:
\begin{equation}
  \text{preconditions}\longrightarrow
  \text{typed transformation}\longrightarrow
  \text{ordered output}\longrightarrow
  \text{local exact checks}.
  \label{eq:skill-contract}
\end{equation}
The operations used by the reconstruction are summarized below.
\begin{center}
\small
\begin{tabularx}{\textwidth}
  {@{}>{\raggedright\arraybackslash}p{.25\textwidth}
      >{\raggedright\arraybackslash}p{.31\textwidth}X@{}}
\toprule
operation & transformation & returned exact information \\
\midrule
symmetry and Jacobi compilation &
raw tensor basis \(\to\) routed graph coordinates &
signs, orbit order, dimension balances, and symbolic graph identities; \\

kinematic quotient &
raw tensor expressions \(\to\) quotient basis and section &
projection, section, presentation kernel, and identity checks; \\

obstruction-column scoring &
candidate column block \(C\) \(\to\) \(L_t^TC\) &
which obstruction components the block can change; \\

kernel-row scoring &
candidate measurement block \(R\) \(\to\) \(R K_t\) &
rank gain and the directions first visible to the new rows; \\

observable classification &
ordered kernel \(K_t\) \(\to\) \(\cS_tK_t\) &
the observable quotient of the cut-preserving deformation space; \\

held-out validation &
affine family \(\to\) fresh cuts and downstream readouts &
particular-solution residuals, kernel residuals, and physical comparison records. \\
\bottomrule
\end{tabularx}
\end{center}

\subsection{Proposal-policy loop and the two feedback branches}

The proposal interface can be written independently of the policy that selects an
operation:
\begin{lstlisting}[language=Python,
caption={A proposal-policy loop with deterministic exact evaluation.},
label={lst:promotion}]
def exact_constraint_search(initial_spec, proposal_policy):
    spec = normalize(initial_spec)
    history = []

    while True:
        compiled = compile_exact_problem(spec)
        diag = exact_evaluate(compiled)
        history.append(record(spec, compiled, diag))

        if completion_criteria_hold(diag, spec):
            validation = independent_validate(spec, diag.affine_family)
            history.append(record(validation))
            if validation.accepted:
                return diag.affine_family, history

        proposal = proposal_policy.propose(
            spec=spec,
            diagnostic=diag,
            available_operations=spec.skills)
        if proposal is None:
            return NO_ADMISSIBLE_REFINEMENT, history

        operation = select_typed_operation(proposal, spec)
        spec = normalize(operation.apply(spec, proposal, diag))
\end{lstlisting}

There are two basic feedback branches.  For an inconsistent system,
\(\lambda^TA=0\) and \(\lambda^Tb\ne0\) define a left-null obstruction.  Candidate
columns are scored by \(\lambda^TC\), and the revised stage is recompiled with the same
physical targets and conventions.  For a consistent underdetermined system,
candidate rows are scored by \(\rank(RK)\).  A row block with nonzero score measures
directions that were invisible to the previous cut set.

The public repository contains small rational examples of both branches.  Their
recorded operation sequences are
\[
  \texttt{add\_column}\longrightarrow\texttt{accept},
  \qquad
  \texttt{add\_row}\longrightarrow\texttt{accept}.
\]
They exercise the proposal schema, typed operations, exact evaluator, and artifact
format.  In the five-point calculation a repeated-contact sector check activates the
repeated-contact block, while the row diagnostic gives
\(\rank(R_DK_{\le T})=4\) for the selected double-cut block.  The released compact
records bind these deterministic transitions to their compiled stages and exact
diagnostics.

\subsection{Recorded outputs and public release}

Each evaluated stage produces a canonical record
\begin{equation}
  \begin{aligned}
  a=\{\,&\texttt{artifact\_type},\texttt{schema\_version},
  \texttt{stage\_spec\_id},\texttt{operation\_id},\\
  &\texttt{upstream\_ids},\texttt{diagnostic\_payload},
  \texttt{content\_id}\,\}.
  \end{aligned}
  \label{eq:artifact-schema}
\end{equation}
The dependency graph connects a proposal to the compiled system, exact diagnostic,
held-out tests, affine lift, and observable-classification records.

\begin{figure}[H]
\centering
\begin{tikzpicture}[node distance=7mm and 8mm,scale=.92,transform shape]
  \node[flow,minimum width=27mm] (spec) {stage input};
  \node[flow,right=of spec,minimum width=29mm] (proposal) {typed proposal};
  \node[flow,right=of proposal,minimum width=31mm] (compiled) {compiled problem};
  \node[evidence,right=of compiled,minimum width=31mm] (diag)
    {exact diagnostic};

  \node[evidence,below=11mm of diag,minimum width=30mm] (held)
    {held-out tests};
  \node[evidence,left=of held,minimum width=31mm] (lift)
    {affine lift and\\kernel checks};
  \node[evidence,left=of lift,minimum width=31mm] (manifest)
    {compact manifest};

  \draw[arrow] (spec)--(proposal);
  \draw[arrow] (proposal)--(compiled);
  \draw[arrow] (compiled)--(diag);
  \draw[arrow] (diag.south)--(held.north);
  \draw[arrow] (held.west)--(lift.east);
  \draw[arrow] (lift.west)--(manifest.east);
  \draw[-{Stealth[length=2.2mm]},thick,densely dashed,
    draw=bcjorange!75!black]
    (diag.north) to[out=90,in=20]
    node[above,font=\scriptsize]{next candidate revision} (proposal.north);
\end{tikzpicture}
\caption{Logical record graph for the proposal interface.  Exact diagnostics either
start another typed revision or feed the affine lift, kernel classification, held-out
tests, and release manifest.}
\label{fig:artifact-dag}
\end{figure}
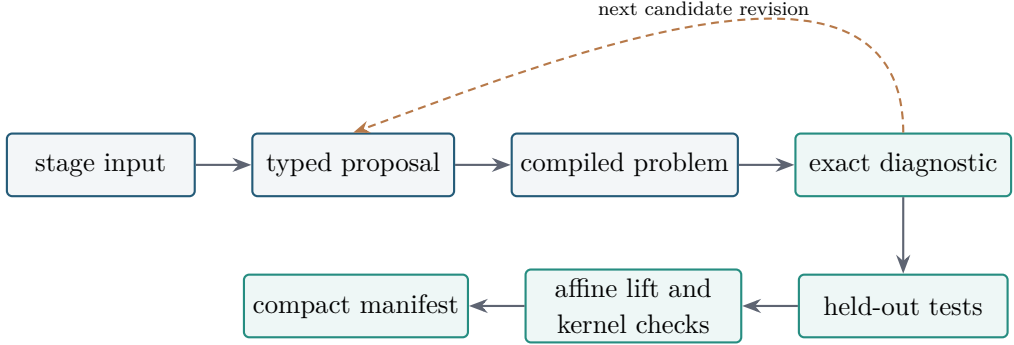

The companion repository is
\url{https://github.com/ybzhang-nxu/AutoBCJFind}.  Its v0.1 release contains the public
schemas, skills, provider-neutral agent profiles, exact-feedback examples, compact
five-point vectors, and selected certificates.  The scientific claim-to-record map is
given in \cref{app:evidence-ledger}.  Current file locations and the boundary between
the compact release and the planned production archive are indexed in the repository's
\href{https://github.com/ybzhang-nxu/AutoBCJFind/blob/main/docs/PAPER_MAP.md}
{\texttt{docs/PAPER\_MAP.md}} and
\href{https://github.com/ybzhang-nxu/AutoBCJFind/blob/main/docs/ARTIFACTS.md}
{\texttt{docs/ARTIFACTS.md}}.

\section{Reproducibility and the v0.1/v0.2 release boundary}
\label{app:reproducibility}
\label{sec:reproducibility}

The companion repository
\url{https://github.com/ybzhang-nxu/AutoBCJFind} currently distributes the compact
v0.1 layer and defines a planned v0.2 production layer.  The compact release contains
the public proposal interface, exact sparse affine
vectors, selected certificates, and the research specifications and runners.  It
supports inspection of the reported mathematical objects and deterministic tests of
the public harness.  The planned approximately \(13\,\mathrm{GB}\) v0.2 production
payload will add the ordered row stores, expanded maps, per-direction source ledgers, and extended
amplitude-comparison records required for replay from primitive production inputs.
The compact claim-to-record map is given in \cref{app:evidence-ledger}.  Current file
locations and the scope of the planned production archive are maintained in the
repository's \texttt{docs/PAPER\_MAP.md} and \texttt{docs/ARTIFACTS.md}.

\subsection{Inputs and two levels of verification}

The production calculation starts from fixed kinematic conventions and an ordered tensor
basis.  It compiles quotient, orbit, cut, observable-equivalence, and validation maps,
solves the resulting finite-field systems, and performs rational reconstruction.  A full
production replay therefore requires the corresponding ordered maps and row stores in
addition to the compact output vectors and certificates.

The v0.1 package provides the following compact verification path:
\begin{enumerate}
  \item check the repository layout, schemas, tests, and compact-file index;
  \item inspect the compact certificates for the basis, reconstruction, validation,
  representative, quotient, integral, and amplitude stages;
  \item validate the formats, ordering metadata, and recorded upstream bindings of
  \(x^{(0)},x^{(1)}\), and the \(207\) ordered sparse kernel vectors;
  \item run the public exact-feedback examples and package preflight.
\end{enumerate}
Rebuilding the rank-\(920\) cut matrix from production rows and recomputing the complete
Stage-\(\cS\) term ledger are the planned v0.2 verification steps.

\paragraph{v0.1 components.}

The compact components are summarized in \cref{tab:core-evidence}.  The repository
indexes listed above give the paths and content identities of individual records.

\begin{table}[H]
\centering
\small
\begin{tabularx}{\textwidth}{@{}
  >{\raggedright\arraybackslash}p{.24\textwidth}
  >{\raggedright\arraybackslash}p{.32\textwidth}
  >{\raggedright\arraybackslash}X@{}}
\toprule
component & data present in v0.1 & supported reader-level check\\ \midrule
implementation & public schemas, skills, agents, exact-harness code, tests, and research
runners & inspect the documented workflow; run the public tests and compact examples\\
basis reduction & two compact basis certificates & inspect the independently recorded
dimensions \(17{,}824\), \(10{,}724\), and \(7{,}100\) and their content bindings\\
affine family & manifests and \(x^{(0)},x^{(1)},k_1,\ldots,k_{207}\) as ordered sparse
rational vectors & check formats, vector count and order, canonical free-coordinate
metadata, and manifest bindings\\
cut reconstruction and validation & compact production and sealed-validation summaries
& inspect the recorded cumulative-rank, pivot/nullity, and zero-residual assertions;
production rows are a v0.2 object\\
representative embedding & coordinate bridge and its compact supporting objects & inspect
the exact family-membership certificate and the \(15\)-direction coordinate record\\
observable-equivalence reduction & one global Stage-\(\cS\) summary certificate & inspect
the recorded all-\(207\) conclusion and its proof-root bindings; per-direction ledgers are
a v0.2 object\\
integral and amplitude checks & selected integral-basis, benchmark, and nonzero-difference
certificates & inspect the compact conclusions and domains; complete reduction and
48-setting records are a v0.2 object\\
\bottomrule
\end{tabularx}
\caption{Compact evidence distributed in v0.1 and the checks it directly supports.  A
certificate may bind a larger production object by hash even when that object itself is
scheduled for the v0.2 archive.}
\label{tab:core-evidence}
\end{table}

\subsection{Repository layout, computing environment, and measured costs}

The v0.1 material used by this paper has the repository-relative layout
\begin{equation}
\begin{array}{ll}
\texttt{release/affine-family/} & \text{compact manifests and ordered sparse vectors},\\
\texttt{release/certificates/} & \text{selected exact certificates and replay summaries},\\
\texttt{specs/, schemas/} & \text{scientific specifications and public data contracts},\\
\texttt{src/, scripts/, tests/} & \text{implementation, package checks, and tests},\\
\texttt{agents/, skills/} & \text{provider-neutral proposal profiles and typed workflows}.
\end{array}
\label{eq:ancillary-layout}
\end{equation}
The v0.2 archive will add the production row stores and expanded maps, with a manifest binding them to
the compact roots listed here.

\paragraph{Software, hardware, and run times.}

The production calculation used Python 3.12, python-flint, SymPy, FORM/TFORM, Kira 3.1,
FIRE 6.5.2, FireFly, Fermat, and the scalar-integral programs recorded during the
calculation.  The main host was a single-NUMA AMD EPYC 9654 system with \(96\) physical
cores, \(125\) GiB memory, and NVMe storage.  The exact reconstruction used the CPU and
memory resources; the GPU remained idle.  v0.1 records this environment summary.  The
external executables and complete production invocation logs are assigned to the planned
v0.2 payload.

Measured wall times were approximately \(444\) seconds for the parallel cut
reconstruction, \(29.5\) minutes for a serial production replay, \(12\) seconds for the
Kira solve of the \(472\)-source configuration, and \(125\) seconds for deterministic
symbolic canonicalization.  These are production timings on the hardware above; the
compact v0.1 package checks have a separate runtime profile.  Intermediate-expression size, rational-function degree,
memory access, serialization, and available memory affect the different stages.

Steps with high memory use were run one at a time.  Worker processes remained active
between related jobs, CPU assignments were fixed, and each exact worker used one
numerical-library thread.  Before a run, the program checked available memory and
recorded command-line and input bindings.  The associated large logs and row-level
records are assigned to the planned v0.2 production payload.

\subsection{Recorded exact assertions and full-replay scope}

The compact records bind the following production assertions:
\begin{align}
  &\rank A=920,\quad AK=0,\quad K_{I_{\rm free}}=I_{207},\nonumber\\
  &A(x^{(0)}+D_sx^{(1)})=b(D_s),\nonumber\\
  &\text{Chinese-remainder and rational-reconstruction uniqueness},\nonumber\\
  &\cS k_a=0\quad(a=1,\ldots,207),\nonumber\\
  &x_{\rm FL}-x_{\rm p}=K(\alpha_0+D_s\alpha_1),\nonumber\\
  &C_{\mathrm{FR-FL}}(D)\ne0.
  \label{eq:minimal-reader-checks}
\end{align}
In v0.1 these equations are represented by compact vectors, manifests, and summary
certificates.  The v0.2 archive will supply the ordered maps, rows, per-direction proof
records, and independent comparison data needed to recompute them from primitive
production inputs.  The two release layers support compact inspection and full
production replay, respectively.

\clearpage
\section{Compact manifest of released calculation records}
\label{app:evidence-ledger}

The public companion repository is
\url{https://github.com/ybzhang-nxu/AutoBCJFind}.  Version v0.1 contains the compact
exact vectors, selected certificates, specifications, schemas, skills, agents, and
public harness described in this paper.  The current claim-to-file index is
\href{https://github.com/ybzhang-nxu/AutoBCJFind/blob/main/docs/PAPER_MAP.md}
{\texttt{docs/PAPER\_MAP.md}}, and the compact-release contents and deferred production
objects are documented in
\href{https://github.com/ybzhang-nxu/AutoBCJFind/blob/main/docs/ARTIFACTS.md}
{\texttt{docs/ARTIFACTS.md}}.  The article presents the scientific map from claims to
record families, while these repository indexes supply the current locations and release
scope.

\begin{table}[H]
\centering
\small
\begin{tabularx}{\textwidth}
  {@{}>{\raggedright\arraybackslash}p{.20\textwidth}
      >{\raggedright\arraybackslash}p{.43\textwidth}X@{}}
\toprule
record family & statement represented by the record & release scope \\
\midrule
basis reduction &
\(17{,}824\) raw structures, a \(10{,}724\)-dimensional kinematic quotient, and a
\(7100\)-dimensional presentation kernel &
two compact certificates in v0.1; expanded projection tables in the production
archive; \\

cut reconstruction &
cumulative ranks \(547,812,916,920\), common pivot/free data, and final nullity \(207\) &
compact solve and serial-replay summaries in v0.1; ordered row stores in the production
archive; \\

affine family &
\(x^{(0)},x^{(1)}\), and the complete ordered set
\(k_1,\ldots,k_{207}\) of sparse rational vectors &
vectors, ordering metadata, and manifests in v0.1; primitive modular records in the
production archive; \\

held-out cut tests &
fresh-prime and fresh-seed residual checks for the particular solution and every
kernel direction &
compact validation certificates in v0.1; complete scheduled task payload in the
production archive; \\

published representative &
exact membership in the affine family, with \(15\) active kernel directions and nine
nonzero \(D_s\) slopes &
coordinate bridge, coordinates, and compact replay summary in v0.1; \\

observable quotient for the \(R_{12345}\) color-ring readout &
\(\cS k_a=0\) for all \(a=1,\ldots,207\), split into coefficient, scaleless-source,
and LSZ branches &
global summary certificate in v0.1; per-direction term ledgers in the production
archive; \\

integral reductions &
canonical master-integral basis and agreement of direct and cut-based reductions &
canonical-basis certificate in v0.1; complete reduction support in the production
archive; \\

amplitude tests &
all-plus, single-minus, and MHV benchmarks, including poles, scheme conversion,
branches, and finite-kinematics comparisons &
selected benchmark certificates in v0.1; extended evaluation records in the production
archive; \\

prescription comparison &
an exact nonzero fixed-routing/forward-limit master-integral coefficient and a
\(48\)-setting numerical comparison &
exact witness in v0.1; full comparison records in the production archive. \\
\bottomrule
\end{tabularx}
\caption{Claim-to-record map for the compact v0.1 release and the planned full production
payload.  Current file locations and release scope are maintained in the repository
indexes cited above.}
\label{tab:compact-record-manifest}
\end{table}

\clearpage
\bibliographystyle{JHEP}
\bibliography{references}

\end{document}